\documentclass[11pt,letterpaper]{article}

\usepackage{preamble}

\begin{document}

\author{
Lorenzo Alvisi\tinyddag \qquad
Allen Clement\tinystar \qquad
Natacha Crooks\tinyddag \qquad
Youer Pu\tinyddag \\
\tinyddag The University of Texas at Austin \qquad \tinystar Google,inc.
 }

\title{Seeing is Believing: \\ A Unified Model for Consistency and Isolation via States}
\date{}

\begin{titlingpage}
    \maketitle
        \begin{abstract}
      This paper introduces a unified model of consistency and
      isolation that minimizes the gap between how these guarantees
      are defined and how they are perceived. Our approach is premised
      on a simple observation: applications view storage systems as
      black-boxes that transition through a series of states, a subset
      of which are observed by applications. For maximum clarity,
      isolation and consistency guarantees should be expressed as
      constraints on those states. Instead, these properties are
      currently expressed as constraints on operation histories that
      are not visible to the application. We show that adopting a
      state-based approach to expressing these guarantees brings forth
      several benefits. First, it makes it easier to focus on the
      anomalies that a given isolation or consistency level allows
      (and that applications must deal with), rather than those that
      it proscribes.  Second, it unifies the often disparate theories
      of isolation and consistency and provides a structure for
      composing these guarantees. We leverage this modularity to apply
      to transactions (independently of the isolation level under
      which they execute) the equivalence between causal consistency
      and session guarantees that Chockler et al.  had proved for
      single operations. Third, it brings clarity to the increasingly
      crowded field of proposed consistency and isolation properties
      by winnowing spurious distinctions: we find that the recently
      proposed parallel snapshot isolation introduced by Sovran et
      al. is in fact a specific implementation of an older guarantee,
      lazy consistency (or PL-2+), introduced by Adya et al.
    \end{abstract}


\end{titlingpage}

\section{Introduction}

\par Large-scale applications such as Facebook, Amadeus, or Twitter
offload the managing of data at scale to replicated and/or distributed
systems. These systems, which often span multiple regions or
continents, must sustain high-throughput, guarantee low-latency, and
remain available across failures.
\par To increase scalability {\em within} a site, databases harness
the power of multicore computing
~\cite{Tu2013silo,diaconu2013hekaton,thomson2012calvin,faleiro2015rethinking,bernstein2015melding}
by deploying increasingly complex concurrency control
algorithms. Faced with the latent scalability bottleneck of
serializability, commercial systems often privilege instead weaker but
more scalable notions of isolation, such as snapshot isolation or read
committed
~\cite{berenson1995ansi,adya99weakconsis,mysqlcluster,postgres}.
Likewise, several recent research efforts focus on improving the
scalability of strong consistency guarantees in distributed storage
systems~\cite{zhang2015tapir,sovran2011walter,kraska2013mdcc,ports2015paxos}.
Modern large-scale distributed systems, however, to increase
scalability {\em across sites} largely renounce strong consistency in
favour of weaker guarantees, from causal consistency to
per-session-only guarantees
~\cite{lloyd13eiger,lloyd2011cops,lakshman2009cassandra,decandia2007dynamo,cooper08pnuts,voldemort,mongodb,riak,cassandraapps,documentdb}.
\par This trend poses an additional burden on the application programmer,
as weaker isolation and consistency guarantees allow for
counter-intuitive application behaviors: relaxing the ordering of
operations yields better performance, but introduces schedules and
anomalies that could not arise if transactions executed atomically and
sequentially.  Consider a bank account with a \$50 
balance and no overdraft allowed. Read-committed allows two
transactions to concurrently withdraw \$45, leaving the
account with a negative balance~\cite{berenson1995ansi}.  Likewise, 
causal consistency ensures that write-read dependencies are enforced, 
but provides no meaningful way to handle write-write
conflicts~\cite{crooks2016tardis}.
\par To mitigate 
programming complexity, many commercial
databases and distributed storage
systems~\cite{cassandra,mongodb,riak,documentdb,azurestorage,cloudbigtable,mysqlcluster,postgres,cloudsql,beanstalk} 
interact with applications through a front-end
that, like a valve, is meant to shield
applications from the complex concurrency
and replication protocols at play. This valve, however, is leaky at best:
a careful understanding of the system that implements
a given isolation or weak consistency level is oftentimes
{\em necessary} to determine which
anomalies the system will admit.

\par Indeed, isolation and consistency levels often assume
features specific to the systems for which they
were first defined---from
the properties of storage (e.g., whether it is single or
multiversioned~\cite{bernstein1981}); to the chosen concurrency
control (e.g., whether it is based on locking or
timestamps~\cite{berenson1995ansi}); or to other system
features (e.g., the existence of a centralized
timestamp~\cite{fekete2005msi}).
These assumptions, furthermore,
are not always explicit:
the claim, in the original ANSI SQL
specification,
that 
serializability
is equivalent to preventing four phenomena~\cite{berenson1995ansi}
only holds 
for
lock-based, single version databases.
Clarity on such matters 
is important: to this day,
multiversioned commercial databases claiming to implement
serializability in fact implement the weaker notion of
snapshot isolation~\cite{fekete2005msi,oracle12c,bailis2013HAT}.
\par We believe that at the root of this complexity is the current practice
of defining consistency and isolation guarantees in terms of the
ordering of low-level operations such as reads and writes, or sends and
receives. This approach has several drawbacks for application
programmers. First, it requires them to reason about
the ordering of operations that they cannot directly
observe.  Second, it makes it easy, as we have seen, to inadvertently
contaminate what should be system-independent guarantees with
system-specific assumptions. Third, by relying on operations that are
only meaningful within one of the layers in the system's stack, it
makes it hard to reason end-to-end about the system's guarantees.
\par To bridge the semantic gap between how isolation and consistency
guarantees are specified and how they are being used, we introduce a
new, unified framework for expressing both consistency and isolation
guarantees that relies exclusively on 
application-observable states rather than on low-level operations. 
The framework is
general: we use it to express most modern consistency and isolation
definitions, and prove that the definitions we obtain are equivalent
to their existing counterparts.
\par In addition to cleanly separating consistency and isolation
guarantees from the implementation of the system to which they apply,
we find that the new framework yields three advantages:
\begin{enumerate}
\item It brings
clarity to the increasingly crowded field of proposed consistency
and isolation properties, winnowing out spurious distinctions. In
particular, we show that the distinction between
PSI~\cite{sovran2011walter,cerone2015framework} and lazy consistency~\cite{adya99weakconsis,adya1997lazyconsistency}
is in fact an artefact of the replication model assumed.
\item It provides a simple framework for composing consistency and
  isolation guarantees.
We leverage its
expressiveness to prove that causal consistency is equivalent to jointly
guaranteeing all four session guarantees (read-your writes, monotonic
reads, writes follow reads, and monotonic writes~\cite{terry94sessionguarantees,chockler2000corba}) independently
of the isolation under which transactions execute.
\item It simplifies reasoning end-to-end about a system's design,
  opening up new opportunities for optimizing its implementation.  In
  particular, we show that PSI can be enforced without
        totally ordering the transactions executed at each site (\changebars{as the original PSI definition instead requires}{as instead
        the original PSI definition requires}), thus making the system less
  prone to have its performance dictated by the rate at which its
        slowest shard \changebars{(data partition)}{} can enforce dependencies.
\end{enumerate}

We provide an extended motivation in Section~\ref{sec:motivation}.
We introduce the model in Section~\ref{sec:model} and use it
to define isolation in Section~\ref{sec:isolation} and consistency
in Section~\ref{sec:consistency}. We highlight practical benefits of
our approach in Section~\ref{sec:contributions}. Finally, we summarize
and conclude in Section~\ref{sec:conclusion}.

\section{A system's case for a new formalism}
\label{sec:motivation}
\par Much of the complexity associated with weakly consistent systems stems
from an intricate three-way semantic gap between how applications are
encouraged to {\em use} these systems, how the guarantees these
systems provide are {\em expressed}, and how they are {\em
implemented}.
\par On the one hand, applications are invited to think of these systems as
black boxes that benevolently hide the complexity involved in
achieving the availability, integrity, and performance guarantees that
applications care about. For example, PaaS (Platform as a Service)
cloud-based storage systems~\cite{azurestorage,cloudbigtable},
databases~\cite{cloudsql} or webservers~\cite{beanstalk} free
applications from having to configure hardware, and let them simply pay
for reserved storage or
throughput~\cite{azurestorage,cloudsql}. Essentially, it is as if
applications were querying or writing to a logically centralized,
failure-free node that will scale as much as one's wallet will allow.
On the other hand, the precise guarantees that these black boxes provide
are generally difficult to pin down, as different systems often give
them implementation-specific twists that can only be understood by
looking inside the box.
\par For example, the exact meaning of {\em session guarantees},
present in Bayou~\cite{terry95bayou}, Corba~\cite{chockler2000corba} and, more
recently, in Pileus~\cite{terry2013pileus} and
DocumentDB~\cite{documentdb}, depends on whether the system to which
they apply implements a total order of write operations across client
sessions.  Consider the execution in Figure~\ref{fig:ser}: does
it satisfy the session guarantee {\em monotonic reads}, which calls
for reads to reflect a monotonically increasing set of writes? The
answer depends on whether the underlying system provides a total order
of write operations across client sessions (like DocumentDB), or just
a partial order based on the order of writes in each session. The
specification of monotonic reads, however, is silent on this issue.
\par Even the classic notion of serializability~\cite{papadimitriou1979serializability} 
 can fall pray to
inconsistencies. In theory, its guarantee is clear: it states that an
interleaved execution of transactions must be equivalent to a serial
schedule. In practice, however, the system interpretations of that
notion can differ. Consider the execution in
Figure~\ref{fig:ser}(a), consisting of two interleaved
write-only transactions: is it serializable? It would appear so, as
it is equivalent to the serial schedule ($T_1,T_2$). The answer,
however, depends on whether the underlying databases allows for writes to
be re-ordered, a feature that is expensive to implement. And indeed, a
majority of the self-declared serializable databases surveyed in
Figure~\ref{fig:ser}(a) do actually reject the schedule. 
In contrast, Figure~\ref{fig:ser}(b) depicts a
non-serializable schedule that exhibits {\em
write-skew}~\cite{fekete2005msi}. Yet, that execution is allowed by
systems that claim to be serializable, such as Oracle 12c. The source
of the confusion is the very definition of serializability. Oracle 12c
uses the definition of the original SQL standard, based on the four
ANSI SQL phenomena that it disallows~\cite{berenson1995ansi}. In
single-versioned systems, preventing these phenomena is indeed
equivalent to serializability; but in a multiversion system (such as
Oracle 12c) it no longer is.
\begin{SCfigure}
\includegraphics[width=0.7\linewidth]{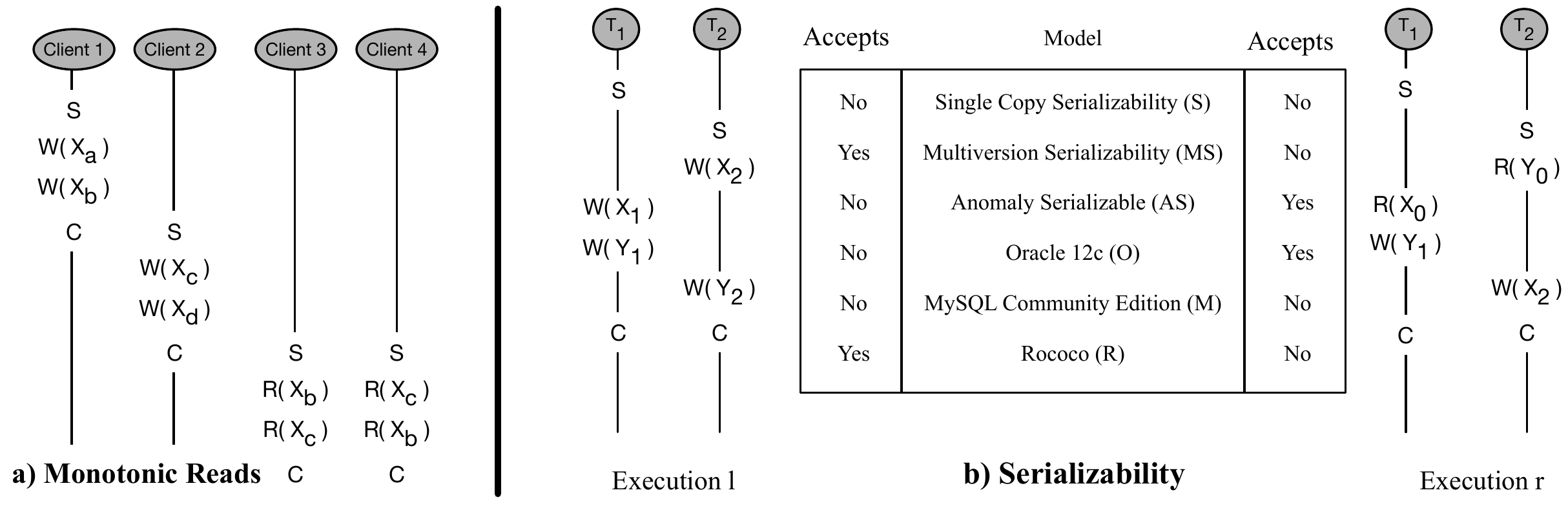}
\caption{a) Monotonic Reads  b) Serializability. Abbrevations refer to: S\cite{papadimitriou1979serializability}, MS\cite{bernstein1983mcc,adya99weakconsis},
    AS\cite{berenson1995ansi} O\cite{oracle12c}, M\cite{mysqlcluster}, R\cite{mu2014extracting}}
\label{fig:ser}
\end{SCfigure}
\par Fundamentally, consistency guarantees are contracts between the storage system and
its clients, specifying a well-defined set of admissible behaviors---
i.e., the set of values that each read is allowed to return. To be
useful, they need to be precise and unchanging. When implicit assumptions about the
implementation of the system are allowed to encroach, however, these
essential attributes can suffer.

\section{Model}
\label{sec:model}

Existing consistency or isolation models are defined as constraints on
the ordering of the read and write operations that the storage system
performs~\cite{adya99weakconsis,bernstein1983mcc,cerone2015framework,chockler2000corba}. Applications, however, cannot directly observe this
ordering. To them, the  storage system is a black box. All
they can observe are the values returned by  the read
operations they issue: they experience the storage system as if it
were going through a sequence of atomic state transitions, of which
they observe a subset.
To make it easier for applications to reason about consistency and
isolation, we adopt the same viewpoint of the
applications that must ultimately use these guarantees.  We propose a
model based on application-observable {\em states} rather than on the
invisible history of the low-level operations performed by the system.

Intuitively, a storage system guarantees a specific isolation or
consistency level if it can produce an {\em execution} (a sequence of
atomic state transitions) that is 1) consistent with the values
observed by applications and 2) valid, in that it satisfies the
guarantees of the desired isolation/consistency level.  In essence,
the values returned by the systems constrain the set of states that
the system can bring forth to demonstrate that it can produce a valid
execution.


More formally, we define a storage system {\em S} with respect to a
set $\mathcal{K}$ of keys, and $\mathcal{V}$ of values; a system
state $s$ is a unique mapping from key to values. \changebars{
For simplicity, we assume that each value is uniquely identifiable,
as is common practice in existing formalisms
~\cite{bernstein1983mcc,adya99weakconsis} and practical systems (ETags in Azure and S3, timestamps in Cassandra). There can thus be no ambiguity, when reading an object, as to which transaction wrote its content.}{For simplicity, we
assume that each value is uniquely identifiable}. In the initial system
state, all keys have value $\bot$.

To unify our treatment of consistency and isolation, we assume that
applications modify the storage system's state using transactions;
we model individual operations as transactions consisting of a single
read or a single write.  A transaction {\em t} is a tuple
$(\Sigma_t, \xrightarrow{to})$, where $\Sigma_t$ is the set of {\em
  operations} in {\em t}, and $\xrightarrow{to}$ is a total order on
$\Sigma_t$. Operation can be either reads or writes. {\em Read}
operation $r(k,v)$ retrieves value {\em v} by reading key {\em k};
{\em write} operation $w(k,v)$ updates $k$ to its new
value $v$.  The {\em read set} of $t$ contains the keys read by $t$:
$\mathcal{R}_t = \{k|r(k,v) \in \Sigma_t\}$. Similarly, the {\em
  write set} of $t$ contains the keys that $t$ updates:
$\mathcal{W}_t = \{k|w(k,v) \in \Sigma_t\}$. For simplicity of
exposition, we assume that a transaction only writes a key once.


Applying a transaction {\em t} to a state {\em s} transitions the 
system to a state $s'$ that is
identical to {\em s} in every key except those written by {\em
  t}.  We refer to {\em s} as the {\em parent state} of 
{\em t}, and refer to the transaction that generated {\em s'} as $t_{s'}$. Formally, 

\begin{definition} $s\xrightarrow[t]{} s' \Rightarrow ([((k,v)\in s' \wedge (k,v) \not
    \in s)$ \changebars{=}{\Rightarrow} $k \in \mathcal{W}_t] \wedge ( w(k,v) \in \Sigma_t)
\Rightarrow (k,v) \in s')$
\label{def:statetrans}
\end{definition}
\nc{check brackets}
We denote the set of keys in  which $s$ and $s'$ differ as \diff{s}{s'}.


An {\em execution}  $e$ for a set of transactions $\mathcal{T}$ is a
totally ordered set defined by the pair
$ (\mathcal{S}_e, \xrightarrow{t\in\mathcal{T}})$, where
$\mathcal{S}_e$ is the set of states generated by applying, starting
from the system's initial state, a permutation of all the transaction
in $\mathcal{T}$. We write $s \xrightarrow[]{*}s'$ (respectively,
$s \xrightarrow[]{+}s'$) to denote a sequence of zero (respectively,
one) or more state transitions from {\em s} to $s'$ in {\em e}.
%
%
%
Note that, while $e$ identifies the state transitions produced by each
transaction $t\in\mathcal{T}$, it does not specify the subset of
states in $\mathcal{S}_e$ that each operation in $t$ can
read from. In general, multiple states in $\mathcal{S}_e$ may be compatible with
the value returned by any given operation. We call this subset the
operation's {\em candidate read states}.
\vspace{2mm}
\begin{definition}\label{defn:readstate}
  Given an execution $e$ for a set of transactions
  $\mathcal{T}$, let $t \in \mathcal{T}$ and let $s_p$ denote $t$'s parent state. The candidate read
    states for a read operation $o =r(k,v) \in \Sigma_t$ is the set of states
  \[\mathcal{RS}_{e}(o) = \{s \in \mathcal{S}_e | 
  s \xrightarrow{*} s_p \wedge\big ( (k,v) \in s \vee (\exists w(k,v) \in \Sigma_t: w(k,v) \xrightarrow{to} r(k,v)) \big)\}\]
\end{definition}
To prevent transactions from {\em reading from the future}, we
restrict the set of valid {\em candidate read states} to those no
later than $s_p$.  Additionally, once $t$
writes $v$ to $k$, we require all subsequent read operations
$o \in \Sigma_t$ to return $v$~\cite{adya99weakconsis}.  

By convention, the candidate read states of a write operation include
all the states $s \in S_e$ that $s \xrightarrow{*} s_p$. It is easy to
prove that the candidate read states of any operation define a
subsequence of contiguous states in the total order that $e$ defines
on $S_e$. We refer to the first state in that sequence as
$s\hspace{-0.1em}f_{o}$, and to the last state as $sl_{o}$.  The
predicate \preread{e}{$\mathcal{T}$} guarantees that such states
exist:
\vspace{2mm}
\begin{definition}
\label{def:preread}
Let $\prereadmath{e}{t} \equiv \forall o \in \Sigma_t:
\mathcal{RS}_{e}(o) \neq \emptyset $. \ Then  $\prereadmath{e}{\mathcal{T}}   \equiv \forall t \in \mathcal{T}: \prereadmath{e}{t}$.
\end{definition}
We say that a state $s$ is {\em complete} for $t$ in $e$ if every
operation in $t$ can read from $s$.  We write:
\vspace{2mm}
\begin{definition}
\(\completemath{e}{t}{s}  \equiv  s \in \bigcap\limits_{o \in \Sigma_t} \mathcal{RS}_{e}(o) \)
\end{definition}
Finally, we introduce the notion of \textit{internal
read consistency}: internal read consistency states that read operations that follow each other in the
transaction order should read from a monotonically increasing
state. We write:
\vspace{2mm}
\begin{definition}
\(
\ircmath{e}{t} \equiv \forall o, \, o' \in \Sigma_t:
o'\xrightarrow{to} o \Rightarrow \lnot (sl_o\xrightarrow{+}s\hspace{-0.1em}f_{o'}) \)
\end{definition}

\section{Isolation}
\label{sec:isolation}
Isolation guarantees specify the valid set of executions for a given
set of transactions $\mathcal{T}$.  The
long-established~\cite{bernstein1981,bernstein1983mcc,adya99weakconsis}
way to accomplish this has been to constrain the history of the
low-level operations that the system is allowed to perform when
processing transactions.  Our new approach eschews this history, which
is invisible to applications, in favor of application-observable
states. Before we introduce our work, we provide some context for it
by summarizing some of the key definitions and results from Adya's
classic, history-based, treatment of isolation~\cite{adya99weakconsis}.

\subsection{A history-based specification of isolation guarantees}

\begin{definition} 
{\em 
A {\em history} $H$ over a set of transactions consists of two parts: $(i)$ a
partial order of events $E$ that reflects the operations (e.g., read,
write, abort, commit) of those transactions; and $(ii)$ a version order,
$<<$, that totally orders committed object versions.}
\end{definition}
\vspace{2mm}
\begin{definition} {\em We consider
  three kinds of {\em direct read/write conflicts:}}
\begin{tabbing}
mm\=mm\=mm\=m\=mm\=mm\=mm\=mm\= \kill
\> {\bf Directly write-depends} $T_i$ {\em writes a version of
$x$ and
$T_j$ writes the next version of} $x$ $(T_i \xrightarrow{ww}
T_j)$ \\
\> {\bf Directly read-depends} $T_i$ {\em writes a version of $x$ that 
$T_j$ then reads  $(T_i
\xrightarrow{wr} T_j)$}  \\
\> {\bf Directly anti-depends} $T_i$ {\em reads a version of $x$, and $T_j$
writes the next version of} $x$ $(T_i \xrightarrow{rw} T_j)$ 
\end{tabbing}
\end{definition}
\vspace{2mm}
\begin{definition} {\em We say that}  {\em $T_j$ {\em start-depends}
    on $T_i$ (denoted as $T_i \xrightarrow{\bstp} T_j$ if $c_i <_t b_j$,
where $c_i$ denotes $T_i$'s commit timestamp and $b_j$ $T_j$'s start timestamp,
i.e., if $T_j$ starts after $T_i$ commits.}
\end{definition}
\vspace{2mm}
\begin{definition}
{\em Each node in the} direct serialization graph DSG(H)  {\em arising from a history} $H$
  {\em corresponds to a committed transaction in} $H$.  {\em Directed edges in}
  DSG(H) {\em correspond to different types of direct conflicts. There is a
  read/write/anti-dependency edge from transaction $T_i$ to
  transaction $T_j$ if $T_j$ directly read/write/antidepends on
  $T_i$.}
\end{definition}
\vspace{2mm}
\begin{definition} {\em The} Started-ordered Serialization Graph
SSG(H) {\em contains the same nodes and edges as} DSG(H) {\em along with start-dependency edges}.
\end{definition}
\vspace{2mm}
\begin{definition} \label{def:phenomena} {\em Adya identifies the following phenomena:}
\begin{description}
\item[\hspace{.3 cm} \gzero: Write Cycles] DSG(H) {\em contains a directed cycle consisting entirely of
  write-dependency edges.}
\item[\hspace{.3 cm} \gonea: Dirty Reads]  H {\em contains an aborted
    transaction $T_i$ and a committed transaction $T_j$ such that
    $T_j$ has read the same object (maybe via a predicate) modified by
    $T_i$.}
\item[\hspace{.3 cm} \goneb: Intermediate Reads] H {\em contains a
    committed transaction $T_j$ that has read a version of object $x$
    written by transaction $T_i$ that was not $T_i$'s final
    modification of $x$.}
\item[\hspace{.3 cm} \gonec: Circular Information Flow] DSG(H)
    {\em contains a directed cycle consisting entirely of
dependency edges.} 
\item[\hspace{.3 cm} \gone:] $\gonea \vee \goneb \vee \gonec$.
\item[\hspace{.3 cm} \gtwo: Anti-dependency Cycles] DSG(H) {\em contains a directed cycle having one or more
  anti-dependency edges.}
\item[\hspace{.3cm} \gs: Single Anti-dependency Cycles] DSG(H) {\em contains a directed cycle with
exactly one anti-dependency edge.}
\item[\hspace{0.3 cm} \gsia: Interference] SSG(H) {\em contains a  read/write-dependency edge from $T_i$ to $T_j$
  without there also  being a start-dependency edge from $T_i$ to $T_j$. } 
\item[\hspace{.3 cm} \gsib: Missed Effects] SSG(H) {\em contains a directed cycle with exactly one
  anti-dependency edge. } 
\item[\hspace{.3 cm} \gsi:] $\gsia \vee \gsib$.
\end{description}
\end{definition}
\vspace{2mm}	
\begin{definition} {\em Adya defines the following isolation levels in
  terms of the phenomena  in Definition~\ref{def:phenomena}}:
\begin{description}
\item[\hspace{.5 cm} Serializability (PL-3)] $ \equiv \lnot \gone \land
  \lnot \gtwo$  \ \ \ \ \ \ \ \ \ \ \ \ \ \ \ {\bf Read Committed (PL-2)}$ \equiv \lnot \gone$
\item[\hspace{.5 cm} Read Uncommitted (PL-1)] $ \equiv \lnot \gzero$   \ \
  \ \ \ \ \ \ \ \ \ \ \ \ \ \ \ 
			{\bf Snapshot Isolation}  $ \equiv \lnot
                        \gone \land \lnot \gsi$
\end{description}
\end{definition}

\subsection{A state-based specification of isolation guarantees}

In our approach based on observable states, isolation guarantees
specify the valid set of executions for a given set of transactions
$\mathcal{T}$ by constraining each transaction $t \in \mathcal{T}$ in
two ways.  First, they limit which states, among those in the
candidate read sets of the operations in $t$, are admissible. Second,
they restrict which states can serve as parent states for $t$.  We
express these constraints by means of a {\em commit test}: for an
execution $e$ to be valid under a given isolation level $\mathcal{I}$,
each transaction \textit{t} in $e$ must satisfy the commit test for
$\mathcal{I}$, written \commit{I}{t}{e}.

\begin{definition}
	A storage system
	satisfies an isolation level $\mathcal{I}  \equiv \exists e:
        \forall t \in T : \commitmath{I}{t}{e}$.
\end{definition}

\begin{table}
\begin{center}
\footnotesize{
\begin{tabular}{|c|c|}
\hline
Serializability & \complete{e}{t}{$s_p$} \\
\hline
Snapshot Isolation &  $\exists s \in S_e.$ \complete{e}{t}{s} $\wedge ($\diff{s}{s_p} $\cap \mathcal{W}_{t} = \emptyset)$ \\ 
\hline
Read Committed  &\preread{e}{$t$}\\ 
\hline
Read Uncommitted & True\\ 
\hline
\end{tabular}
}
\end{center}
\caption{ANSI SQL Commit Tests}
\label{table:isolation}
\end{table}

\noindent Table~\ref{table:isolation} shows the commit tests for the four most
common ANSI SQL isolation levels.  We informally motivate their
rationale below.

\par\noindent{\bf Serializability}.  Serializability requires the values
observed by the operations in each transaction $t$ to be consistent with
those that would have been observed in a sequential execution.  The
commit test enforces this requirement through two complementary
conditions on observable states. First, all operations of $t$ should read
from the same state $s$, thereby ensuring that transactions  never observe the
effects of concurrently running transactions.  Second, $s$ should be
the parent state of $t$, i.e., the state that $t$ transitions
from. 

\par\noindent{\bf Snapshot isolation (SI)}. Like
serializability, SI prevents every transaction $t$ from seeing the
effects of concurrently running transactions. \changebars{
The commit test enforces this requirement by having all 
operations in $t$ read from the same state $s$, where $s$ is
a state produced by a transaction that precedes $t$ in the
execution $e$}{The commit test enforces
this requirement by having all operations in $t$ read from the same
state $s$ produced by a transaction that precedes $t$ in the execution
$e$.} However, SI no longer insists on $s$ being $t$'s parent state $s_p$:
other transactions may commit in between $s$ and $s_p$, whose operations $t$ will
not observe. The commit test only forbids $t$ from modifying any of the keys
that changed value as the system's state progressed from $s$ to
$s_p$. 

\par\noindent{\bf Read committed}. Read committed allows $t$ to see
the effects of concurrent transactions, as long as they are
committed. The commit test therefore no longer constrains all
operations in $t$ to read from the same state; instead, it only
requires them to read from state that precedes $t$ in
the execution $e$.

\par\noindent{\bf Read uncommitted}.  Read uncommitted allows
$t$ to see the effects of concurrent transactions, whether they have
committed or not. The commit test reflects this permissiveness, to the
point of allowing transactions to read arbitrary values. The reason
for this seemingly excessive laxity is that isolation models in
databases consider only committed transactions and are therefore
unable to distinguish between values produced by aborted transactions
and by altogether imaginary writes.  This distinction is not lost in
environments, such as transactional memory, where correctness depends on
providing guarantees such as opacity~\cite{guerraoui2008opacity} for all live
transactions.  We discuss this further in Section~\ref{sec:conclusion}.

Although these tests make no mention of histories, they each admit the
same set of executions of the corresponding  history-based condition
formulated by Adya~\cite{adya99weakconsis}. In~\ref{appendix:ser},~\ref{appendix:si}
and~\ref{appendix:rc}, we prove the following theorems,
respectively:

\begin{theorem}
\label{theorem:ser} Let $\mathcal{I}$ be Serializability (SER). Then $\ \ \exists e:\forall t
  \in \mathcal{T}: \commitmath{SER}{t}{e}  \equiv \lnot \gone \land \lnot
  \gtwo$.
\end{theorem}
\begin{theorem}
\label{theorem:si}Let $\mathcal{I}$ be Snapshot Isolation (SI). Then $\ \ \exists e:\forall t
  \in \mathcal{T}: \commitmath{SI}{t}{e} \equiv \lnot \gone \land \lnot \gsi$
\end{theorem}
\begin{theorem} 
\label{theorem:rc}
Let $\mathcal{I}$ be Read Committed (RC). Then $\ \ \exists e:\forall t
  \in \mathcal{T}: \commitmath{RC}{t}{e}  \equiv  \lnot \gone$
\end{theorem}

\begin{theorem}
\label{theorem:ru} Let $\mathcal{I}$ be Read Uncommitted (RU). Then $\ \ \exists e:\forall t
  \in \mathcal{T}: \commitmath{RU}{t}{e}  \equiv  \lnot \gzero$

\end{theorem}

\subsubsection{Discussion} 

The above theorems establish that a specification of isolation
guarantees based on client-observable states is as expressive as one
based on histories. Adopting a client-centric perspective, however,
has a distinct advantage: it makes it easier for application
programmers to understand the anomalies allowed by weak isolation
levels.

To illustrate this ``intuition gap'', consider the simple banking example of
Figure~\ref{fig:banking}.  Alice and Bob share checking (C) and saving
(S) accounts, each holding \$30.  To avoid the bank's wrath, before performing a withdrawal
they check that the total funds in their accounts allow for it. They
then withdraw the amount from the specified account, using the other
account to eventually cover any overdraft.  Suppose Alice and Bob try
concurrently to each withdraw \$40 from, respectively, their checking
and savings account, and issue transactions $t_{w1}$ and
$t_{w2}$. Figure~\ref{fig:banking}(a) shows an execution under
serializability. Because transactions read from their parent
state, $t_{w2}$ observes $t_{w1}$'s withdrawal and, since the
balance of Bob's accounts is below \$40, aborts.

In contrast, consider the execution under snapshot isolation in
Figure~\ref{fig:banking}(b).  It is legal for both $t_{w1}$ and
$t_{w2}$ to read the same state $s_1$, find that the combined funds in
the two accounts exceed \$40, and, unaware of each other, proceed to
generate an execution whose final state $s_3$ will get Alice and Bob
in trouble with the bank.  This anomaly, commonly referred to as
write-skew, arises because $t_{w2}$ is allowed to read from a state
other than the most recent state.  Defining snapshot isolation in
terms of observable states makes the source of this anomaly obvious,
arguably to a greater degree than the standard history-based
definition, which characterizes snapshot isolation as ``disallowing
all cycles consisting of direct (write-write and write-read)
dependencies and at most a single anti-dependency''.


Though we have focused our discussion on  ANSI SQL isolation levels
that do not consider real-time, our model can straightforwardly be
extended to support strict serializability~\cite{Herlihy90Linearizability} as follows.

Let $\mathcal{O}$ be a time oracle that assigns distinct
\textit{start} and \textit{commit} timestamps (\textit{t.start} and
\textit{t.commit}) to every transaction $t \in \mathcal{T}$.  A
transaction $t_1$ time-precedes $t_2$ (we write $t_1<_s t_2$) if
$t_1.commit < t_2.start$. Strict serializability can then be defined
by adding the following condition to the serializability commit test:
$\forall t' \in \mathcal{T}: t'
<_s t \Rightarrow s_{t'} \xrightarrow{*} s_t$.

\begin{figure}
\center
\includegraphics[width=0.9\linewidth]{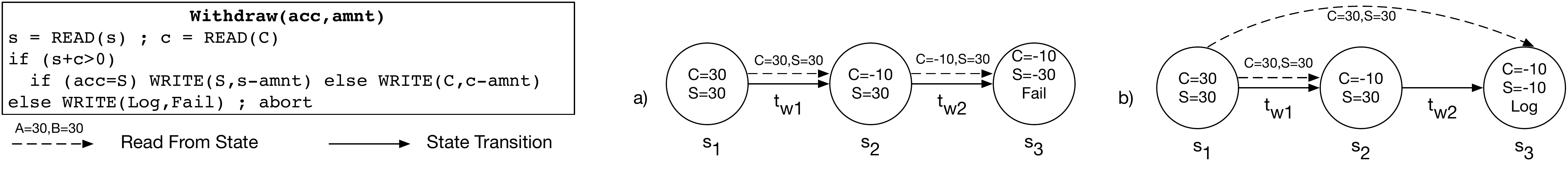}
\caption{Simple Banking Application. Alice and Bob share checking and savings
accounts. 
Withdrawals are allowed as long as the sum of both account is greater
than zero} 
\label{fig:banking}
\end{figure}

\section{Consistency}
\label{sec:consistency}
Though isolation guarantees typically do not regulate how transactions
from a given client should be ordered~\footnote{Strict
serializability is the exception to this rule.}, they tacitly
assume that transactions from the same client will be executed in
client-order, as they naturally would in a centralized or
synchronously replicated storage system.
In weakly consistent systems, where transactions can be asynchronously
replicated between sites, this assumption no longer holds: two
transactions from the same client may be re-ordered if they happen to
be executed on different replicas.  To bring back order, distributed
systems introduce the notion of \textit{sessions}.  Sessions
encapsulate the sequence of transactions performed by a single entity
(a thread, client, or application). Informally, their aim is to
provide each entity with a view of the system consistent with its own
actions;\changebars{}{formally, a session \textit{se} is a tuple
$(T_{se} ,\xrightarrow{se})$ where $\xrightarrow{se}$ is a total order
over the transactions in $\mathcal{T}_{se}$.  The set of all sessions
is denoted by $SE$.}
\changebars{
\begin{definition}
A session \textit{se} is a tuple
$(T_{se} ,\xrightarrow{se})$ where $\xrightarrow{se}$ is a total order
    over the transactions in $\mathcal{T}_{se}$ such that $\mathcal{T}_{se}
    \subseteq \mathcal{T}$.
\end{definition}
}{}

To provide a foundation to a common theory of isolation and
consistency, we define session-based consistency guarantees for
transactions. Session guarantees have traditionally been defined for
operations~\cite{chockler2000corba,terry94sessionguarantees,brzezinski04sessions}:
our definitions can be mapped back by considering single-operation
transactions.

We first introduce the definition of sequential
consistency~\cite{Lamport79How}. Sequential consistency requires read operations within
each transaction observe monotonically increasing states and have
non-empty candidate read sets and, like previously defined isolation
levels, demands that all sessions observe a single execution. Unlike
isolation levels, however, sequential consistency also requires
transactions to take effect in the order specified by their
session. We define the commit test for sequential consistency as
follows:
\[
\commitmath{SC}{e}{t} \equiv \prereadmath{e}{t} \wedge
\ircmath{e}{t} \wedge (\forall se \in SE : \forall t_i
\xrightarrow{se} t_j : (s_{t_{i}} \xrightarrow{+} s_{t_{j}} \wedge
\forall o \in \Sigma_{t_{j}}: s_{t_{i}} \xrightarrow{*} sl_o))
\]

Guaranteeing the existence of a single execution
across all clients is often prohibitively expensive
if sites are geographically distant.
Many systems instead allow clients in different sessions to observe
{\em distinct} executions. Clients consequently perceive the system as
consistent with their own actions, but not necessarily with those of
others.
To this effect, we reformulate the \textit{commit test} into a
\textit{session test}: for an execution $e$ to be valid under a given
session $se$ and session guarantee $SG$, each transaction \textit{t}
in $\mathcal{T}_{se}$ must satisfy the session test for $SG$, written
\session{SG}{se}{t}{e}.

\begin{definition}
	A storage system
	satisfies a session guarantee $SG  \equiv \forall se \in SE: \exists e:
        \forall t \in \mathcal{T}_{se} : \sessionmath{SG}{se}{t}{e}$.
\end{definition}

Intuitively, session tests invert the order between the existential qualifier for execution and the universal quantifier for sessions. Table~\ref{table:sessions} shows the session tests for the most common
session guarantees. We informally motivate their rationale below.

\par \textbf {Read-My-Writes} This session guarantee states that a
client will read from a state that includes any preceding writes {\em
  in its session}. RMW is a fairly weak guarantee: it does not
constrain the order in which writes take effect, nor does it provide
any guarantee on the reads of a client who never writes, \changebars{
or say anything}{nor has
anything to say} about the outcome of reads performed in other sessions
(as it limits the scope of PREREAD only to the transactions of its
session). The session test is simply content with asking for the read
state of every operation in a transaction's session to be after the
commit state of all preceding update transactions in that session.
\par \textbf {Monotonic Reads} Monotonic reads, instead, constrains
a client's reads to observe an increasingly up-to-date state of the database:
this applies to transactions in a session, and to operations within the
transaction (by IRC).
The notion of  ``up-to-date" here may vary by client, as the storage is free
to arrange transactions differently for each session's execution. The only way for the client
to detect an MR violation is hence to read a value three times, reading the initial value, a new value, and the initial value again. Moreover, a client is not guaranteed to see the effects of its
own write: MR allows for clients to read from monotonically increasing but
stale states.
\par \textbf {Monotonic Writes} In contrast, monotonic writes
constrains the ordering of writes within a session: the sequence of
state transitions in each execution must be consistent with the order
of update transactions in {\em every session}.  Unlike MR and RMW,
monotonic writes is a global guarantee, at least when it comes to
update transactions. The PREREAD requirement for read operations,
instead, continues to apply only within each session.
\par \textbf {Writes-Follow-Reads} Like monotonic writes,
writes-follow-reads is a global guarantee, this time covering reads as
well as writes. It  states that, if a
transaction reads from a state \textit{s}, all transactions that
follow in that session must be ordered after \textit{s} in every
execution that a client may observe.
\par \textbf{Causal Consistency}
Finally, causal consistency guarantees that any execution will order transactions
in a causally consistent order: read operations in a session will see monotonically increasing read states, and commit in session order. Likewise, transactions that read from a state \textit{s} will be ordered
after \textit{s} in all sessions. This relationship is transitive: every transaction that reads $s$ (or that follows in the session) will also be ordered
after $s$. 
\begin{table}
\footnotesize{
\begin{center}
\begin{tabular}{|c|c|}
\hline
Read-My-Writes (RMW)& $\prereadmath{e}{\mathcal{T}_{se}} \wedge \forall o \in \Sigma_t: \forall t'\xrightarrow{se} t:\mathcal{W}_{t'} \neq \emptyset \Rightarrow s_{t'} \xrightarrow{*}sl_o$\\ 
\hline
Monotonic Reads (MR) & $\prereadmath{e}{\mathcal{T}_{se}} \wedge \ircmath{e}{t} \wedge \forall o \in \Sigma_t: \forall t'\xrightarrow{se}t: \forall o' \in \Sigma_{t'}: \lnot (sl_o\xrightarrow{+}s\hspace{-0.1em}f_{o'}) $\\
\hline
Monotonic Writes (MW) & $\prereadmath{e}{\mathcal{T}_{se}} \wedge \forall se' \in SE: \forall t_i \xrightarrow{se'} t_j: (\mathcal{W}_{t_i} \neq \emptyset \wedge
\mathcal{W}_{t_j} \neq \emptyset) \Rightarrow s_{t_i} \xrightarrow{+} s_{t_j} $\\
\hline
Writes-Follow-Reads (WFR) & $\prereadmath{e}{\mathcal{T}} \wedge \forall se' \in SE: \forall t_i\xrightarrow{se'}t_j: \forall o_i \in \Sigma_{t_i}:\mathcal{W}_{t_j} \neq \emptyset \Rightarrow sf_{o_i} \xrightarrow{+} s_{t_j} $\\
\hline
Causal Consistency (CC) &  \tabincell{c}{$ \prereadmath{e}{\mathcal{T}} \wedge \ircmath{e}{t} \wedge (\forall o \in \Sigma_t: \forall t' \xrightarrow{se} t:  s_{t'} \xrightarrow{*} sl_o) $\\ $\wedge (\forall se' \in SE: \forall t_i \xrightarrow{se'} t_j: s_{t_i} \xrightarrow{+} s_{t_j})$ } \\
\hline
\end{tabular}
\end{center}
}
\caption{Session Guarantees}
\label{table:sessions}
\end{table}



\section {Practical benefits}
\label{sec:contributions}
In addition to reducing the gap between how isolation and consistency
guarantees are defined and how they are perceived by their users, definitions based
on client-observable states provide further benefits.
\subsection{Economy}
Focusing on client-observable states frees definitions from
implementation-specific assumptions. Removing these artefacts can bring
out similarities and winnow out spurious distinctions in the
increasingly crowded field of isolation and consistency.

Consider, for example, parallel snapshot isolation (PSI), recently
proposed by Sovran et al~\cite{sovran2011walter} and lazy consistency,
introduced by Adya et
al~\cite{adya99weakconsis,adya1997lazyconsistency}.  Both isolation
levels are appealing as they are implementable at scale in
geo-replicated settings. Indeed, PSI aims to offer a scalable
alternative to snapshot isolation by relaxing the order in which
transaction are allowed to commit on geo-replicated sites. The
specification of PSI is given as an abstract specification code that
an implementation must emulate.
\begin{definition}
PSI enforces three main properties:
\begin{itemize}
\item \textit{P1 (Site Snapshot Read): All operations read the most recent
committed version at the transaction's site as of the time when the transaction
began.}
\item \textit{P2 (No Write-Write Conflicts): The write sets of each pair of committed
somewhere-concurrent transactions must be disjoint (two transactions
are somewhere concurrent if they are concurrent on site($T_1$) or site($T_2$).}
\item \textit{P3 (Commit Causality Across Sites): If a transaction $T_1$ commits at a site
A before a transaction $T_2$ starts at site A, then $T_1$ cannot commit after $T_2$
at any site.}
\end{itemize}
\end{definition}

PL-2+, on the other hand, guarantees consistent reads (transactions
never partially observe the effects of other transactions)
and disallows lost updates. Formally:
\begin{definition}
Lazy Consistency (PL-2+) $\equiv \lnot G1 \land \lnot G\text{-single}$
\end{definition}
At first blush, these system-centric definitions bear little
resemblance to each other.  Yet, when their authors explain their
intuitive meaning from a client's perspective, similarities
emerge. Cerone et al.~\cite{cerone2015framework} characterizes PSI as
requiring that transactions read from a causally consistent state and
that concurrent transactions do not write the same object. Adya 
describes PL-2+ in intriguingly similar  terms: ``PL-2+ provides a notion
of “causal consistency” since it ensures that a transaction is placed
fter all transactions that causally affect it''~\cite{adya99weakconsis}. The
``intuition gap'' between how these guarantees are formally expressed
and how they are experienced by clients makes it hard to appreciate
how these guarantees actually compare.  

Formulating isolation and consistency in terms of client-observable
states eliminates this gap by {\em forcing} definitions that
inherently specify guarantees according to how they are perceived by
clients.  The client-centric definition of PSI given below, for
example, makes immediately clear that a valid PSI execution must
ensure that all transactions observe the effects of transactions that
they depend on.

\begin{definition}
For each transaction $t$, let its  \textit{precede-set} 
contain the set of transactions after which $t$ is ordered. A
transaction $t'$  precedes $t$ if (i) $t$ reads a value that $t'$
wrote; or (ii) $t$ writes an object modified by $t'$ and the execution
orders $t'$ before $t$; or (iii) $t'$ precedes $t''$ and $t''$
precedes $t$.
\[
\ddependsmath{e}{\hat{t}}=
\{t| \exists o \in \Sigma_{\hat{t}}: t=t_{{sf_o}}\}
                \cup \{t|s_{t} \xrightarrow{+}s_{\hat{t}} \wedge \mathcal{W}_{\hat{t}} \cap \mathcal{W}_{t} \neq \emptyset \}\]
\[
\dependsmath{e}{\hat{t}} = \left( \cup_{t \in \ddependsmath{e}{\hat{t}}} \dependsmath{e}{t} \right) \cup \ddependsmath{e}{\hat{t}}
\]
\[
\commitmath{PSI}{t}{e} \equiv \prereadmath{e}{t} \wedge \forall o \in \Sigma_t: \forall t' \in \dependsmath{e}{t}: o.k \in \mathcal{W}_{t'} \Rightarrow s_{t'}  \xrightarrow{*} \slomath{o}
\]
\end{definition}

Motivated by the additional clarity afforded by this definition, we
investigated further the relationship between PSI and lazy consistency
and found that, perhaps surprisingly, PL-2+ and PSI are indeed {\em
  equivalent}.

In~\ref{appendix:psi} we prove that this client-centric,
state-based  definition of PSI is equivalent to both the axiomatic
formulation of PSI ($PSI_A$) by Cerone et al. and to the cycle-based
specification of PL-2+:
\begin{theorem}
\label{theorem:2pl}
Let $\mathcal{I}$ be PSI. Then $\exists e:\forall t \in \mathcal{T}:\commitmath{PSI}{t}{e} \equiv \lnot G1 \land \lnot G\text{-single}$ 
\end{theorem}
\begin{theorem}
\label{theorem:psi}
Let $\mathcal{I}$ be PSI. Then $\exists e:\forall t \in \mathcal{T}:\commitmath{PSI}{t}{e} \equiv PSI_A$
\end{theorem}

\subsection{Composition}

Formulating isolation and consistency guarantees on the basis of
client-observable states makes them not just easier to understand, but
also to compose. Composing such guarantees is often desirable
in practice to avoid counter-intuitive behaviors. For example, when
considered in isolation, monotonic reads lets transactions take effect
in an order inconsistent with session order, and monotonic writes puts
no constraints on reads. Systems like DocumentDB thus compose multiple
such guarantees in their implementation, but have no way of
articulating formally the new guarantee that their implementations
offer.
Expressing guarantees as local session tests makes it easy to
formalize their composition.
\begin{definition}
A storage system satisfies a set of session guarantees  $\mathcal{G} \ 
\mathit {iff}$ \\  $\forall se \in  SE : \exists e : \forall t \in T_{se} :\forall SG \in \mathcal{G}:  \sessionmath{SG}{se}{t}{e}$
\end{definition}
An analogous definition specifies the meaning of combining isolation
levels.  Once formalized, such combinations can be easily compared
against existing standalone consistency guarantees. For example,
generalizing a result by Chockler et al's~\cite{chockler2000corba}, we
prove in~\ref{appendix:causal} that also when using transactions  the four
session guarantees, taken together, are equivalent to causal
consistency. 
\begin{theorem}\label{theorem:causal}
    Let $\mathcal{G} = \{RMW,MR,MW,WFR \}$, then \\ $\forall se \in SE :\exists e: \forall t \in T_{se} :\sessionmath{\mathcal{G}}{se}{t}{e} 
    \equiv \forall se \in SE : \exists e :  \forall t \in T_{se}: \sessionmath{CC}{se}{t}{e}$
\end{theorem}
This result holds independently of the transactions' isolation
level, as it enforces no relationship between a transaction's parent
state and the read states of the operations of that transaction.
Sometimes, however, constraining the isolation level of transactions
within a session may be useful: think, for example, of a large-scale
distributed system that would like transactions to execute atomically,
while preserving session order.  Once again, formulating isolation and
consistency guarantees in terms of observable states makes expressing
such requirements straightforward: all that is needed to modularly
combine guarantees is to combine their corresponding commit and
session tests.
\begin{definition}

A storage system
	satisfies a session guarantee SG and isolation level
        $\mathcal{I}$ iff \\ $ \forall se \in SE: \exists e:
        \forall t \in T_{se} : \sessionmath{SG}{se}{t}{e} \wedge \commitmath{\mathcal{I}}{t}{e}$.
\end{definition}

\subsection{Scalability}

As we observed earlier, applications are encouraged to think of weakly
consistent systems as benevolent black boxes that hide, among others,
the details of replication.  Since replicas are invisible to clients,
our new client-centric definitions of consistency and isolation make
no mention of them, giving developers full flexibility in how these
guarantees should be implemented.

In contrast, in their current system-centric formulations, several
consistency and isolation guarantees not only explicitly refer to
replicas, but they subject operations within each replica to stronger
requirements than what is called for by these guarantees' end-to-end
obligations.  For example, the original definition of parallel
snapshot isolation~\cite{sovran2011walter} requires individual replicas to
enforce snapshot isolation, even as it globally only offers (as we
prove in Theorem~\ref{theorem:2pl}) the guarantees of Lazy
Consistency/PL-2+~\cite{adya1997lazyconsistency,adya99weakconsis}. Similarly, several definitions of causal
consistency interpret the notion of \textit{thread of
  execution}~\cite{ahamad94causalmemory} as serializing all operations
that execute on the same
site~\cite{mahajan11depot,attiya2015occ,cerone2015disc}.  Requiring
individual sites to offer stronger guarantees that what applications
can observe not only runs counter to an intuitive end-to-end argument,
but---as these guarantees translate into unnecessary dependencies
between operations and transactions that must then be honored across
sites--- can have significant implications on a system's ability to
provide a given consistency guarantee at scale.  We show
in Figure~\ref{fig:dependencies} that, for a simple transactional workload,
guaranteeing per-site PL-2+ rather than snapshot isolation
can reduce the number of dependencies per transaction by two
orders of magnitude ($175 \times$).

\begin{SCfigure}[1.5]
\includegraphics[width=0.3\textwidth]{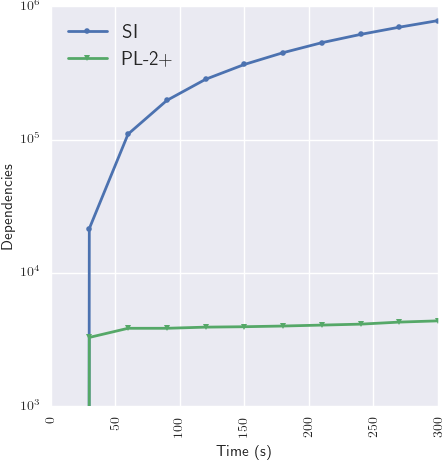}
\label{fig:dependencies}
\caption{Number of dependencies per transactions as a function of time
  when running PL-2+ vs SI at each site. Experiments are run using
  TARDiS~\cite{crooks2016tardis}, a geo-replicated transactional
  key-value store that supports both snapshot isolation and parallel
  snapshot isolation.  Our workload consists of read-write
  transactions (three reads, three writes) accessing 10,000 keys
  according to a uniform distribution.  Experiments are run on a
  shared local cluster of machines equipped with a 2.67GHz Intel Xeon
  CPU X5650 with 48GB of memory and connected by a 2Gbps
  network. Inter-machine ping latencies average 0.15 ms. Each
  experiment is run with three dedicated server machines. SI
  artificially forces transactions to be dependent on every other
  transaction at that site, while lazy consistency establishes a
  dependency only when there is an actual data conflict.  This has
  little impact on performance during failure-free executions, but
  can exacerbate the impact of slowdown cascades.}
\end{SCfigure}

\section{Related work and conclusion}
\label{sec:conclusion}

Most past  definitions of isolation and
consistency~\cite{bernstein87db,berenson1995ansi,bernstein1981,bernstein1983mcc,papadimitriou1979serializability,sovran2011walter,fekete2005msi,adya99weakconsis,ardekani2013nonmonotonic,herlihy1990linearizability,brzezinski04sessions,terry94sessionguarantees,chockler2000corba} 
%
refer to specific orderings of low-level operations and to system
properties that cannot be easily observed or understood by
applications.  To better align these definitions with
what clients perceive, recent work 
~\cite{attiya2015occ,mahajan11cacTR,cerone2015framework} distinguishes between
{\em concrete} executions (the nuts-and-bolts implementations detail)
and {\em abstract} executions (the system behaviour as perceived by
the client) on its way to introduce observable causal
consistency, a refinement of causal consistency
where causality can be inferred by client observations.
Attiya et al. introduce the notion of observable causal
consistency~\cite{attiya2015occ}, a refinement on causal consistency
where causality can be inferred by client observations. Likewise,
Cerone et al.~\cite{cerone2015framework} introduce the dual notions of
visibility and arbitration to define, axiomatically, a large number of
existing isolation levels. All continue, however, to reason about
correctness as constraints on ordering of read and write
operations. Our model takes their approach a step further: it
\textit{directly defines} consistency and isolation in terms of the 
observable states that are routinely used by developers to express
application invariants~\cite{bailis2015feral,alvaro2013borders}.

\par \textbf{Limitations}  Our model 
has two main limitations.  First, it does not constrain the behavior
of ongoing transactions, as it assumes that applications never
make externally visible decisions based on uncommitted data. It thus
cannot express consistency models, like opacity~\cite{guerraoui2008opacity} or virtual
world consistency~\cite{imbs2012virtual}  designed to prevent  STM
transactions from accessing 
an invalid memory location.
Second, our model enforces a total order on executions, which does not
easily account with recently proposed forking consistency
models~\cite{li2004sundr, mahajan11depot,crooks2016tardis,cachin2009faust}.  We leave extending the model in this direction as future work.

\par \textbf{Conclusion} We present a new way to reason about
consistency and isolation based on application-observable
states. This approach
\one\ maps more naturally to what applications can observe, in turn
making it obvious what anomalies distinct isolation/consistency levels
allow; \two\ provides  a structure to compose and compare isolation and consistency
guarantees; and  \three\ enables performance optimizations by reasoning about consistency guarantees end-to-end.

\bibliographystyle{abbrv}
\small{
\bibliography{refs}
}
\pagebreak
\setcounter{section}{0}
\setcounter{subsection}{0}
\setcounter{figure}{0}
\setcounter{table}{0}
\appendix
\renewcommand\thesection{Appendix \Alph{section}}
\renewcommand\thefigure{\Alph{section}\arabic{figure}}
\renewcommand\thetable{\Alph{section}\arabic{table}}

\section{Adya et al's model for specifying weak isolation\nc{TODO: find a better name}}
\label{appendix:adya}
Adya et al.~\cite{adya99weakconsis} introduces a cycle-based framework
for specifying weak isolation levels. We summarize its main
definitions and theorems here.

To capture a given system run, Adya uses the  notion of \textit{history}.
\begin{definition}
	A history H over a set
	of transactions consists of two parts: i) a partial order of events E that reflects the operations (e.g.,
	read, write, abort, commit) of those transactions, and ii) a version order, $<<$, that is a total order on
	committed object versions.
\end{definition}
We note that the version-order associated with a history is implementation specific. As stated in Bernstein et al~\cite{bernstein1983mcc}: as long as there exists a version order such that the corresponding direct serialization graph satisfies a given isolation level, the history satisfies
that isolation level.

The model introduces several types of direct read/write conflicts,
used to specify the \textit{direct serialization graph}.

\begin{definition}Direct conflicts:
		\begin{description}
			\item[Directly write-depends] $T_i$ writes a version of $x$, and $T_j$ writes the next version of $x$, denoted as $T_i \xrightarrow{ww} T_j$ 
			\item[Directly read-depends] $T_i$ writes a version of $x$, and $T_j$ reads the version of $x$ $T_i$ writes, denoted as $T_i \xrightarrow{wr} T_j$ 
			\item[Directly anti-depends] $T_i$ reads a version of $x$, and $T_j$ writes the next version of $x$, denoted as $T_i \xrightarrow{rw} T_j$
            \end{description}
\end{definition}

\begin{definition} Time-Precedes Order. 
	The time-precedes order, $\tprec$ , is a partial order specified for history H such that:
	\begin{enumerate}
		\item $b_i \tprec c_i$, i.e., the start point of a transaction precedes its commit point.
		\item for all i and j, if the scheduler chooses $T_j$'s start point after 
		$T_i$ 's commit point, 
		we have $c_i \tprec s_j$; otherwise, we have $b_j \tprec c_i$.
	\end{enumerate}
\end{definition}

\begin{definition}
	Direct Serialization Graph. We define the direct serialization graph arising from
	a history H, denoted DSG(H), as follows. Each node in DSG(H) corresponds to a committed
	transaction in H and directed edges correspond to different types of direct conflicts. There is a
	read/write/anti-dependency edge from transaction $T_i$ to transaction $T_j$ if $T_j$ directly read/write/antidepends
	on $T_i$. 
\end{definition}

The model is augmented with a logical notion of time, used to define the \textit{start-ordered serialization graph}.

\begin{definition}
	Start-Depends. $T_j$ start-depends on $T_i$ if $c_i \tprec s_j$ , i.e., if it starts after $T_i$ commits. We write $T_i \xrightarrow{b} T_j$ 
\end{definition}

\begin{definition}
	Start-ordered Serialization Graph or SSG. For a history H, SSG(H) contains the
	same nodes and edges as DSG(H) along with start-dependency edges.
\end{definition}

The model introduces several \textit{phenomema}, of which isolation levels proscribe a subset.

\begin{definition}Phenomena:
    
		\begin{description}
			\item[\gzero: Write Cycles] A history H exhibits phenomenon
			\gzero if DSG(H) contains a directed cycle consisting entirely of write-dependency edges.
			\item [ \gonea: Dirty Reads] A history H exhibits phenomenon \gonea~if it contains an aborted
			transaction $T_i$ and a committed transaction $T_j$ such that
			$T_j$ has read the same object (maybe via a predicate) modified by $T_i$.
			\item [ \goneb: Intermediate Reads] A history H exhibits phenomenon
			\goneb~if it contains a committed transaction $T_j$ that has read a version of
			object x written by transaction $T_i$ that was not $T_i$'s final modification
			of x.
			\item [ \gonec: Circular Information Flow] A history H exhibits phenomenon
		    \gonec if DSG(H) contains a directed cycle consisting entirely of
			dependency edges.
			\item [\gtwo: Anti-dependency Cycles] A history H exhibits phenomenon \gtwo if DSG(H) contains a directed cycle having one or more anti-dependency edges.
            \item[\gs: Single Anti-dependency Cycles] {\em DSG(H) contains a directed cycle with
exactly one anti-dependency edge.}
			\item[\gsia: Interference] A history H exhibits phenomenon \gsia~if SSG(H) contains a
			read/write-dependency edge from $T_i$ to $T_j$ without there also being a start-dependency
			edge from $T_i$ to $T_j$.
			\item [\gsib: Missed Effects] A history H exhibits phenomenon \gsib~if SSG(H) contains
			a directed cycle with exactly one anti-dependency edge. 
		\end{description}
	\end{definition}

\begin{definition}Isolations:
		\begin{description}
			\item[Serializability (PL-3)] $\lnot \gone \land \lnot \gtwo$
			\item[Read Committed (PL-2)] $\lnot \gone$
			\item[Read Uncommitted (PL-1)] $\lnot \gzero$
			\item[Snapshot Isolation]  $\lnot \gone \land \lnot \gsi$
		\end{description}
\end{definition}

\section{State-based and cycle-based model equivalence}
\label{appendix:model}
This section proves the following theorems:
\par \textbf{Theorem ~\ref{theorem:ser}}. Let $\mathcal{I}$ be Serializability (SER). Then $\ \ \exists e:\forall t
  \in \mathcal{T}. \commitmath{SER}{t}{e}  \equiv \lnot \gone \land \lnot
  \gtwo$.

\par \textbf{Theorem ~\ref{theorem:si}}. Let $\mathcal{I}$ be Snapshot Isolation (SI). Then $\ \ \exists e:\forall t
  \in \mathcal{T}. \commitmath{SI}{t}{e} \equiv \lnot \gone \land \lnot \gsi$

\par \textbf{Theorem ~\ref{theorem:rc}}. Let $\mathcal{I}$ be Read Committed (RC). Then $\ \ \exists e:\forall t
  \in \mathcal{T}. \commitmath{RC}{t}{e}  \equiv  \lnot \gone$

\subsection{Serializability}
\label{appendix:ser}
\par \textbf{Theorem ~\ref{theorem:ser}}. Let $\mathcal{I}$ be Serializability (SER). Then $\ \ \exists e:\forall t
  \in \mathcal{T}. \commitmath{SER}{t}{e}  \equiv \lnot \gone \land \lnot
  \gtwo$.

\begin{proof}
	  \textbf{We first prove} $\lnot \gone \land \lnot \gtwo \Rightarrow \exists e: \forall t \in \mathcal{T}: \commitmath{SER}{t}{e}$. 
	
Let $H$ define a history over $\mathcal{T}=\{t_1, t_2, ..., t_n\}$
  and let $DSG(H)$ be the corresponding direct serialization graph.
 Together $\lnot \gonec$ and $\lnot \gtwo$ state that the DSG(H) must not contain anti-dependency or dependency cycles: DSG(H) must therefore be acyclic. Let $i_1,...i_n$ be a permutation of $1,2,...,n$ such that $t_{i_1},...,t_{i_n}$ is a topological sort of DSG(H) (DSG(H) is acyclic and can thus be topologically sorted).

 We construct an execution $e$ according to the topological order defined above: $e: s_0 \rightarrow s_{t_{i_1}}\rightarrow s_{t_{i_2}} \rightarrow ... \rightarrow s_{t_{i_n}}$ and show that $\forall t
  \in \mathcal{T}. \commitmath{SER}{t}{e}$. Specifically, we show that for
  all $t = t_{i_j}, \completemath{e}{t_{i_j}}{s_{t_{i_{j-1}}}}$ where $s_{t_{i_{j-1}}}$
  is the parent state of $t_{i_j}$.
	
Consider the three possible types of operations in $t_{i_j}$: 
\begin{enumerate}
\item \textit{External Reads}: an operation reads an object version
that was created by another transaction. 
\item \textit{Internal Reads}: an operation reads an object version that
itself created.
\item \textit{Writes}: an operation creates a new object version. 
\end{enumerate}
We show that the parent state of $t_{i_j}$ is included in the read set of
each of those operation types:
\begin{enumerate}
		\item \textit{External Reads}. Let $r_{i_j}(x_{i_k})$  read the
        version for $x$ created by $t_{i_k}$, where $k \neq j$.\\ \\
We first show that $s_{t_{i_k}} \xrightarrow[]{*}  s_{t_{i_{j-1}}}$.
As $t_{i_j}$ directly read-depends on $t_{i_k}$, there must exist an edge $t_{i_k}\xrightarrow{wr}t_{i_j}$ in $DSG(H)$, and  $t_{i_k}$ must
therefore be ordered before $t_{i_j}$ in the topological sort of $DSG(H)$ ($k<j$). Given $e$ was constructed by applying every transaction in $\mathcal{T}$ in topological order, it follows that  $s_{t_{i_k}} \xrightarrow[]{*}  s_{t_{i_{j-1}}}$. \\ \\
Next, we argue that the state $s_{t_{i_{j-1}}}$ contains the object-value pair $(x,x_{i_k})$. Specifically, we show that there does not exists a $s_{t_{i_l}}$, where $k<l<j$, such that $t_{i_l}$ writes a different version of $x$. We prove this by contradiction. Consider the smallest such $l$: $t_{i_j}$ reads the version of $x$ written by $t_{i_k}$ and $t_{i_l}$ writes a different version of $x$.
\changebars{$t_{i_l}$, in fact, writes the next version of $x$ as $e$ is constructed following $ww$ dependencies: if there existed an intermediate version of $x$, then either $t_{i_l}$ was not the smallest transaction, or $e$ does not respect $ww$ dependencies.}{}
Note that $t_{i_j}$ thus directly anti-depends on $t_{i_l}$, i.e. $t_{i_j}\xrightarrow{rw} t_{i_l}$.
As the topological sort of $DSG(H)$ from which we constructed $e$ respects
anti-dependencies, we finally have  $s_{i_j}\xrightarrow[]{*} s_{t_{i_l}}$, i.e. $j \leq l$, a contradiction. We conclude: 	$(x,x_{i_k})\in s_{t_{i_{j-1}}} $, and therefore $s_{t_{i_{j-1}}} \in \RSmath{r_{i_j}(x_{i_k})}$.
		\item \textit{Internal Reads}. Let $r_{i_j}(x_{i_j})$
        read $x_{i_j}$ such that $w(x_{i_j})\xrightarrow{to}r(x_{i_j})$. By definition, the read state set of such an operation consists
        of $\forall s \in \mathcal{S}_e: s\xrightarrow{*} s_p$. Since $s_{t_{i_{j-1}}}$ is $t_{i_j}$'s parent state, it trivially follows that $s_{t_{i_{j-1}}} \in \RSmath{r_{i_j}(x_{i_j})}$. 
		\item \textit{Writes}.  Let $w_{i_j}(x_{i_j})$ be a write operation. By definition, its read state set consists of all the states before $s_{t_{i_{j}}}$ in
        the execution. Hence it also trivially follows that $s_{t_{i_{j-1}}} \in \RSmath{w_{i_j}(x_{i_j})}$.
	\end{enumerate}
	
	Thus $s_{t_{i_{j-1}}} \in \bigcap\limits_{o \in \Sigma_{t_{i_{j}}}} \RSmath{o} $. We have $\completemath{e}{t_{i_j}}{s_{t_{i_{j-1}}}}$ for any $t_{i_j}: \forall t \in \mathcal{T}: \commitmath{SER}{t}{e}$.

	($\Leftarrow$) \textbf{We next prove} $\exists e: \forall t \in
          \mathcal{T}: \commitmath{SER}{t}{e} \Rightarrow  \lnot \gone
          \wedge \lnot \gtwo$.

To do so, we prove the contrapositive $\gone \vee \gtwo \Rightarrow \forall e\ \exists t \in \mathcal{T}: \lnot\commitmath{SER}{t}{e}$. Let $H$ be a history that displays phenomena $\gone$ or $\gtwo$. We generate a contradiction. Consider any execution
$e$ such that $\forall t \in \mathcal{T}: \commitmath{SER}{t}{e}$. 

We first instantiate the version order for $H$, denoted as $<<$, as follows: given an execution $e$ and an object $x$, $ x_i << x_j$ if and only if $x \in \mathcal{W}_{t_i} \cap \mathcal{W}_{t_j} \land  s_{t_{i}} \xrightarrow{*} s_{t_{j}} $.
	
First, we show that: 	
	\begin{clm}
		$t_i \rightarrow t_j$ in $DSG(H) \Rightarrow s_{t_i} \xrightarrow{*} s_{t_j}$ in the execution e ($i\neq j$).
		\label{ser1}
	\end{clm}
	\begin{proof}
		Consider the three edge types in $DSG(H)$:
		\begin{description}
			\item{$t_i\xrightarrow{ww}t_j$} There exists an object $x$ s.t. $x_i<<x_j$ (version order). 
By construction, we have $s_{t_i} \xrightarrow{*} s_{t_j}$.
			
			\item{$t_i\xrightarrow{wr}t_j$} There exists an object $x$ s.t. $t_j$ reads version 
            $x_i$ written by $t_i$. Let $s_{t_k}$ be the parent state of $s_{t_j}$, i.e. $s_{t_k} \rightarrow s_{t_j}$. By assumption $\commitmath{SER}{e}{{t}}$ ($t=t_j$), i.e. $\completemath{e}{t_j}{s_{t_k}}$, hence we have $(x,x_i) \in s_{t_{k}}$. For the effects of $t_{i}$ to be visible in $s_{t_{k}}$, $t_i$ must have been applied at an earlier point in the execution. Hence we have: $s_{t_i} \xrightarrow{*} s_{t_k} \rightarrow s_{t_j}$.
			\item{$t_i\xrightarrow{rw}t_j$} There exist an object $x$ s.t. $t_i$ reads version $x_m$ written by $t_m$, $t_j$ writes $x_j$ and $x_m << x_j$. By construction, $x_m << x_j$ implies $s_{t_m} \xrightarrow{*} s_{t_j}$. Let $s_{t_k}$ be the parent state of $s_{t_j}$, i.e. $s_{t_k} \rightarrow s_{t_i}$. As $\commitmath{SER}{e}{{t}}$, where $t=t_j$, holds
by assumption, i.e. $\completemath{e}{t_i}{s_{t_k}}$, the key-value pair $(x,x_m) \in s_{t_{k}}$, hence $s_{t_m} \xrightarrow{*}s_{t_k}$ as before. In contrast, $s_{t_i} \xrightarrow{*} s_{t_j}$: indeed,$(x,x_m) \in s_{t_k}$ and $x_m << x_j$. Hence,
$t_j$ has not yet been applied.  We thus have
$s_{t_k} \xrightarrow{} s_{t_i} \xrightarrow{*} s_{t_j}$.
		\end{description}
		
	\end{proof}

We now derive a contradiction in all cases of the
disjunction $\gone \lor \gtwo$: 
\begin{itemize}
\item Let us assume that $H$ exhibits
phenomenon $\gonea$ (aborted reads). There
must exists events $w_i(x_i), r_j(x_i)$ in
$H$ such that $t_i$ subsequently aborted. $\mathcal{T}$ and any corresponding execution $e$, however, consists only of committed
transactions. Hence $\forall e: \not\exists s \in \mathcal{S}_e, s.t.\ s \in \RSmath{ r_j(x_i)}$: no
complete state can exists for $t_j$. There thus
exists a transaction for which the commit test cannot be satisfied, for any e. We have a contradiction.
\item  Let us assume that $H$ exhibits 
phenomenon $\goneb$ (intermediate reads). In an execution $e$, only the final writes of a transaction are applied. Hence,$\not\exists s \in \mathcal{S}_e, s.t.\ s \in \RSmath{ r(x_{intermediate})}$. There
thus exists a transaction, which for all $e$, will not 
satisfy the commit test. We once again have a contradiction.
		\item Finally, let us assume that the history $H$ displays one or both phenomena $\gonec$ or $\gtwo$. 
Any history that displays $\gonec$ or $\gtwo$ will contain a cycle in the $DSG$. Hence, there must exist a chain of transactions $t_i\rightarrow t_{i+1} \rightarrow ... \rightarrow t_{j}$ such that $i=j$ in $DSG(H)$.
By Claim~\ref{ser1}, we thus have $s_{t_i} \xrightarrow{*} s_{t_{i+1}} \xrightarrow{*} \dots \xrightarrow{*} s_{t_j}$ for any $e$. By definition however, a valid execution must be totally ordered. We have our
final contradiction.
\end{itemize}
All cases generate a contradiction. We have
$\gone \vee \gtwo \Rightarrow \forall e: \exists t \in \mathcal{T}: \lnot\commitmath{SER}{e}{t}$. This
completes the proof. 
\end{proof}

\subsection{Snapshot Isolation}
\label{appendix:si}
\par \textbf{Theorem ~\ref{theorem:si}}. Let $\mathcal{I}$ be Snapshot Isolation (SI). Then $\ \ \exists e:\forall t
  \in \mathcal{T}. \commitmath{SI}{t}{e} \equiv \lnot \gone \land \lnot \gsi$

\begin{proof}
	  \textbf{We first prove} $\lnot \gone \land \lnot \gsi \Rightarrow \exists e: \forall t \in \mathcal{T}: \commitmath{SI}{t}{e}$.

\par {\bf \lgraph\ properties} To do so, we introduce the notion of \textit{logical order} between transactions and capture this relationship in a \textit{logic-order}\changebars{ directed }{}\nc{find
better name} graph\changebars{ (LDG)}{}.\yp{How is Logic-order Directed Graph?} This
order refines the pre-existing start-order present in snapshot isolation to include transitive observations. $\lgraph$ contains, as nodes,
the set of committed transactions in $H$ and includes the following two edges:
\begin{itemize}
\item $t_i \xrightarrow{\ledge{1}} t_j$ iff $t_i \xrightarrow[]{\bstp} t_j$
(and $t_i \neq t_j$). \la{I assume this was $t_i$, not $t_1$} This is a
  simple renaming of the $\bstp$-edge, for clarity. \la{What is a b-edge?}
\item $t_i \xrightarrow{\ledge{2}} t_j$ iff   $\exists t_k: t_i \xrightarrow[]{\bstp} t_k \xrightarrow[]{rw} t_j$. Intuitively, this
edge captures the observation that if $t_1$ is ordered before $t_2$ and $t_2$'s reads entail that $t_2$ is before $t_3$, then $t_1$ must be before $t_3$. We note that, by construction $t_i \neq t_j$ (otherwise, the history $H$ will display
a \gsib{} cycle).
\end{itemize}

Both $\ledge{1}$ and $\ledge{2}$ edges are transitive:
\begin{clm}
\label{l1trans}
	$\ledge{1}$ edge is transitive:
$t_1 \xrightarrow[]{\ledge{1}} t_2 \xrightarrow[]{\ledge{1}} t_3 \Rightarrow t_1 \xrightarrow[]{\ledge{1}} t_3$
\end{clm}
\begin{proof}
We have  $t_1 \xrightarrow[]{\ledge{1}} t_2 \Leftrightarrow t_1 \xrightarrow[]{\bstp} t_2 \Leftrightarrow c_1 \tprec b_2$. By definition of $\tprec$,
$b _2 \tprec c_2$. Hence,  $t_1 \xrightarrow[]{\ledge{1}} t_2
\Leftrightarrow b_2 \tprec c_2$. Similarly, $t_2 \xrightarrow[]{\ledge{1}} t_3
\Leftrightarrow c_2 \tprec b_3$. Consequently, we have $c_1 \tprec b_2 \tprec c_2 \tprec b_3$. By definition $c_1 \tprec b_3 \Rightarrow t_1 \xrightarrow[]{\bstp} t_3 \Rightarrow t_1 \xrightarrow[]{\ledge{1}} t_3$.
The proof is complete. \la{what are b and c?}
\end{proof}

\begin{clm}
	\label{l2trans}
	$\ledge{2}$ edge is transitive: $t_1 \xrightarrow[]{\ledge{2}} t_2 \xrightarrow[]{\ledge{2}} t_3 \Rightarrow t_1 \xrightarrow[]{\ledge{2}} t_3$
\end{clm}
\begin{proof}
By definition of $\ledge{2}$, there must exist transactions $t_{12}, t_{23}$, such that $t_1 \xrightarrow[]{\bstp} t_{12} \xrightarrow[]{rw} t_2$ and  $t_2 \xrightarrow[]{\bstp} t_{23} \xrightarrow[]{rw} t_3$. Consider the path in $SSG(H)$: $t_1 \xrightarrow[]{\bstp} t_{12} \xrightarrow[]{rw} t_2 \xrightarrow[]{\bstp} t_{23}$. We show that: $t_1  \xrightarrow[]{\bstp} t_{23}$.
By definition, $t_1 \xrightarrow[]{\bstp} t_{12} \Rightarrow c_1 \tprec b_{12}$ and $t_{2} \xrightarrow[]{\bstp} t_{23} \Rightarrow c_{2} \tprec b_{23}$. 
By assumption, the history $H$ does not contain an \gsib{} cycle, there thus cannot exists an edge $t_{12} \xleftarrow[]{\bstp} t_{2}$ as $t_2$ either happens after or is concurrent with $t_{12}$. In both cases, we have $b_{12} \tprec c_2$.
The following relation thus hold:
 $c_1 \tprec b_{12} \tprec c_2 \tprec b_{23} $. By definition, the edge
$ t_1\xrightarrow[]{\bstp} t_{23} $ is therefore included in the $SSG(H)$, as
is the path $t_1 \xrightarrow[]{\bstp} t_{23} \xrightarrow[]{rw} t_3$.
Hence, by construction, $t_1 \xrightarrow{\ledge{2}} t_3$ exists in $\lgraph$. This completes the proof.
\end{proof}

\begin{clm}
\label{l2merge}
 $t_1 \xrightarrow[]{\ledge{2}} t_2 \xrightarrow{\ledge{1}} t_3
\Rightarrow t_1 \xrightarrow[]{\ledge{1}} t_3$.
\end{clm}
\begin{proof}
If  $t_1 \xrightarrow[]{\ledge{2}} t_2 \xrightarrow{\ledge{1}} t_3$
exists in $\lgraph$, there must exist a transaction 
 $t_{12}$ such that $t_1 \xrightarrow[]{\bstp} t_{12}\xrightarrow[]{rw} t_2 \xrightarrow[]{\bstp} t_3$ in SSG(H). By definition,    
$t_1 \xrightarrow[]{\bstp} t_{12} \Rightarrow c_1 \tprec b_{12} $ and
$t_2 \xrightarrow[]{\bstp} t_3 \Rightarrow c_2 \tprec b_3$.
Moreover, by assumption, the history $H$ does not contain an $\gsib{}$ cycle; there thus cannot exists an edge $t_{12} \xleftarrow[]{\bstp} t_{2}$ as $t_2$ either happens after or is concurrent with $t_{12}$. In both cases, we have $ b_{12} \tprec c_{2}$. Thus the following inequalities hold:
$ c_1 \tprec b_{12} \tprec c_{2} \tprec b_3$. By definition, we thus
have $t_1 \xrightarrow[]{\ledge{1}} t_3$. 
\end{proof}

Finally, we prove the following lemma:
\begin{lemma}
	\label{lacyc}
	$\lgraph$ is acyclic if H disallows phenomena \gone{} and \gsi.
\end{lemma}

\begin{proof}
\par\textbf{l1-edges} First, we show that $\lgraph$ does not contain a cycle consisting only of  $\ledge{1}$ edges. $\ledge{1}$ edges are transitive (Claim~\ref{l1trans}). Hence any cycle
consisting only of $\ledge{1}$ edges will result in a self-loop, which we have argued is impossible. Thus no such cycle can exist.

\par\textbf{l2-edges} Next, we show that  $\lgraph$ does not contain a cycle consisting only of  $\ledge{2}$ edges. The proof proceeds as above:  $\ledge{2}$ edges are transitive (Claim~\ref{l2trans}). Hence any cycle
consisting of $\ledge{2}$ edges only will result in a self-loop, which we have argued is impossible. Thus no such cycle can exist.

\par\textbf{l1/l2-edges} Finally, we show that $\lgraph$ does not contain a cycle consisting of both $\ledge{1}$ and  $\ledge{2}$ edges. We do this by contradiction. 
Assume such a cycle $cyc_1$ exists. It must contain the following sequence of transactions: 
$t_1 \xrightarrow[]{\ledge{2}} t_2 \xrightarrow[]{\ledge{1}} t_3$ in $\lgraph$.
By Claim~\ref{l2merge}, there must exist an alternative cycle $cyc_2$ containing the
same set of edges as $cyc_1$ but with the edges $t_1 \xrightarrow[]{\ledge{2}} t_2 \xrightarrow[]{\ledge{1}} t_3$ replaced by the edge $t_1 \xrightarrow[]{\ledge{1}} t_3$. If $t_1 \xrightarrow[]{\ledge{2}} t_2$ was the only $\ledge{2}$ edge in $cyc_1$,
$cyc_2$ is a cycle with only $\ledge{1}$ edges, which we previously proved could
not arise. Otherwise, we apply the same reasoning, starting from $cyc_2$, and
generate an equivalent cycle $cyc_3$ consisting of one fewer $\ledge{2}$ edge, until
we obtain a cycle $cyc_n$  with all $n$ $\ledge{2}$ edges removed, and consisting
only of $\ledge{1}$ edges. In all cases, we generate a contradiction,
and no cycle consisting of both $\ledge{1}$ and $\ledge{2}$ edges can exist.
This concludes the proof.
\end{proof}

\par\textbf{Commit Test} Armed with an acyclic $\lgraph$, we can construct an execution $e$ such that every committed transaction satisfies the commit test $\commitmath{SI}{e}{t}$.
Let $i_1,...i_n$ be a permutation of $1,2,...,n$ such that $t_{i_1},...,t_{i_n}$ is a topological sort of $\lgraph$ ($\lgraph$ is acyclic and can thus be topologically sorted).
We construct an execution $e$ according to the topological order defined above: $e: s_0 \rightarrow s_{t_{i_1}}\rightarrow s_{t_{i_2}} \rightarrow ... \rightarrow s_{t_{i_n}}$ and show that $\forall t
  \in \mathcal{T}. \commitmath{SI}{t}{e}$. 
Specifically, we prove the following: let there be
a largest \textit{k} \nc{TODO: this is a bad formulation} such that  $t_{i_k} \xrightarrow{\ledge{1} } t_{i_j}$,
then $\completemath{e}{t_{i_j}}{s_{t_{i_k}}}
\wedge (\diffmath{s_{t_{i_k}}}{s_{t_{i_{j-1}}}} \cap 
\mathcal{W}_{s_{t_{i_j}}} = \emptyset)$.
		
\par\textbf{Complete State} We first prove that $\completemath{e}{t_{i_j}}{s_{t_{i_k}}}$. Consider the three possible types of operations in $t_{i_j}$: 
\begin{enumerate}
\item \textit{External Reads}: an operation reads an object version
that was created by another transaction. 
\item \textit{Internal Reads}: an operation reads an object version that
itself created.
\item \textit{Writes}: an operation creates a new object version. 
\end{enumerate}
We show that the  $s_{t_{i_k}}$ is included in the read set of
each of those operation types:
\begin{enumerate}
		\item \textit{External Reads}. Let $r_{i_j}(x_{i_q})$  read the
        version for $x$ created by $t_{i_q}$, where $q \neq j$.\\ \\
We first show that $s_{t_{i_k}} \xrightarrow[]{*}  s_{t_{i_{q}}}$.
As $t_{i_j}$ directly read-depends on $t_{i_q}$, there must exist an edge $t_{i_q}\xrightarrow{wr}t_{i_j}$ in $SSG(H)$. Given that $H$ disallows phenomenon \gsia{}
by assumption, there must therefore exist a
start-dependency edge $t_{i_q}\xrightarrow{\bstp}t_{i_j}$ in $SSG(H)$.
$\lgraph$ will consequently contain the following edge: $t_{i_q}\xrightarrow{\ledge{1}}t_{i_j}$. Given
$e$ was constructioned by applying 
every transaction $\mathcal{T}$ in topological order, and that we select the largest $k$ such that $t_{i_k}\xrightarrow{\ledge{1}}t_{i_j}$, it follows
that  $q\leq k<j$ and $s_{i_q}\xrightarrow{*}s_{i_k} \xrightarrow{+}{s_{i_j}}$

Next, we argue that the state $s_{t_{i_k}}$ contains
the object value pair $(x, x_{i_q})$. Specifically, we argue that there does not exist a $s_{t_{i_m}}$, where $q<m\leq k$, such that $t_{i_m}$ writes a new version of $x$. We prove this by contradiction. 
Consider the smallest such $m$: $t_{i_k}$ reads the version of $x$ written by $t_{i_q}$ and $t_{i_m}$ writes the next version of $x$. $t_{i_j}$ thus directly anti-depends on $t_{i_m}$--- i.e., $t_{i_j}\xrightarrow{rw} t_{i_m}$.  In addition,
it holds by assumption that  $t_{i_k}\xrightarrow{\ledge{1}} t_{i_j}$. Hence, the following sequence of eddges exists in $SSG(H)$: 
$t_{i_k}\xrightarrow{\bstp} t_{i_j}\xrightarrow{rw} t_{i_m}$. Equivalently, $\lgraph$ will contain the edge $t_{i_k}\xrightarrow{\ledge{2}} t_{i_m}$. As the topological sort of $\lgraph$ from which we constructed $e$ respects $\ledge{2}$ edges, we finally have $s_{i_k} \xrightarrow{+} s_{i_m}$ ($k<m$), a contradiction. We conclude: $(x, x_{i_q}) \in s_{t_{i_k}}$ and therefore  $s_{t_{i_{k}}} \in \RSmath{r_{i_j}(x_{i_q})}$.

\item \textit{Internal Reads}. Let $r_{i_j}(x_{i_j})$
        read $x_{i_j}$ such that $w(x_{i_j})\xrightarrow{to}r(x_{i_j})$. By definition, the read state set of such an operation consists
        of $\forall s \in \mathcal{S}_e: s\xrightarrow{*} s_p$. Since $s_{t_{i_{k}}}$ is precedes $s_{t_{i_j}}$ in the topological order (${t_{i_k}}$ $\ledge{1}$-precedes $t_{i_j}$ and $e$ respects $\ledge{1}$ edges)  , it trivially follows that $s_{t_{i_{k}}} \in \RSmath{r_{i_j}(x_{i_j})}$. 
		\item \textit{Writes}.  Let $w_{i_j}(x_{i_j})$ be a write operation. By definition, its read state set consists of all the states before $s_{t_{i_{j}}}$ in the execution. Hence it also trivially follows that $s_{t_{i_{k}}} \in \RSmath{w_{i_j}(x_{i_j})}$.
	\end{enumerate}
	Thus $s_{t_{i_{k}}} \in \bigcap\limits_{o \in \Sigma_{t_{i_{j}}}} \RSmath{o} $.

\par \textbf{Distinct Write Sets} We now prove the second half of the commit test:
$(\diffmath{s_{t_{i_k}}}{s_{t_{i_{j-1}}}} \cap 
\mathcal{W}_{s_{t_{i_j}}} = \emptyset)$
We prove this by contradiction. Consider the largest
$m$, where $k<m<j$ such that  $\mathcal{W}_{s_{t_{i_m}}} \cap \mathcal{W}_{s_{t_{i_j}}} \neq \emptyset$.
$t_{i_m}$ thus directly write-depends on $t_{i_j}$, i.e. $t_{i_m} \xrightarrow{ww} t_{i_j}$.
By assumption, $H$ proscribes phenomenon \gsia{}. Hence, there must exist an edge $t_{i_m} \xrightarrow{\bstp} t_{i_j}$ in $SSG(H)$, and equivalently $t_{i_m} \xrightarrow{\ledge{1}} t_{i_j}$ in $\lgraph$.
As $e$ respects the topological order of $\lgraph$, it follows that $s_{i_m} \xrightarrow{+} s_{i_j}$ ($m<j$).
By assumption however,  $t_{i_k}$ is the latest transaction 
in $e$ such that $t_{i_k} \xrightarrow{\ledge{1}} t_{i_j}$, so
$m \leq k$. Since we
had assumed that $k<m<j$, we have a contradiction.
Thus, $\forall m, k<m<j , \mathcal{W}_{s_{t_{i_m}}} \cap \mathcal{W}_{s_{t_{i_j}}} = \emptyset$. We conclude that \diff{s_{t_{i_k}}}{s_{t_{i_{j-1}}}} $\cap \mathcal{W}_{s_{t_{i_j}}} = \emptyset$

We have 
  \complete{e}{t_{i_j}}{$s_{t_{i_k}}$} $\wedge ($\diff{s_{t_{i_k}}}{s_{t_{i_{j-1}}}} $\cap \mathcal{W}_{s_{t_{i_j}}} = \emptyset)$ for
  any $t_{i_j}$: $\forall t \in \mathcal{T}:
  \commitmath{SI}{t}{e}$. 

($\Leftarrow$) \textbf{We next prove} $\exists e: \forall t \in
          \mathcal{T}: \commitmath{SI}{t}{e} \Rightarrow  \lnot \gone
          \wedge \lnot \gsi$.

Let $e$ be an execution such that $\forall t \in
	\mathcal{T}: \commitmath{SI}{t}{e}$, and $H$ be a history for committed transactions
$\mathcal{T}$. We first instantiate the version order for $H$, denoted as $<<$, as follows: given an execution $e$ and an object $x$, $ x_i << x_j$ if and only if $x \in \mathcal{W}_{t_i} \cap \mathcal{W}_{t_j} \land  s_{t_{i}} \xrightarrow{*} s_{t_{j}} $. It follows that, for any two states such that $(x,x_i) \in t_{i_m}
\wedge (x,x_j) \in t_{i_n} \Rightarrow s_{t_m} \xrightarrow{+} s_{t_n}$.  We next assign the start and commit points of each transaction. We assume the existence of a monotonically
increasing timestamp counter: if a transaction $t_i$ requests a
timestamp $ts$, and a transaction $t_j$ subsequently requests a timestamp $ts'$, then
$ts<ts'$.  
Writing $e$ as $s_0 \rightarrow s_{t_1} \rightarrow s_{t_2} \rightarrow \dots \rightarrow s_{t_n}$, our timestamp assignment logic is then the following: 
\begin{enumerate}
\item Let $i = 0$.
\item Set $s = s_{t_i}$; if $i=0$, $s=s_0$.
\item Assign a commit timestamp to $t_{s_i}$ if $i\neq 0$.
\item Assign a start timestamp to all transactions
$t_k$ such that $t_k$ satisfies  \\  $\completemath{e}{t_k}{s} \wedge (\diffmath{s}{s_p(t_k)} \cap \mathcal{W}_{s_{t_k}} = \emptyset)$ and 
$t_k$ does not already have a start timestamp.
\item Let $i = i+1$. Repeat 1-4 until the final
state in $e$ is reached.
\end{enumerate}

We can relate the history's start-dependency order
and execution order as follows: 
\begin{clm}
$\forall t_i,t_j \in \mathcal{T}: s_{t_j} \xrightarrow{*} s_{t_i} \Rightarrow \lnot t_i \xrightarrow{\bstp} t_j$
\end{clm}
\begin{proof} We have $t_i \xrightarrow{\bstp} t_j \Rightarrow c_i \tprec b_j$ by definition. Moreover, the start point
of a transaction $t_i$ is always assigned before its commit point. Hence: $c_i \tprec b_j \tprec c_j$.  
It follows from our timestamp assignment logic that $s_{t_i} \xrightarrow{+} s_{t_j}$. We conclude:
$t_i \xrightarrow{\bstp} t_j \Rightarrow s_{t_i} \xrightarrow{+} s_{t_j}$. Taking the contrapositive
of this implication completes the proof. 
\end{proof}

\par \textbf{G1} We first prove that:
 $\forall t \in
          \mathcal{T}: \commitmath{SI}{t}{e} \Rightarrow  \lnot \gone$.  We do so by contradiction.

\par\textbf{G1a} Let us assume that $H$ exhibits
phenomenon \gonea{} (aborted reads). There
must exist events $w_i(x_i), r_j(x_i)$ in
$H$ such that $t_i$ subsequently aborted. $\mathcal{T}$ and any corresponding execution $e$, however, consists only of committed
transactions. Hence $\forall e: \not\exists s \in \mathcal{S}_e: s \in \RSmath{ r_j(x_i)}$: no
complete state can exists for $t_j$. There thus
exists a transaction for which the commit test cannot be satisfied, for any $e$. We have a contradiction.

\par\textbf{G1b} Let us assume that $H$ exhibits 
phenomenon \goneb{} (intermediate reads). In an execution $e$, only the final writes of a transaction are applied. Hence,$\not\exists s \in \mathcal{S}_e: s \in \RSmath{ r(x_{intermediate})}$. There
thus exists a transaction, which for all $e$, will not 
satisfy the commit test. We once again have a contradiction.
	
\par\textbf{G1c} Finally, let us assume that $H$ exhibits phenomenon \gonec: SSG(H) must contain a cycle
of read/write dependencies. We consider each possible
edge in the cycle in turn:
\begin{itemize}
\item{$t_i\xrightarrow{ww}t_j$} There must exist
an object $x$ such that $x_i << x_j$ (version order). By construction, version in $H$ is consistent with the execution order $e$: we have  $s_{t_i} \xrightarrow{*} s_{t_j}$.
\item{$t_i\xrightarrow{wr}t_j$} There must exist
a read $r_j(x_i) \in \Sigma_{t_j}$ such that
 $t_j$ reads version $x_i$ written by $t_i$.
By assumption, $\commitmath{SI}{e}{t_j}$ holds.
There must therefore exists a state $s_{t_k} \in \mathcal{S}_e$ such that $ \completemath{e}{t_j}{s_{t_k}}$. If
$s_{t_k}$ is a complete state for $t_j$,
$s_{t_k} \in \RSmath{r_j(x_i)}$ and $(x,x_i) \in s_{t_{k}}$. For the effects of $t_{i}$ to be visible in $s_{t_{k}}$, $t_i$ must have been applied at an earlier point in the execution. Hence we have: $s_{t_i} \xrightarrow{*} s_{t_k} $.
Moreover, by definition of the candidate read states, 
$s_{t_k} \xrightarrow{*} s_p(t_j) \xrightarrow{}
s_{t_j}$(Definition~\ref{defn:readstate}). It follows that  $s_{t_{i}} \xrightarrow{*} s_{t_{j}}$.
\end{itemize}
If a history $H$ displays phenomenon \gonec, there must exist a chain of transactions $t_i\rightarrow t_{i+1} \rightarrow ... \rightarrow t_{j}$ such that $i=j$.
A corresponding cycle
must thus exist in the execution $e$ $s_{t_i} \xrightarrow{*} s_{t_{i+1}} \xrightarrow{*} \dots \xrightarrow{*} s_{t_j}$. By definition however, a valid execution must be totally ordered. We once again have a contradiction.  

We generate a contradiction in all cases of the disjunction: we conclude that the history $H$ cannot display phenomenon \gone.

\par\textbf{G-SI} We now prove that $\forall t \in
          \mathcal{T}: \commitmath{SI}{t}{e} \Rightarrow  \lnot
          \gsi$.
\par\textbf{G-SIa} We first show that \gsia{} cannot happen for both write-write dependencies and
write-read dependencies:
\begin{itemize}
\item{$t_i\xrightarrow{wr}t_j$} There must exist an object $x$ such that $t_j$ reads version $x_i$ written by $t_i$. Let $s_{t_k}$ be the first state in $e$ such that $\completemath{e}{t_j}{s_{t_k}} \wedge (\diffmath{s_{t_k}}{s_p(t_j)} \cap \mathcal{W}_{s_{t_j}} = \emptyset)$. Such a state must exist since $\commitmath{SI}{e}{t_j}$ holds by assumption. As $s_{t_k}$ is complete, we have $(x,x_i) \in s_{t_{k}}$. For the effects of $t_i$ to be visible  
in $s_{t_{k}}$, $t_i$ must have been applied at an earlier point in the
execution. Hence we have: $s_{t_i} \xrightarrow{*} s_{t_k} \xrightarrow{*} s_{t_j}$. It follows from our timestamp assignment logic that $c_i \tpreceq c_k$.
Similarly, the start point of $t_j$ must have been assigned after $t_k$'s commit point (as $s_{t_k}$ is $t_j$'s earliest complete state), hence $c_k \tprec s_j$. Combining the two inequalities results in $c_i \tprec s_j$: there will exist a
start-dependency edge $t_i \xrightarrow{\bstp} t_j$. $H$ will not
display \gsia{} for write-read dependencies.
\item{$t_i\xrightarrow{ww}t_j$} There must exist an object $x$ such that $t_j$ writes the version $x_j$ that follows $x_i$. By construction, it follows that
$s_{t_i} \xrightarrow{*} s_{t_j}$. Let $s_{t_k}$ be the first state in the execution
such that  $\completemath{e}{t_j}{s_{t_k}} \wedge (\diffmath{s_{t_k}}{s_p(t_j)} \cap \mathcal{W}_{{t_j}} = \emptyset)$. We first show that:
$s_{t_i} \xrightarrow{*} s_{t_k}$. Assume by way of contradiction that $s_{t_k} \xrightarrow{+} s_{t_i}$. The existence of a write-write dependency between $t_i$ and $t_j$ implies that $\mathcal{W}_{t_i} \cap \mathcal{W}_{t_j} \neq \emptyset$,
and consequently, that $\diffmath{s_{t_k}}{s_p(t_j)} \cap \mathcal{W}_{t_j} \neq \emptyset$, contradicting our assumption that $\commitmath{SI}{e}{t_j}$.
We conclude that:  $s_{t_i} \xrightarrow{*} s_{t_k}$. It follows from our timestamp assignment logic that $c_i \tpreceq c_k$.
Similarly, the start point of $t_ij$ must have been assigned after $t_k$'s commit point (as $s_{t_k}$ is $t_j$'s earliest complete state), hence $c_k \tprec s_j$. Combining the two inequalities results in $c_i \tprec s_j$: there will exist a
start-dependency edge $t_i \xrightarrow{\bstp} t_j$. $H$ will not
display  \gsia for write-write dependencies.
\end{itemize}
The history $H$ will thus not display phenomenon \gsia.

\par{\textbf{G-SIb}} We next prove that $H$ will not display
phenomenon \gsib.
Our previous result states that $H$ proscribes \gsia: all read-write dependency edges between two transactions implies the existence of a start
dependency edge between those same transactions. We prove by contradiction
that $H$ proscribes  \gsib. Assume that SSG(H) consists of a directed cycle $cyc_1$ with exactly one anti-dependency edge (it displays \gsib) but proscribes
\gsia. All other dependencies will therefore be write/write dependencies, write/read dependencies, or start-depend edges. By \gsia, there must exist an equivalent cycle $cyc_2$ consisting of a directed cycle with
exactly one anti-dependency edge and start-depend edges only. Start-edges are transitive (Claim~\ref{l1trans}), hence there must exist a cycle $cyc_3$ with exactly one
anti-dependency edge and one start-depend edge. We write $t_i \xrightarrow{rw} t_j \xrightarrow{\bstp} t_i$. Given $t_i \xrightarrow{rw} t_j$, there must exist an object $x$ and transaction $t_m$
such that $t_m$ writes $x_m$, $t_i$ reads $x_m$ and $t_j$ writes the next version of $x$, $x_j$ ($x_m << x_j$). Let $s_{t_k}$ be the earliest complete state of $t_i$. Such a state must exist as $\commitmath{SI}{e}{t_i}$ by assumption.
Hence, by definition of read state $(x,x_m) \in s_{t_k}$. Similarly,
$(x,x_j) \in s_{t_j}$ by the definition of state transition
(Definition~\ref{def:statetrans}). By construction, we have $s_{t_k} \xrightarrow{+} s_{t_j}$. 
Our timestamp assignment logic maintains the following invariant: given
a state $s_t$, $\forall t_k:\completemath{e}{t_k}{s_t} : \forall s_{t_n}: s_t \xrightarrow{+} {s_{t_n}} \Rightarrow b_k \prec_t c_n$. Intuitively, the start timestamp
of all transactions associated with a particular complete state $s_t$ is smaller
than the commit timestamp of any transaction that follows $s_t$ in the execution.
We previously showed that $s_{t_k} \xrightarrow{+} s_{t_j}$. Given $s_{t_k}$
is a complete state for $t_i$, we conclude $b_i \prec_t c_j$. However,
the edge  $t_j \xrightarrow{\bstp} t_i$ implies that $c_j \prec_t b_i$.
We have a contradiction: no such cycle can exist and $H$ will not 
display phenomenon \gsi.

We generate a contradiction in all cases of the conjunction, hence $\forall t \in
          \mathcal{T}: \commitmath{SI}{t}{e} \Rightarrow  \lnot
          \gsi $ holds.

We conclude $\forall t \in
          \mathcal{T}: \commitmath{SI}{t}{e} \Rightarrow  \lnot
          \gsi \land \gone$. This completes the proof.
          \end{proof}

\subsection{Read Committed}
\label{appendix:rc}

\par \textbf{Theorem ~\ref{theorem:rc}}. Let $\mathcal{I}$ be Read Committed (RC). Then $\ \ \exists e:\forall t
  \in \mathcal{T}. \commitmath{RC}{t}{e}  \equiv \lnot \gone $.

\begin{proof}
	  \textbf{We first prove} $\lnot \gone  \Rightarrow \exists e: \forall t \in \mathcal{T}: \commitmath{RC}{t}{e}$. 
	
Let $H$ define a history over $\mathcal{T}=\{t_1, t_2, ..., t_n\}$
and let $DSG(H)$ be the corresponding direct serialization graph.
 $\lnot \gonec$ states that the $DSG(H)$ must not contain dependency cycles:
 the subgraph of $DSG(H)$, $SDSG(H)$ containing the same nodes but
 including only dependency edges,  must be acyclic. Let $i_1,...i_n$ be a permutation of $1,2,...,n$ such that $t_{i_1},...,t_{i_n}$ is a topological sort of $SDSG(H)$ ($SDSG(H)$ is acyclic and can thus be topologically sorted).

 We construct an execution $e$ according to the topological order defined above: $e: s_0 \rightarrow s_{t_{i_1}}\rightarrow s_{t_{i_2}} \rightarrow ... \rightarrow s_{t_{i_n}}$ and show that $\forall t
  \in \mathcal{T}. \commitmath{RC}{t}{e}$. Specifically, we show that for
  all $t = t_{i_j}, \prereadmath{e}{t}$.
	
Consider the three possible types of operations in $t_{i_j}$: 
\begin{enumerate}
\item \textit{External Reads}: an operation reads an object version
that was created by another transaction. 
\item \textit{Internal Reads}: an operation reads an object version that
itself created.
\item \textit{Writes}: an operation creates a new object version. 
\end{enumerate}
We show that the read set for  
each of operation type is not empty:
\begin{enumerate}
		\item \textit{External Reads}. Let $r_{i_j}(x_{i_k})$  read the
        version for $x$ created by $t_{i_k}$, where $k \neq j$.\\ \\
We first show that $s_{t_{i_k}} \xrightarrow[]{*}  s_{t_{i_j}}$.
As $t_{i_j}$ directly read-depends on $t_{i_k}$, there must exist an edge $t_{i_k}\xrightarrow{wr}t_{i_j}$ in $SDSG(H)$, and  $t_{i_k}$ must
therefore be ordered before $t_{i_j}$ in the topological sort of $SDSG(H)$ ($k<j$), it follows that  $s_{t_{i_k}} \xrightarrow[]{+}  s_{t_{i_j}}$. 
As $(x,x_{i_k}) \in s_{t_{i_k}}$, we have $s_{t_{i_{k}}} \in \RSmath{r_{i_j}(x_{i_k})}$, and consequently $\RSmath{r_{i_j}(x_{i_k})} \neq \emptyset$.
		
		\item \textit{Internal Reads}. Let $r_{i_j}(x_{i_j})$
        read $x_{i_j}$ such that $w(x_{i_j})\xrightarrow{to}r(x_{i_j})$. By definition, the read state set of such an operation consists
        of $\forall s \in \mathcal{S}_e: s\xrightarrow{*} s_p$.  $ s_0 \xrightarrow{*} s$ trivially holds. We conclude $s_0 \in \RSmath{r_{i_j}(x_{i_j})}$, i.e. $\RSmath{r_{i_j}(x_{i_j})} \neq \emptyset$. 
		\item \textit{Writes}.  Let $w_{i_j}(x_{i_j})$ be a write operation. By definition, its read state set consists of all the states before $s_{t_{i_{j}}}$ in
        the execution. Hence $s_0 \in \RSmath{r_{i_j}(x_{i_j})}$, i.e. $\RSmath{r_{i_j}(x_{i_j})} \neq \emptyset$. 
	\end{enumerate}
	
	Thus $\forall o \in \Sigma_t: \RSmath{o} \neq \emptyset$. We have $\prereadmath{e}{t_{i_j}}$ for any $t_{i_j}: \forall t \in \mathcal{T}: \commitmath{RC}{t}{e}$.

	($\Leftarrow$) \textbf{We next prove} $\exists e: \forall t \in
          \mathcal{T}: \commitmath{RC}{t}{e} \Rightarrow  \lnot \gone
        $.

To do so, we prove the contrapositive $\gone \Rightarrow \forall e\ \exists t \in \mathcal{T}: \lnot\commitmath{RC}{t}{e}$. Let $H$ be a history that displays phenomena \gone. We generate a contradiction. Assume that there exists an execution
$e$ such that $\forall t \in \mathcal{T}: \commitmath{RC}{t}{e}$. 

We first instantiate the version order for $H$, denoted as $<<$, as follows: given an execution $e$ and an object $x$, $ x_i << x_j$ if and only if $x \in \mathcal{W}_{t_i} \cap \mathcal{W}_{t_j} \land  s_{t_{i}} \xrightarrow{+} s_{t_{j}} $.
	
First, we show that: 	
	\begin{clm}
		$t_i \rightarrow t_j$ in $SDSG(H) \Rightarrow s_{t_i} \xrightarrow{+} s_{t_j}$ in the execution e ($i\neq j$).
		\label{claim:rc}
	\end{clm}
	\begin{proof}
		Consider the three edge types in $DSG(H)$:
		\begin{description}
			\item{$t_i\xrightarrow{ww}t_j$} There exists an object $x$ s.t. $x_i<<x_j$ (version order). 
By construction, we have $s_{t_i} \xrightarrow{+} s_{t_j}$.
			
			\item{$t_i\xrightarrow{wr}t_j$} There exists an object $x$ s.t. $t_j$ reads version 
            $x_i$ written by $t_i$, i.e. $r_j(x,x_i) \in \Sigma_{t_i}$. By assumption $\commitmath{RC}{e}{{t}}$ ($t=t_j$), i.e. $\prereadmath{e}{t_j}$, $\RSmath{o} \neq \emptyset$. Let $s \in \RSmath{o}$, by definition of $\RSmath{o}$, we have $(x,x_i) \in s \wedge s \xrightarrow{+} s_{t_j}$, therefore $t_i$ must be applied before or on state $s$, hence we have $s_{t_i}\xrightarrow{*} s \xrightarrow{+} s_{t_j}$, i.e. $s_{t_i} \xrightarrow{+} s_{t_j}$. 
			
		\end{description}
		
	\end{proof}

We now derive a contradiction in all cases of $\gone$: 
\begin{itemize}
\item Let us assume that $H$ exhibits
phenomenon \gonea~(aborted reads). There
must exists events $w_i(x_i), r_j(x_i)$ in
H such that $t_i$ subsequently aborted. $\mathcal{T}$ and any corresponding execution $e$, however, consists only of committed
transactions. Hence $\forall e: \not\exists s \in \mathcal{S}_e, s.t.\ s \in \RSmath{ r_j(x_i)}$: no
complete state can exists for $t_j$. There thus
exists a transaction for which the commit test cannot be satisfied, for any e. We have a contradiction.
\item  Let us assume that $H$ exhibits 
phenomenon \goneb~(intermediate reads). In an execution $e$, only the final writes of a transaction are applied. Hence,$\not\exists s \in \mathcal{S}_e, s.t.\ s \in \RSmath{ r(x_{intermediate})}$. There
thus exists a transaction, which for all e, will not 
satisfy the commit test. We once again have a contradiction.
		\item Finally, let us assume that the history $H$ displays \gonec. 
Any history that displays \gonec will contain a cycle in the SDSG(H). Hence, there must exist a chain of transactions $t_i\rightarrow t_{i+1} \rightarrow ... \rightarrow t_{j}$ such that $i=j$.
By Claim~\ref{claim:rc}, we thus have $s_{t_i} \xrightarrow{+} s_{t_{i+1}} \xrightarrow{+} \dots \xrightarrow{+} s_{t_j}$, $i=j$ for any $e$. By definition however, a valid execution must be totally ordered. We have our
final contradiction.
\end{itemize}
All cases generate a contradiction. We have
$\gone \Rightarrow \forall e: \exists t \in \mathcal{T}: \lnot\commitmath{RC}{e}{t}$. This
completes the proof. 
\end{proof}

\section{Causality and Session Guarantees}
\label{appendix:causal}

We prove the following theorems:
\par \textbf{Theorem~\ref{theorem:causal}}
    Let $\mathcal{G} = \{RMW,MR,MW,WFR \}$, then \\ $\forall se \in SE :\exists e: \forall t \in T_{se} :\sessionmath{\mathcal{G}}{se}{t}{e} 
    \equiv \forall se \in SE : \exists e :  \forall t \in T_{se}: \sessionmath{CC}{se}{t}{e}$

    We first state a number of useful lemmas about the
    $\prereadmath{e}{\mathcal{T}}$ predicate
    (Definition~\ref{def:preread}): if $\prereadmath{e}{\mathcal{T}}$
    holds, then the candidate read set of all operations in all
    transactions in $\mathcal{T}$ is not empty. The first
    lemma states that an operation's read state must reflect writes
    that took place before the transaction committed, while the second
    lemma simply argues that the predicate is closed under subset.
\vspace{3mm}

\begin{lemma}
	\label{sfolemma}
	For any $\mathcal{T}'$ such that $\mathcal{T}' \subseteq \mathcal{T}$, $\prereadmath{e}{\mathcal{T}'} \Leftrightarrow \forall t \in \mathcal{T}': \forall o \in \Sigma_t: \sfomath{o} \xrightarrow{+} s_{t}$.
\end{lemma}
\vspace{-3mm}
\begin{proof}
	($\Rightarrow$)  We first prove  $\prereadmath{e}{\mathcal{T}'} \Rightarrow \forall t \in \mathcal{T}': \forall o \in \Sigma_t: \sfomath{o} \xrightarrow{+} s_{t}$. 

By the definition of $\prereadmath{e}{\mathcal{T}'}$, we have $ \forall t \in \mathcal{T}', \forall o \in \Sigma_t$, $\RSmath{o} \neq \emptyset$.
We consider the two types of operations: reads and writes.
\begin{addmargin}[5mm]{0pt}
\textbf{Reads}	The set of candidate read states of a read operation $o= r(k,v)$
is defined as $\RSmath{o} = \{s \in \mathcal{S}_e | 
	s \xrightarrow{*} s_p \wedge\big ( (k,v) \in s \vee (\exists w(k,v) \in \Sigma_t: w(k,v) \xrightarrow{to} r(k,v)) \big)\}$. The disjunction considers
two cases:
	\begin{enumerate}
		\item \textit{Internal Reads} if $\exists w(k,v) \in \Sigma_t: w(k,v) \xrightarrow{to} r(k,v)$, $\RSmath{o} = \{s \in \mathcal{S}_e | s \xrightarrow{*} s_p\}$. Hence $s_0 \in \RSmath{o}$. It follows that $\sfomath{o} = s_0 \xrightarrow{+} s_t$
		\item \textit{External Reads} By $\prereadmath{e}{\mathcal{T}'}$, we have that $\RSmath{o} \neq \emptyset$. There must therefore exist a state $s \in \mathcal{S}_e$ such that $s \xrightarrow{*} s_p \wedge (k,v) \in s$. Since $\sfomath{o}$ is, by definition, the first such $s$, we have that $\sfomath{o} \xrightarrow{*}s_p \rightarrow s_t$. We conclude: $\sfomath{o} \xrightarrow{+} s_t$.
	\end{enumerate}
\end{addmargin}
\begin{addmargin}[5mm]{0pt}
\textbf{Writes} The candidate read states set for write operations $o=w(k,v)$, is defined as $\RSmath{o} = \{s \in \mathcal{S}_e | 
	s \xrightarrow{*} s_p\}$. Hence, $s_0 \in \RSmath{o}$. It trivially
    follows that $(\sfomath{o} = s_0) \xrightarrow{+} s_t$.
\end{addmargin}	
	We conclude: $\prereadmath{e}{\mathcal{T}'} \Rightarrow \forall t \in \mathcal{T}': \forall o \in \Sigma_t: \sfomath{o} \xrightarrow{+} s_{t}$.
	
	($\Leftarrow$) Next, we prove that
  $\prereadmath{e}{\mathcal{T}'} \Leftarrow \forall t \in \mathcal{T}': \forall o \in \Sigma_t: \sfomath{o} \xrightarrow{+} s_{t}$. 

    By assumption, $\forall t \in \mathcal{T}': \forall o \in \Sigma_t: \sfomath{o} \xrightarrow{+} s_{t}$. By definition, $\sfomath{o} \in \RSmath{o}$.
It trivially follows that  $\forall t \in \mathcal{T}': \forall o \in \Sigma_t: \RSmath{o} \neq \emptyset$, i.e. $\prereadmath{e}{\mathcal{T}'}$ holds.
\end{proof}
\vspace{1mm}
\begin{lemma}
	\label{prereadlemma}
	For any $\mathcal{T}'$ that $\mathcal{T}' \subseteq \mathcal{T}$, $\prereadmath{e}{\mathcal{T}} \Rightarrow \prereadmath{e}{\mathcal{T}'}$.
\end{lemma}
\vspace{-3mm}
\begin{proof}
	Given $\prereadmath{e}{\mathcal{T}}$, by definition we have $ \forall t \in \mathcal{T}, \forall o \in \Sigma_t$, $\RSmath{o} \neq \emptyset$. Since $\mathcal{T}' \subseteq \mathcal{T}$, $\forall t \in \mathcal{T}' \Rightarrow \forall t \in \mathcal{T}$, it follows that $ \forall t \in \mathcal{T}', \forall o \in \Sigma_t$, $\RSmath{o} \neq \emptyset$, i.e. $\prereadmath{e}{\mathcal{T}'}$.
\end{proof}

We now begin in earnest our proof of Theorem~\ref{theorem:causal}.

	($\Leftarrow$) \textbf{We first prove that}  
 $ \forall se \in SE : \exists e :  \forall t \in \mathcal{T}_{se}: \sessionmath{CC}{se}{t}{e} \Rightarrow \forall se \in SE :\exists e: \forall t \in \mathcal{T}_{se} :\sessionmath{\mathcal{G}}{se}{t}{e} $. 

For any $se \in SE$, consider the execution $e$, such that  $\forall t \in \mathcal{T}_{se}: \sessionmath{CC}{se}{t}{e}$. We show that this
same execution satisfies the session test of all four session guarantees.
\begin{itemize}
\item\textbf{CC $\Rightarrow$ RMW}: By
assumption, $\sessionmath{CC}{se}{t}{e}$ for
all $\mathcal{T}_{se}$. Hence: $\forall t \in \mathcal{T}_{se}: \forall o \in \Sigma_t: \forall t' \xrightarrow{se} t:  s_{t'} \xrightarrow{*} sl_o$. Weakening
this statement gives the following
implication: $\forall o \in \Sigma_t: \forall t'\xrightarrow{se} t:\mathcal{W}_{t'} \neq \emptyset \Rightarrow s_{t'} \xrightarrow{*}sl_o$. Additionally,
$e$ satisfies  $\prereadmath{e}{\mathcal{T}}$ (by assumption) and
therefore  $\prereadmath{e}{\mathcal{T}_{se}}$ as $\mathcal{T}_{se}
\subseteq \mathcal{T}$ (by Lemma \ref{prereadlemma}). Putting it all  together: $e$ 
satisfies $\prereadmath{e}{\mathcal{T}_{se}} \wedge \forall o \in
\Sigma_t: \forall t'\xrightarrow{se} t:\mathcal{W}_{t'} \neq \emptyset
\Rightarrow s_{t'} \xrightarrow{*}sl_o$. 

We conclude that $\forall se \in SE: \exists e:  \forall t \in
\mathcal{T}_{se} : \sessionmath{RMW}{se}{t}{e}$.
	
\item\textbf{CC $\Rightarrow$ MW}: 
By assumption, $\sessionmath{CC}{se}{t}{e}$ for
all $t\in \mathcal{T}_{se}$. Hence, it holds that $\forall se' \in SE: \forall t_i \xrightarrow{se'} t_j: s_{t_i} \xrightarrow{+} s_{t_j}$. Weakening
this statement gives the following
implication: $\forall se' \in SE: \forall t_i' \xrightarrow{se'} t_j: (\mathcal{W}_{t_i} \neq \emptyset \wedge
		\mathcal{W}_{t_j} \neq \emptyset) \Rightarrow s_{t_i} \xrightarrow{+} s_{t_j} $.  Additionally,
$e$ satisfies  $\prereadmath{e}{\mathcal{T}}$ (by assumption) and
therefore  $\prereadmath{e}{\mathcal{T}_{se}}$ as $\mathcal{T}_{se}
\subseteq \mathcal{T}$ (by Lemma \ref{prereadlemma}). Putting it all together: 
$\prereadmath{e}{\mathcal{T}_{se}} \wedge \forall se' \in SE: \forall
t_i' \xrightarrow{se'} t_j: (\mathcal{W}_{t_i} \neq \emptyset \wedge
\mathcal{W}_{t_j} \neq \emptyset) \Rightarrow s_{t_i} \xrightarrow{+}
s_{t_j} $. 

We conclude that $\forall se \in SE: \exists e:  \forall t \in
\mathcal{T}_{se} : \sessionmath{MW}{se}{t}{e}$.

\item\textbf{CC $\Rightarrow$ MR}: By 
assumption, $ \sessionmath{CC}{se}{t}{e}$, 
hence $e$ ensures that $\forall t \in \mathcal{T}_{se}: \forall o \in \Sigma_t: \forall t' \xrightarrow{se} t:  s_{t'} \xrightarrow{*} sl_o$.  
Moreover, by assumption, we have that $\prereadmath{e}{\mathcal{T}}$. It follows that
$\forall t' \in \mathcal{T}: \forall o' \in \Sigma_{t'}: \sfomath{o'} \xrightarrow{+} s_{t'}$ (Lemma~\ref{sfolemma}). 
		Combining the two statements, we have $\forall o \in \Sigma_t: \forall t' \xrightarrow{se} t: \forall o' \in \Sigma_{t'}: \sfomath{o'} \xrightarrow{*} s_{t'} \xrightarrow{*} \slomath{o} $, i.e. $\sfomath{o'} \xrightarrow{*}  \slomath{o}$.
Finally, we have that $e$ satisfies
 $\ircmath{e}{t}$ by assumption, and $\prereadmath{e}{\mathcal{T}_{se}}$
by Lemma~\ref{prereadlemma}: we have	$\prereadmath{e}{\mathcal{T}}$
and $\mathcal{T}_{se} \subseteq \mathcal{T}$. Putting
it all together, $e$ satisfies
		$\prereadmath{e}{\mathcal{T}_{se}} \wedge
                \ircmath{e}{t} \wedge \forall o \in \Sigma_t: \forall
                t'\xrightarrow{se}t: \forall o' \in \Sigma_{t'}:
                sf_{o'}\xrightarrow{*}sl_o $.

 We conclude that  $\forall se \in SE: \exists e:  \forall t \in \mathcal{T}_{se}: \sessionmath{MR}{se}{t}{e}$.

\item\textbf{CC $\Rightarrow$ WFR}: By assumption,
$ \sessionmath{CC}{se}{t}{e}$. Hence $e$ satisfies $\forall se' \in SE: \forall t_i \xrightarrow{se'} t_j: s_{t_i} \xrightarrow{+} s_{t_j}$.
By assumption, $e$ respects $\prereadmath{e}{\mathcal{T}}$. It follows from lemma \ref{sfolemma} that $\forall t_i \in \mathcal{T}: \forall o_i \in \Sigma_{t_i}: \sfomath{o_i} \xrightarrow{+} s_{t_i}$. 
We have, by combining these two statements, that: $\forall se' \in SE: \forall t_i\xrightarrow{se'}t_j: \forall o_i \in \Sigma_{t_i}: sf_{o_i} \xrightarrow{+} s_{t_i} \xrightarrow{+} s_{t_j} $, i.e. $sf_{o_i} \xrightarrow{+} s_{t_j} $. Weakening this
statement results in the following implication: 
$\forall se' \in SE: \forall t_i\xrightarrow{se'}t_j: \forall o_i \in \Sigma_{ti}:\mathcal{W}_{t_j} \neq \emptyset \Rightarrow sf_{o_i} \xrightarrow{+} s_{t_j} $. Putting it all together, $e$ satisfies	
		$\prereadmath{e}{\mathcal{T}} \wedge \forall se' \in
                SE: \forall t_i\xrightarrow{se'}t_j: \forall o_i \in
                \Sigma_{ti}:\mathcal{W}_{t_j} \neq \emptyset
                \Rightarrow sf_{o_i} \xrightarrow{+} s_{t_j}$.

                We conclude that 
                $\forall se \in SE: \exists e: \forall t \in
                \mathcal{T}_{se} : \sessionmath{WFR}{se}{t}{e}$.

\end{itemize}

	($\Rightarrow$) \textbf{We now prove that}, given	
	$\mathcal{G} = \{RMW,MR,MW,WFR \}$, the following implication holds: $\forall se \in SE :\exists e: \forall t \in \mathcal{T}_{se} :\sessionmath{\mathcal{G}}{se}{t}{e} 
	\Rightarrow \forall se \in SE : \exists e :  \forall t \in \mathcal{T}_{se}: \sessionmath{CC}{se}{t}{e}$.

To this end, we prove that for any session $se$, given the execution
$e$ such that $\forall t \in \mathcal{T}_{se} :\sessionmath{\mathcal{G}}{se}{t}{e} $, we can construct an equivalent execution $e'$ that satisfies all four session guarantees, such that $\forall t \in \mathcal{T}_{se}: \sessionmath{CC}{se}{t}{e'}$.

The need for constructing an alternative but equivalent execution $e'$ may be counter-intuitive at first. We motivate it informally here: session guarantees
place no requirements on the commit order of read-only transactions. In contrast,
causal consistency requires all transactions to commit in session order. As read-only transactions have no effect on the candidate read states of other read-only or update transactions, it is therefore always possible to generate an equivalent execution
$e'$ such that update transactions commit in the same order as in $e$, and read-only transactions commit in session order. Our proof shows that, if $e$ satisfies
all four session guarantees, $e'$ will satisfy all four session guarantees. This
will in turn imply that $e'$ satisfies causal consistency.

\par \textbf{Equivalent Execution} We now describe more formally how to construct this execution $e'$: first, we apply in $e'$ all transactions $t \in \mathcal{T}$ such that $\mathcal{W}_t \neq \emptyset$, respecting their commit order in $e$. We denote the states created by applying $t$ in $e$ and $e'$ as $s_{e,t}$ and $s_{e',t}$ respectively. Our construction enforces the following relationship between
$e$ and $e'$: $\forall t \in \mathcal{T} \wedge \mathcal{W}_t \neq \emptyset:  (k,v) \in s_{e,t} \Leftrightarrow (k,v) \in s_{e',t}$. All update transactions
are applied in the same order, and read-only transactions have no effect on the state. Next, we consider the read-only transactions $t_i \in \mathcal{T}$ in session
order: we select the parent state for $t_i$ to be $\max\{\max_{o \in \Sigma_{t_i}}\{\sfomath{o}\}, s_{t_{i-1}}\}$, where $ t_{i-1}$ denotes the transaction that directly precedes $t$ in a session. A session defines a total order of transactions: $t_{i-1}$ is unique. 
\yp{multiple read-only transactions may select the same $\max_{o \in \Sigma_{t}}\{\sfomath{o}\}$ to be the parent state. We never mention this case (too hard to talk about it clearly, and it is not a fundamental case that we need to distinguish), but the statement is still true.}
If $t_{i}$ is the first transaction in the 
session, we simply set $s_{t_{i-1}}$ to $s_0$. As transactions do not change
the value of states, this process maintains the previously stated invariant:
 $\forall t \in \mathcal{T} \wedge \mathcal{W}_t \neq \emptyset: s_{e,t} \equiv s_{e',t}$.

 We now proceed to prove that $e'$ satisfies
 $\prereadmath{e'}{\mathcal{T}}$ and the session test for all session
 guarantees.

\begin{addmargin}[5mm]{3mm}
\par \textbf{Preread} First, we show that
$\prereadmath{e'}{\mathcal{T}}$ holds.
We distinguish between update and read-only transactions: 
\begin{itemize}
\item \textit{Read-Only Transactions.} By construction, the parent
state of a read-only transaction $t_i$ is $s_p(t_i) = \max\{\max_{o \in \Sigma_t}\{\sfomath{o}\}, s_{t_{i-1}}\}$. It follows that $\forall o \in \Sigma_t, s_p(t) \geq \sfomath{o}$ and consequently $\sfomath{o} \xrightarrow{*} s_p(t) \rightarrow s_t$. We have $\forall o \in \Sigma_t:\sfomath{o} \xrightarrow{+} s_t$ in $e'$.
\item \textit{Update Transactions.} Update transactions $t$ consist of both read and write operations. A write operation $o=w(k,v)\in \Sigma_t$ has for
candidate read set the set of all states $s \in \mathcal{S}_{e'}$ such that $s \xrightarrow{*} s_p(t)$. Hence $\sfomath{o} = s_0$ and $\sfomath{o} \xrightarrow{+} s_t$ in $e'$ trivially holds. The state corresponding to $\sfomath{o}$ for read operations $o=r(x_i)$ is the state created by the transaction $t_i$ that wrote version $x_i$ of object $x$: $\sfomath{o}=s_{e',t_i}$. By assumption, $e$ satisfies
$\prereadmath{e}{\mathcal{T}}$, hence by Lemma~\ref{sfolemma}, we have $s_{t_i} \xrightarrow{+} s_{t}$ in $e$. By construction (update transactions
are applied in $e'$ in the same order as $e$), it follows that $s_{t_i} \xrightarrow{+} s_{t}$ in $e'$, i.e. $\sfomath{o} \xrightarrow{+} s_{t}$ in $e'$.
\end{itemize}
By Lemma~\ref{sfolemma}, we conclude that $\prereadmath{e'}{\mathcal{T}}$ holds.\vspace{3mm}

\par \textbf{MW} We next show that $e'$ satisfies
$\sessionmath{MW}{se}{t}{e'}$ for all sessions $se \in SE$ and
$\forall t \in \mathcal{T}_{se}$.  Consider any session $se'$ and
two transactions $t_i, t_j \in \mathcal{T}_{se'}$ such that
$t_i \xrightarrow{se'} t_j$, and
$\mathcal{W}_{t_i} \neq \emptyset \wedge \mathcal{W}_{t_j} \neq
\emptyset$.
As $e$, by assumption, satisfies
$\forall t \in \mathcal{T}_{se} :\sessionmath{MW}{se}{t}{e} $, we have
$s_{e,t_i} \xrightarrow{+} s_{e,t_j} $.  Since $t_i$ and $t_j$ are
update transactions, they are applied in the same order in $e'$ as in
$e$: hence, $s_{e',t_i} \xrightarrow{+} s_{e',t_j}$.  Further, recall
that we previously
showed that $\prereadmath{e'}{\mathcal{T}}$ (and consequently
$\prereadmath{e'}{\mathcal{T}_{se}}$ by Lemma~\ref{prereadlemma}).

Putting it all together, we conclude that 
$\prereadmath{e'}{\mathcal{T}_{se}} \wedge \forall se' \in SE: \forall
t_i \xrightarrow{se'} t_j: (\mathcal{W}_{t_i} \neq \emptyset \wedge
\mathcal{W}_{t_j} \neq \emptyset) \Rightarrow s_{t_i} \xrightarrow{+}
s_{t_j} $,
i.e.,  $\forall t \in \mathcal{T}_{se} :\sessionmath{MW}{se}{t}{e'}$.

\vspace{3mm}
 				
\par \textbf{WFR} Consider all update transactions such that  $\forall se' \in SE: \forall t_i\xrightarrow{se'}t_j: \forall o_i \in \Sigma_{t_i}:\mathcal{W}_{t_j} \neq \emptyset $. We prove that $\sfomath{o_i} \xrightarrow{+} s_{t_j} $. Consider the two types of operations that arise in an update transaction:
\begin{itemize}
\item \textit{Reads.} The state corresponding to $\sfomath{o_i}$ for read operations $o=r_i(x,x_k)$ is the state created by the transaction $t_k$ that wrote version $x_k$ of object $x$: $\sfomath{o_i}=s_{t_k}$.
By assumption, $e$ satisfies $\forall t \in \mathcal{T}_{se}
:\sessionmath{WFR}{se}{t}{e}$, we have $\sfomath{o_i} \xrightarrow{+}
s_{t_j}$ in $e$, i.e. $s_{t_k} \xrightarrow{+} s_{t_j}$ in $e$. Since
we apply update transactions in $e'$ the same order as in $e$, it follows that $s_{t_k} \xrightarrow{+} s_{t_j}$ in $e'$, i.e., $\sfomath{o_i} \xrightarrow{+} s_{t_j}$ in $e'$.
\item \textit{Writes.} The candidate read states set for write operations $o=w(x_i)$, is defined as $\RSmath{o} = \{s \in \mathcal{S}_e | 
	s \xrightarrow{*} s_p\}$. It trivially
    follows that $\sfomath{o} = s_0 \xrightarrow{+} s_t$.
\end{itemize}
We conclude: $\forall se \in SE :\forall t \in \mathcal{T}_{se} :\sessionmath{WFR}{se}{t}{e'} $.				
 
\vspace{5mm}

Before proving that the remaining session guarantees hold, we prove an intermediate result:
\begin{clm}
\label{bubu}
$\forall se' \in SE: \forall t_i \xrightarrow{se'} t_j: s_{t_i}
\xrightarrow{+} s_{t_j}$. Intuitively, all transactions commit in session order.
\end{clm}

\begin{proof}
We first prove this result for update transactions, then generalise it to
all transactions.
\par\textbf{Update Transactions} For a given session $se'$, let
$T_u$ be the set of all update transactions in $\mathcal{T}_{se'}$,
and let $t_j$ be an arbitrary transaction in $T_u$. It thus holds that
$t_j \in \mathcal{T}_{se'} \wedge \mathcal{W}_{t_j} \neq \emptyset$.
We associate with each such $t_j$ two further sets: $T_{{pre}_u}$
and $T_{{pre}_r}$. $T_{{pre}_u}$ contains all update transactions $t_i$
such that $t_i \xrightarrow{se'} t_j$. Similarly,  $T_{{pre}_r}$ contains all read-only transactions $t_i$ such that $t_i \xrightarrow{se'} t_j$.$T_{pre}$ is the
union of those two sets. We prove that $\forall
t_i \in T_{pre}: s_{t_i} \xrightarrow{+} s_{t_j}$. 
If $t_i \in T_{{pre}_u}$, hence $\mathcal{W}_i \neq \emptyset \wedge
\mathcal{W}_j \neq \emptyset$, the result trivially follows from monotonic writes. We previously proved that $\forall t \in \mathcal{T}_{se'} :\sessionmath{MW}{se'}{t}{e'}$. As such, the conjunction $\mathcal{W}_i \neq \emptyset \wedge
\mathcal{W}_j \neq \emptyset$ implies $s_{t_i} \xrightarrow{+} s_{t_j}$ in $e'$.

The proof is more complex if $t_i \in T_{{pre}_r}$ (read-only transaction). We proceed by induction:
\par\textbf{Base Case} Consider the first read-only transaction $t_i \in T_{{pre}_r}$ according to the session order $se'$. This transaction is unique
(sessions totally order transactions). Recall that we choose the parent state of a read-only transaction as $s_p(t_i) = \max\{\max_{o_i \in \Sigma_{t_i}}\{\sfomath{o_i}\}, s_{t_{i-1}}\}$, where $ t_{i-1}$ denotes the transaction that directly precedes $t_i$ in session $se'$ ($s_{t_{i-1}} = s_0$ if
$t_i$ is the first transaction in the session). Hence, $t_i$'s parent state is either $s_p(t_i) = \max_{o_i \in \Sigma_{t_i}}\{ \sfomath{o_i} \} $, or $s_p(t_i) = s_{t_{i-1}}$
\begin{enumerate}
\item If $s_p(t_i) = \max_{o_i \in \Sigma_{t_i}}\{ \sfomath{o_i} \}$:
We previously proved that $\forall t \in \mathcal{T}_{se} :\sessionmath{WFR}{se}{t}{e'}$.
It follows that $\forall o_i \in \Sigma_{t_i}: \sfomath{o_i} \xrightarrow{+} s_{t_j}$ in $e'$ and consequently, $\max_{o_i \in \Sigma_{t_i}}\{ \sfomath{o_i} \} \xrightarrow{+} s_{t_j}$ in $e'$. Given that $s_p(t_i) = \max_{o_i \in \Sigma_{t_i}}\{ \sfomath{o_i} \}$, the following then holds $s_p(t_i)  \xrightarrow{+} s_{t_j}$ in $e'$. Finally, we note that, by definition
(Definition~\ref{def:statetrans}), the parent state of
a transaction directly precede its commit state. We can thus rephrase the
aforementioned relationship as $s_p(t_i) \rightarrow s_{t_i} \xrightarrow{*} s_{t_j}$ in $e'$, concluding the proof for this subcase.
\item If $s_p({t_i}) = s_{t_{i-1}}$: We defined
$t_{i}$ to be the first read-only transaction in the session. Given that,
by construction $t_{i-1} \xrightarrow{se'} t_{i}$, $t_{i-1}$ is
necessarily an update transaction, where $\mathcal{W}_{t_{i-1}} \neq \emptyset$. Consider the pair of transactions $(t_{i-1}, t_j)$. The session order is transitive, hence
$t_{i-1} \xrightarrow{se'} t_j$ given that $t_{i-1} \xrightarrow{se'} t_{i}$
and $t_{i} \xrightarrow{se'} {t_j}$ both hold. By construction, we have
$\mathcal{W}_{t_{i-1}} \neq \emptyset \wedge \mathcal{W}_{t_j} \neq \emptyset$.
We previously proved that $\forall t \in \mathcal{T}_{se} :\sessionmath{MW}{se}{t}{e'}$.
Hence, if $\mathcal{W}_{t_{i-1}} \neq \emptyset \wedge \mathcal{W}_{t_j} \neq \emptyset$, it follows that $ s_{t_{i-1}} \xrightarrow{+} s_{t_j}$. As above,
we conclude that: $s_p(t_i)  \xrightarrow{+} s_{t_j}$ in $e'$,  and finally $s_p(t_i) \rightarrow s_{t_i} \xrightarrow{*} s_{t_j}$. 
 				
\end{enumerate}
To complete the base case, we note that $t_i \neq t_j$ as $t_i \xrightarrow{se'} t_j$. We conclude:  $s_{t_i} \xrightarrow{+} s_{t_j}$.

\par\textbf{Induction Step} Consider the k-th read-only transaction $t_i$ in $se'$ 
such that $t_i \xrightarrow{se'} t_j$. We assume that it satisfies 
the induction hypothesis $s_{t_i} \xrightarrow{+} s_{t_j}$. Now
consider the (k+1)-th read-only transaction $t_{i'}$ in $se'$, such that $t_{i'} \xrightarrow{se'} t_{j}$. By construction, we once again distinguish two cases:
$t_{i'}$'s parent state is either $s_p(t_{i'}) = \max_{o_{i'} \in \Sigma_{t_i'}}\{ \sfomath{o_i'} \} $, or $s_p({t_{i'}}) = s_{t_{{i'-1}}}$, where $t_{i'-1}$ denotes the transaction directly preceding $t_{i'}$ in a session.
%
\begin{enumerate}
	\item If $s_p(t_{i'}) = \max_{o_{i'} \in \Sigma_{t_{i'}}}\{ \sfomath{o_i'} \}$:
We previously proved that $\forall t \in \mathcal{T}_{se} :\sessionmath{WFR}{se}{t}{e'}$.
It follows that $\forall o_{i'} \in \Sigma_{t_{i'}}: \sfomath{o_{i'}} \xrightarrow{+} s_{t_j}$ in $e'$ and consequently, $\max_{o_{i'}} \in \Sigma_{t_{i'}}\{ \sfomath{o_{i'}} \} \xrightarrow{+} s_{t_j}$ in $e'$. Given that $s_p(t_{i'}) = \max_{o_{i'} \in \Sigma_{t_{i'}}}\{ \sfomath{o_{i'}} \}$, the following then holds $s_p(t_{i'})  \xrightarrow{+} s_{t_j}$ in $e'$. Finally, we note that, by definition
(Definition~\ref{def:statetrans}), the parent state of
a transaction must directly precede its commit state. We can thus rephrase the
aforementioned relationship as $s_p(t_{i'}) \rightarrow s_{t_{i'}} \xrightarrow{*} s_{t_j}$ in $e'$, concluding the proof for this subcase.

					\item If $s_p{(t_{i'})} = s_{t_{i'-1}}$:
First, we note that $t_{i'-1} \xrightarrow{se'} t_j$ holds, as the session order is transitive and we have both $t_{i'-1} \xrightarrow{se'} t_{i'}$ and $t_{i'} \xrightarrow{se'} {t_j}$.
We then distinguish between two cases: $t_{i'-1}$ is an update transaction,
and $t_{i'-1}$ is a read only transaction. If $t_{i'-1}$ is an
update transaction, the following conjunction holds:
$\mathcal{W}_{t_{i'-1}} \neq \emptyset \wedge \mathcal{W}_{t_j} \neq \emptyset$.
Given that we previously proved $\forall t \in \mathcal{T}_{se} :\sessionmath{MW}{se}{t}{e'}$
, we can infer that $s_{t_{i'-1}} \xrightarrow{+} s_{t_j}$, i.e. $s_p(t_{i'})  \xrightarrow{+} s_{t_j}$ in $e'$. We again note that by definition
(Definition~\ref{def:statetrans}), the parent state of
a transaction must directly precede its commit state. We can thus rephrase the
aforementioned relationship as $s_p(t_{i'}) \rightarrow s_{t_{i'}} \xrightarrow{*} s_{t_j}$ in $e'$. 
We now consider the case where $t_{i'-1}$  is a
read-only transaction. If $t_{i'}$ is the k+1-th read-only transaction, 
then, by construction $t_{i'-1}$ is the k-th read-only transaction. Hence
$t_i = t_{i'-1}=s_p(t_{i'})$. Our induction hypothesis states that $s_{t_{i}} \xrightarrow{+} s_{t_j}$. It thus follows that $s_p(t_{i'})  \xrightarrow{+} s_{t_j}$ in $e'$. 
As previously, we conclude that: $s_p(t_i) \rightarrow s_{t_i} \xrightarrow{*} s_{t_j}$ in $e'$. 
\end{enumerate}
To complete the induction step, we note that $t_i' \neq t_j$ as $t_i' \xrightarrow{se'} t_j$. We conclude:  $s_{t_i'} \xrightarrow{+} s_{t_j}$.

We proved the desired result for both the base case and induction step. By induction,
we conclude that:  given any $t_j$ such that $\mathcal{W}_{t_j} \neq \emptyset$,and any read-only transactions $t_i$ such that $t_i \xrightarrow{se'} t_j$, $s_{t_i} \xrightarrow{+}  s_{t_j}$ holds.

We conclude that given any $t_j$ such that $\mathcal{W}_{t_j} \neq
\emptyset$,  and any transaction $t_i$ such that $t_i
\xrightarrow{se'} t_j$, $s_{t_i} \xrightarrow{+}  s_{t_j}$ holds.
	
\par\textbf{Read-Only Transactions} We now generalise the result to both update
and read-only transactions. Specifically, we prove that in $e'$, $\forall t_i \xrightarrow{se'} t_j: s_{t_i} \xrightarrow{+} s_{t_j}$. We first prove this
statement for any two consecutive transactions in a session, and then extend it to all transactions in a session. Consider any two pair of transactions $t_i$,$t_{i-1}$ in $\mathcal{T}_{se'}$ such that $t_{i-1}$ directly precede $t_i$ in $se'$
(${t_{i-1}} \xrightarrow{se'} {t_{i}}$). If $\mathcal{W}_{t_i} \neq \emptyset$,
$s_{t_{i-1}} \xrightarrow{+} s_{t_i}$ as proven above. If ${t_i}$ is a read-only transaction, its parent state, by construction, is equal to $s_p(t_i) =\max\{\max_{o_i \in \Sigma_{t_i}}\{ \sfomath{o_i} \},s_{t_{i-1}} \} $. Since $ \max\{\max_{o_i \in \Sigma_{t_i}}\{ \sfomath{o_i} \},s_{t_{i-1}} \} \geq s_{t_{i-1}}$ by definition, it follows that $s_{t_{i-1}} \xrightarrow{*}  s_p(t_i) $.
As $s_p(t_i) \rightarrow s_{t_i}$, it follows that $s_{t_{i-1}} \xrightarrow{+} s_{t_i}$.
Together each such pair of consecutive transactions $t_{i-1}, t_i$ defines a sequence: $t_1 \xrightarrow{se'} t_2 \xrightarrow{se'} \dots \xrightarrow{se'} t_{k} $, where $\mathcal{T}_{se'} = \{t_1,\dots,t_k\}$. From the implication derived in the previous paragraph, it follows that 
$s_{t_1} \xrightarrow{+} s_{t_2} \xrightarrow{+} \dots \xrightarrow{+} s_{t_k}$. Noting that session order is transitive, we conclude: $\forall t_i \xrightarrow{se'} t_j: s_{t_i} \xrightarrow{+} s_{t_j}$. 

This completes the proof of Claim~\ref{bubu}.
\end{proof}

\par\textbf{RMW} We now return to session guarantees and prove that
$\forall t \in \mathcal{T}_{se} : \sessionmath{RMW}{se}{t}{e'}$. Specifically,
we show that  $\prereadmath{e'}{\mathcal{T}_{se}} \wedge \forall o \in
\Sigma_t: \forall t'\xrightarrow{se} t:\mathcal{W}_{t'} \neq \emptyset
\Rightarrow s_{t'} \xrightarrow{*}sl_o$. 

We proceed to prove that each of the two clauses holds true.

By assumption,
$e$ guarantees read-my-writes: $\forall t \in \mathcal{T}_{se} :\sessionmath{RMW}{se}{t}{e}$. Consider an arbitrary
transaction $t$, and all update transactions $t'$ that precede $t$ in the session: $\forall o \in \Sigma_t: \forall t'\xrightarrow{se} t:\mathcal{W}_{t'} \neq \emptyset$. We distinguish between read operations and write operations:
\begin{itemize}
\item Let $o$ be a read operation $o = r(k,v)$.
Its candidate read set is $\mathcal{RS}_{e'}(o) = \{s \in \mathcal{S}_e |	s \xrightarrow{*} s_p(t) \wedge\big ( (k,v) \in s \vee (\exists w(k,v) \in \Sigma_t: w(k,v) \xrightarrow{to} r(k,v)) \big)\}$. $\slomath{o}$,
the last state in $\mathcal{RS}_{e'}(o)$ can have
one of two values: $\slomath{o} = s_p(t)$, disallowing
states created after $t$'s commit point, or $\slomath{o} = s_p(\hat{t})$, where $\hat{t}$ is the
update transaction that writes the next version of $k$.
\begin{itemize}
\item $\slomath{o} = s_p(t)$. We previously proved
that $\forall se' \in SE: \forall t_i \xrightarrow{se'} t_j: s_{t_i} \xrightarrow{+} s_{t_j} $. By construction, $t' \xrightarrow{se} t$. It
follows that $s_{t'} \xrightarrow{+} s_{t}$ and consequently that $s_{t'} \xrightarrow{*} s_p(t)$.
Given $s_p(t) = \slomath{o}$, we conclude: $s_{t'} \xrightarrow{*} \slomath{o}$.					
\item $\slomath{o} = s_p(\hat{t})$. Consider
first the relationships between read states
and commit states in $e$. By 
assumption, $e$ satisfies $\forall t \in \mathcal{T}_{se} :\sessionmath{RMW}{se}{t}{e} $, i.e. $s_{t'} \xrightarrow{*} \slomath{o}$ in $e$. Since $\hat{t}$ wrote the next version of the object that $t$ read,
we have that $\slomath{o} \xrightarrow{*} s_p(\hat{t}) \xrightarrow{+} s_{\hat{t}}$ in $e$. Combining the guarantee given by
read-my-writes $s_{t'} \xrightarrow{*} \slomath{o}$ and $\slomath{o} \xrightarrow{+} s_{\hat{t}}$, we obtain $s_{t'} \xrightarrow{+} s_{\hat{t}}$ in $e$. Returning to the execution $e'$, since
$t'$ and $\hat{t}$ are both update transactions, $\mathcal{W}_{t'} \neq \emptyset \wedge \mathcal{W}_{\hat{t}} \neq \emptyset$, if $s_{t'} \xrightarrow{+} s_{\hat{t}}$ in $e$, then $s_{t'} \xrightarrow{+} s_{\hat{t}}$ in $e'$. Given $e'$ is a total
order and $s_p({\hat{t}}) \rightarrow s_{\hat{t}}$,
we conclude $s_{t'} \xrightarrow{*} s_p(\hat{t})$,  and $s_{t'} \xrightarrow{*} \slomath{o}$.
\end{itemize}
\item Let $o=w(k,v)$ be a write operation. By Claim~\ref{bubu}, it holds that $\forall se' \in SE: \forall t_i \xrightarrow{se'} t_j: s_{t_i} \xrightarrow{+} s_{t_j} $. As $t' \xrightarrow{se} t$, it follows that $s_{t'} \xrightarrow{+} s_{t}$ 
and consequently that $s_{t'} \xrightarrow{*} s_p(t)$. We conclude: $s_{t'} \xrightarrow{*} \slomath{o}$.
\end{itemize}
Finally, as $\prereadmath{e'}{\mathcal{T}}$ 
and $\mathcal{T}_{se} \subseteq \mathcal{T}$, by
Lemma~\ref{prereadlemma}, $\prereadmath{e'}{\mathcal{T}_{se}}$ holds.

 We conclude that $\prereadmath{e'}{\mathcal{T}_{se}} \wedge \forall o \in \Sigma_t: \forall t'\xrightarrow{se} t:\mathcal{W}_{t'} \neq \emptyset \Rightarrow s_{t'} \xrightarrow{*}\sfomath{o}$, i.e. $\forall t \in \mathcal{T}_{se} :\sessionmath{RMW}{se}{t}{e'}$.
 				
\par\textbf{MR} Finally, we prove that $e'$ satisfies the
final session guarantee: $\sessionmath{MR}{se}{t}{e'}$. Specifically,
we show that: \newline
$\prereadmath{e'}{\mathcal{T}_{se}} \wedge \ircmath{e'}{t} \wedge
\forall o \in \Sigma_t: \forall t'\xrightarrow{se}t: \forall o' \in
\Sigma_{t'}: sf_{o'}\xrightarrow{*}sl_o$.
Intuitively, this states that the read state of $o$ must include any
write seen by $o'$.

We proceed to prove that each of the three clauses holds true. 

The first clause follows directly from Lemma~\ref{prereadlemma} and
the fact that $\mathcal{T}_{se} \subseteq \mathcal{T}$. We can then
conclude that $\sfomath{o},\sfomath{o'}, \slomath{o}, \slomath{o'}$
must exist.

We can now proceed to prove that the third clause holds. We consider
two cases, depending on whether $o'$ is a read or a write operation.

\begin{addmargin}[5mm]{0pt}
\textbf{Read} The read operation $o' = r(k',v')$ entails the
existence of an update transaction $\hat{t'} \in \mathcal{T}$ that
writes version $v'$ of object $k'$, i.e $\sfomath{o'} = s_{\hat{t'}}$,
$k \in \mathcal{W}_{\hat{t'}}$. Now, $o$ can be either a read or a
write operation. 
\end{addmargin}
\begin{itemize}
\item Let us first assume that $o$ is a read operation $o=r(k,v)$. Its
  candidate read set is
  $\mathcal{RS}_{e'}(o) = \{s \in \mathcal{S}_{e'} | s \xrightarrow{*}
  s_p(t) \wedge\big ( (k,v) \in s \vee (\exists w(k,v) \in \Sigma_t:
  w(k,v) \xrightarrow{to} r(k,v)) \big)\}$.
  $\slomath{o}$, the last state in $\mathcal{RS}_{e'}(o)$ can be one of two cases $\slomath{o} = s_p(t)$, disallowing states created
  after $t$'s commit point, or $\slomath{o} = s_p(\hat{t})$, where
  $\hat{t}$ is the update transaction that writes the next version of
  $k$.
\begin{itemize}
\item $\slomath{o} = s_p(t)$. We previously proved
that $\forall se' \in SE: \forall t_i \xrightarrow{se'} t_j: s_{t_i} \xrightarrow{+} s_{t_j} $. By construction, $t' \xrightarrow{se} t$. It
follows that $s_{t'} \xrightarrow{+} s_{t}$ and consequently that $s_{t'} \xrightarrow{*} s_p(t)$.
Given $s_p(t) = \slomath{o}$, we conclude: $s_{t'} \xrightarrow{*} \slomath{o}$. Moreover, Lemma~\ref{sfolemma} states that given $\prereadmath{e'}{\mathcal{T}_{se}}$, we have $\sfomath{o'} \xrightarrow{+} s_{t'}$ in $e'$, hence: $\sfomath{o'} \xrightarrow{*} \slomath{o}$ in $e'$.
\item $\slomath{o} = s_p(\hat{t})$. Consider
first the relationships between read states
and commit states in $e$. By 
assumption, $e$ satisfies $\forall t \in \mathcal{T}_{se} :\sessionmath{MR}{se}{t}{e} $, i.e. $\sfomath{o'} \xrightarrow{*} \slomath{o}$ in $e$. Since $\hat{t}$ wrote the next version of the object that $t$ read,
we have that in $\slomath{o} \xrightarrow{+} s_{\hat{t}}$ in $e$. Combining the guarantee given by monotonic reads $\sfomath{o'} \xrightarrow{*} \slomath{o}$ and $\slomath{o} \xrightarrow{+} s_{\hat{t}}$, we obtain  
$\sfomath{o'} \xrightarrow{+} s_{\hat{t}}$ in $e$, i.e. $s_{\hat{t'}} \xrightarrow{+} s_{\hat{t}}$.
Returning to the execution $e'$, since
$\hat{t'}$ and $\hat{t}$ are both update transactions, $\mathcal{W}_{\hat{t'}} \neq \emptyset \wedge \mathcal{W}_{\hat{t}} \neq \emptyset$; if $s_{\hat{t'}} \xrightarrow{+} s_{\hat{t}}$ in $e$, then $s_{\hat{t'}} \xrightarrow{+} s_{\hat{t}}$ in $e'$. By definition
$\sfomath{o'} \xrightarrow{+} s_{t'}$. Moreover, by assumption, $\slomath{o} = s_p(\hat{t})$. Putting
this together, we obtain the desired result 
$\sfomath{o'} \xrightarrow{*} \slomath{o}$ in $e'$.
\end{itemize}
\item Let $o=w(k,v)$ be a write operation. The write set of a write operation is defined as  $\mathcal{RS}_{e'}(o) = \{s \in \mathcal{S}_e | 
 							s \xrightarrow{*} s_p\}$. It
                        follows that: $\slomath{o} = s_p(t)$. 
We previously proved that in $e'$, $\forall se' \in SE: \forall t_i \xrightarrow{se'} t_j: s_{t_i} \xrightarrow{+} s_{t_j} $: given
$t' \xrightarrow{se} t$, it
thus follows that $s_{t'} \xrightarrow{+} s_{t}$,
and consequently that $s_{t'} \xrightarrow{*} s_p(t)$. Noting
that $s_p(t) = \slomath{o}$ , we write $s_{t'} \xrightarrow{*} \slomath{o}$. Moreover,
as $\prereadmath{e'}{\mathcal{T}_{se}}$ holds for $e'$, by Lemma \ref{sfolemma}, we have $\sfomath{o'} \xrightarrow{+} s_{t'}$ in $e'$. Combining
the relationships, we conclude: $\sfomath{o'} \xrightarrow{+} \slomath{o}$ in $e'$.
\end{itemize} 

\begin{addmargin}[5mm]{0pt}
\textbf{Write} The candidate read state set for a write operation
$o'=w(k,v)$ is defined as the set of all states before $t'$'s commit
state. Hence $\sfomath{o'} = s_0$. Thus $\sfomath{o'} \xrightarrow{*}
\slomath{o}$ trivially holds.
\end{addmargin}

We can now prove that the second clause, $\ircmath{e'}{t}$,
holds---namely, that
$\forall o, \, o' \in \Sigma_t: o'\xrightarrow{to} o \Rightarrow
\sfomath{o'}\xrightarrow{*}\slomath{o}$. Once again we consider two
cases, depending on whether $o'$ is a read or a write operation.

\begin{addmargin}[5mm]{0pt}
\textbf{Read} The presumpted read operation $o' = r(k',v')$ entails the
existence of an update transaction $\hat{t'} \in \mathcal{T}$ that writes version $v'$ of object $k'$, i.e $\sfomath{o'} = s_{\hat{t'}}$, $k \in \mathcal{W}_{\hat{t'}}$. 
\begin{itemize}
\item Let us first assume that  
$o$ is a read operation $o=r(k,v)$, Its candidate read set $\mathcal{RS}_{e'}(o) = \{s \in \mathcal{S}_{e'} |	s \xrightarrow{*} s_p(t) \wedge\big ( (k,v) \in s \vee (\exists w(k,v) \in \Sigma_t: w(k,v) \xrightarrow{to} r(k,v)) \big)\}$. $\slomath{o}$,
the last state in $\mathcal{RS}_{e'}(o)$ can have
one of two values: $\slomath{o} = s_p(t)$, disallowing
states created after $t$'s commit point, or $\slomath{o} = s_p(\hat{t})$, where $\hat{t}$ is the
update transaction that writes the next version of $k$.
\begin{itemize}
\item $\slomath{o} = s_p(t)$. We previously
showed that $\prereadmath{e'}{\mathcal{T}}$. Given that $o'$, like $o$ is in $\Sigma_{t}$, it follows
by Lemma \ref{sfolemma} that $\sfomath{o'} \xrightarrow{+} s_{t}$ in $e'$, and consequently
that $\sfomath{o'} \xrightarrow{*} s_p(t) $. 
Setting  $s_p(t)$ to $\slomath{o}$, we conclude $\sfomath{o'} \xrightarrow{*} \slomath{o} $ in $e'$. 
\item $\slomath{o} = s_p(\hat{t})$: Consider first the relationships between read states and commit states in $e$. By assumption,
$e$ satisfies $\ircmath{e}{t}$, i.e. 
$\sfomath{o'} \xrightarrow{*} \slomath{o}$ holds in $e$.  Since $\hat{t}$ wrote the next version of the object that $o$ read, 
 we have that $\slomath{o} \xrightarrow{*} s_p(\hat{t}) \xrightarrow{+} s_{\hat{t}}$ in $e$. Combining the guarantee given by
monotonic reads $\sfomath{o'} \xrightarrow{*} \slomath{o}$ and $\slomath{o} \xrightarrow{+} s_{\hat{t}}$, it follows that $\sfomath{o'} \xrightarrow{+} s_{\hat{t}}$ in $e$ i.e. $s_{\hat{t'}} \xrightarrow{+} s_{\hat{t}}$.
Returning to the execution $e'$, since
$\hat{t'}$ and $\hat{t}$ are both update transactions, $\mathcal{W}_{\hat{t'}} \neq \emptyset \wedge \mathcal{W}_{\hat{t}} \neq \emptyset$, if $s_{\hat{t'}} \xrightarrow{+} s_{\hat{t}}$ in $e$, then $s_{\hat{t'}} \xrightarrow{+} s_{\hat{t}}$ and consequently $\sfomath{o'} \xrightarrow{+} s_{\hat{t}}$ . Given $e'$ is a total
order and $s_p({\hat{t}}) \rightarrow s_{\hat{t}}$,
we conclude $\sfomath{o'} \xrightarrow{*} s_p(\hat{t})$,  and $\sfomath{o'} \xrightarrow{*} \slomath{o}$ in $e'$, as desired.
\end{itemize}
\item Let us next assume that
$o$ is a write operation. The 
candidate read set of a write operation $o=w(k,v)$ is  $\mathcal{RS}_{e'}(o) = \{s \in \mathcal{S}_e | 
s \xrightarrow{*} s_p\}$, where, consequently, $\slomath{o} = s_p(t)$.
We previously
showed that $\prereadmath{e'}{\mathcal{T}}$. Given that $o'$, like $o$ is in $\Sigma_{t}$, it follows
by Lemma \ref{sfolemma} that $\sfomath{o'} \xrightarrow{+} s_{t}$ in $e'$, and consequently
that $\sfomath{o'} \xrightarrow{*} s_p(t) $. 
Setting  $_p(t)$ to $\slomath{o}$, we conclude $\sfomath{o'} \xrightarrow{*} \slomath{o} $ in $e'$. 
\end{itemize}
\end{addmargin}
\begin{addmargin}[5mm]{0pt}
\textbf{Write} The candidate read state set for a write operation $o'=w(k,v)$ is defined as the set of all states before $t'$'s commit state. Hence $\sfomath{o'} = s_0$. Thus $\sfomath{o'} \xrightarrow{*} \slomath{o}$ trivially holds.
\end{addmargin}

We conclude that  $\ircmath{e'}{t}$ holds.

This completes the proof that Monotonic Reads holds for execution
$e'$.  This was the last outstanding session guarantees: we have then
proved that  $e'$ satisfies all four session
guarantees.  
\end{addmargin}

\vspace{3mm}

We now proceed to prove that $e'$ satisfies causal consistency:
$\forall t \in \mathcal{T}_{se} :\sessionmath{cc}{se}{t}{e'}$.  More
specifically, we must prove that:
$\prereadmath{e'}{\mathcal{T}} \wedge \ircmath{e'}{t} \wedge (\forall o
\in \Sigma_t: \forall t' \xrightarrow{se} t:  s_{t'}
\xrightarrow{*}\slomath{o}) \wedge (\forall se' \in SE: \forall t_i
\xrightarrow{se'} t_j: s_{t_i} \xrightarrow{+} s_{t_j})$.  We proceed
by proving 
that each of the clauses holds.

\vspace{5mm}

The first two clauses are easy to establish. We previously proved that $\prereadmath{e'}{\mathcal{T}}$ holds. Likewise,
$\ircmath{e'}{t}$ holds as $\forall t \in \mathcal{T}_{se}
:\sessionmath{MR}{se}{t}{e'}$. To establish the fourth clause, we note
that  
we previously proved that $\forall se' \in SE: \forall t_i \xrightarrow{se'} t_j: s_{t_i} \xrightarrow{+} s_{t_j} $.

We are then left to prove only the third clause; namely we must prove
that $\forall t_j \in \mathcal{T}_{se} : (\forall o_j \in \Sigma_{t_j}:
\forall t_i \xrightarrow{se} t_j:  s_{t_i} \xrightarrow{*}
\slomath{o_j}) $. 

We distinguish between two cases: $t_i$ is an update transaction and $t_i$ is a read-only transaction. If $t_i$ is an update transaction, the desired result
holds as $e'$ guarantees read-my-writes: $\forall t_j \in \mathcal{T}_{se} : (\forall o_j \in \Sigma_{t_j}: \forall t_i \xrightarrow{se} t_j:  \mathcal{W}_{t_i} \neq \emptyset \Rightarrow s_{t_i} \xrightarrow{*} \slomath{o_j})$ 

If $t_i$ is a read-only transaction, we proceed by induction. We consider an arbitrary $t_j$, and an arbitrary $o_j \in \Sigma_{t_j}$. 
\par\textbf{Base Case}  Consider the first read-only transaction $t_i$ in $se'$ 
such that $t_i \xrightarrow{se'} t_j$. Recall that we choose the parent state of a read-only transaction as $s_p(t_i) = \max\{\max_{o_i \in \Sigma_{t_i}}\{\sfomath{o_i}\}, s_{t_{i-1}}\}$, where $ t_{i-1}$ denotes the transaction that directly precedes $t_i$ in session $se'$ ($s_{t_{i-1}} = s_0$ if
$t_i$ is the first transaction in the session). Hence, $t_i$'s parent state is either $s_p(t_i) = \max_{o_i \in \Sigma_{t_i}}\{ \sfomath{o_i} \} $, or $s_p(t_i) = s_{t_{i-1}}$.

\begin{itemize}
\item If $s_p(t_i) = \max_{o_i \in \Sigma_{t_i}}\{ \sfomath{o_i} \} $: We previously proved that $\forall t \in \mathcal{T}_{se} :\sessionmath{MR}{se}{t}{e'}$,
hence that $\forall t_i \xrightarrow{se} t_j :\forall o_i \in \Sigma_{t_i}:  \sfomath{o_i} \xrightarrow{*} \slomath{o_j}$, and consequently $\max_{o_i \in \Sigma_{t_i}} \{\sfomath{o_i} \} \xrightarrow{*} \slomath{o_j}$ in $e'$. Noting
that $s_p(t_i) = \max_{o_i \in \Sigma_{t_i}}\{ \sfomath{o_i} \}$, we obtain the desired result
$s_p({t_i}) \xrightarrow{*} \slomath{o_j}$ in $e'$.  
		\item If $s_p({t_i}) = s_{t_{i-1}}$. We defined
$t_{i}$ to be the first read-only transaction in the session. 
By construction $t_{i-1} \xrightarrow{se'} t_{i}$, $t_{i-1}$ is
necessarily an update transaction given $t_i$ is the first read-only transaction, where $\mathcal{W}_{t_{i-1}} \neq \emptyset$.
Given that, by transitivity $t_{i-1} \xrightarrow{se} t_j$, and that $e'$ guarantees read-my-writes 
$\forall t \in \mathcal{T}_{se} :\sessionmath{RMW}{se}{t}{e'}$,
we have $ s_{t'_{i-1}} \xrightarrow{*} \slomath{o_j}$ in $e'$. Noting that $s_p({t_i}) = s_{t_{i-1}}$,
we conclude $ s_p({t_i}) \xrightarrow{*} \slomath{o_j}$ in $e'$. 
\end{itemize}
Finally, we argue that $s_p(t_i) \neq \slomath{o_j}$
(and therefore that $s_{t_i} \xrightarrow{*} \slomath{o_j}$ as $s_p(t_i) \xrightarrow s_{t_i}$). Read-only transactions, like $t_i$ do not change the state on which they are applied, hence $\forall (k,v) \in s_p(t_i)  \Rightarrow (k,v) \in s_{t_i}$.
Moreover, by Claim~\ref{bubu}, transactions commit in session order: $\forall se' \in SE: \forall t_i \xrightarrow{se'} t_j: s_{t_i} \xrightarrow{+} s_{t_j}$.  We thus have $s_{t_i} \xrightarrow{+} s_{t_j}$ and consequently $s_p(t_i) \in \mathcal{RS}_{e'}(o_j) \Rightarrow s_{t_i} \in \mathcal{RS}_{e'}(o_j)$, i.e. $s_p(t_i) \neq \slomath{o_j}$: if $s_p(t_i)$ is 
in $\mathcal{RS}_{e'}(o_j)$, so is $s_{t_i}$. 
As $s_{t_i}$ follows $s_p(t_i)$ in the execution,
$s_p(t_i)$ will never be $\sfomath{o_j}$. 
The following thus holds in the base case: $s_p({t_i}) \xrightarrow{+} \slomath{o_j}$, i.e. $s_p({t_i}) \rightarrow s_{t_i}\xrightarrow{*} \slomath{o_j}$.
			
\par\textbf{Induction Step} 
Consider the k-th read-only transaction $t_i$ in $se'$ 
such that $t_i \xrightarrow{se'} t_j$. We assume that it satisfies 
the induction hypothesis $s_{t_i} \xrightarrow{*}\slomath{o_j}$. Now
consider the (k+1)-th read-only transaction $t_{i'}$ in $se'$, such that $t_{i} \xrightarrow{se'} t_{i'}$. By construction, we once again distinguish two cases:
$t_{i'}$'s parent state is either $s_p(t_{i'}) = \max_{o_{i'} \in \Sigma_{t_i'}}\{ \sfomath{o_i'} \} $, or $s_p({t_{i'}}) = s_{t_{{i'-1}}}$, where $t_{i'-1}$ denotes the transaction directly preceding $t_{i'}$ in a session.
\begin{enumerate}
\item If $s_p(t_{i'}) = \max_{o_{i'} \in \Sigma_{t_{i'}}}\{ \sfomath{o_{i'}} \} $. We previously proved that $\forall t \in \mathcal{T}_{se} :\sessionmath{MR}{se}{t}{e'}$,
hence that $\forall t_{i'} \xrightarrow{se} t_j :\forall o_{i'} \in \Sigma_{t_{i'}}:  \sfomath{o_{i'}} \xrightarrow{*} \slomath{o_j}$, and consequently $\max_{o_{i'} \in \Sigma_{t_{i'}}} \{\sfomath{o_{i'}} \} \xrightarrow{*} \slomath{o_j}$ in $e'$. Noting
that $s_p(t_{i'}) = \max_{o_{i'} \in \Sigma_{t_{i'}}}\{ \sfomath{o_{i'}} \}$, we obtain the desired result
$s_p({t_{i'}}) \xrightarrow{*} \slomath{o_j}$ in $e'$.  
\item If $s_p({t_{i'}}) = s_{t_{i'-1}}$:
First, we note that $t_{i'-1} \xrightarrow{se'} t_j$ holds, as the session order is transitive and we have both $t_{i'-1} \xrightarrow{se'} t_{i'}$ and $t_{i'} \xrightarrow{se'} {j}$.
We distinguish between two cases:
if $t_{i'-1}$ is a read-only transaction, then it must be the k-th such transaction (as, by construction, 
it directly precedes $t_{i'}$ in the session).
Hence $t_{i'-1} = t_i$. Our induction hypothesis
states that $s_{t_i} \xrightarrow{*} \slomath{o_j}$,
and consequently that $s_{t_{i'-1}} \xrightarrow{*} \slomath{o_j}$. Noting that  $s_p({t_{i'}}) = s_{t_{i'-1}}$, we obtain $s_p(t_{i'}) \xrightarrow{*} \slomath{o_j}$. If $t_{i'-1}$ is an
update transaction, we note that $e'$ guarantee
read-my-writes: $\forall t \in \mathcal{T}_{se} :\sessionmath{RMW}{se}{t}{e'}$. As $t_{i'-1} \xrightarrow{se} t_j$, we have $s_{t_{i'-1}} \xrightarrow{*} \slomath{o_j}$ in $e'$. Noting that $s_p({t_{i'}}) = s_{t_{i'-1}}$, we conclude: $s_p(t_{i'}) \xrightarrow{*} \slomath{o_j}$ in $e'$.				
\end{enumerate}
Finally, we argue that $s_p(t_{i'}) \neq \slomath{o_j}$
(and therefore that $s_{t_{i'}} \xrightarrow{*} \slomath{o_j}$ as $s_p(t_{i'}) \rightarrow s_{t_{i'}}$). Read-only transactions, like $t_{i'}$ do not change the state on which they are applied, hence $\forall (k,v) \in s_p(t_{i'})  \Rightarrow (k,v) \in s_{t_{i'}}$.
Moreover, by Claim~\ref{bubu}, transactions commit in session order: $\forall se' \in SE: \forall t_{i'} \xrightarrow{se'} t_j: s_{t_{i'}} \xrightarrow{+} s_{t_j}$.  We thus have $s_{t_{i'}} \xrightarrow{+} s_{t_j}$ and consequently $s_p(t_{i'}) \in \mathcal{RS}_{e'}(o_j) \Rightarrow s_{t_{i'}} \in \mathcal{RS}_{e'}(o_j)$, i.e. $s_p(t_{i'}) \neq \slomath{o_j}$: if $s_p(t_{i'})$ is 
in $\mathcal{RS}_{e'}(o_j)$, so is $s_{t_{i'}}$. 
As $s_{t_{i'}}$ succedes $s_p(t_{i'})$ in the execution,
$s_p(t_{i'})$ will never be $\slomath{o_j}$. 
The following thus holds in the induction case: $s_p({t_{i'}}) \xrightarrow{+} \slomath{o_j}$, i.e. $s_p({t_{i'}}) \rightarrow s_{t_{i'}}\xrightarrow{*} \slomath{o_j}$.
We proved the desired result for both the base case and induction step. By induction,
we conclude that, for read-only transactions: 
$\forall t_j \in \mathcal{T}_{se}$, $\forall o_j \in \Sigma_{t_j}: \forall t_i \xrightarrow{se} t_j, \mathcal{W}_{t_i} =\emptyset \Rightarrow s_{t_i} \xrightarrow{*}\slomath{o_j}$.
Hence, the desired result holds for both read-only and update transactions $\forall t_j \in \mathcal{T}_{se} : (\forall o_j \in \Sigma_{t_j}: \forall t_i \xrightarrow{se} t_j:  s_{t_i} \xrightarrow{*} \slomath{o_j}) $. 
	
\par \textbf{Conclusion} Putting everything together, if $e$ guarantees
all four session guarantees, there exists an equivalent execution $e'$ such that $e'$ also satisfies the session guarantees and is causally consistent: $ \prereadmath{e}{\mathcal{T}} \wedge \ircmath{e}{t} \wedge (\forall o \in \Sigma_t: \forall t' \xrightarrow{se} t:  s_{t'} \xrightarrow{*} sl_o)  \wedge (\forall se' \in SE: \forall t_i \xrightarrow{se'} t_j: s_{t_i} \xrightarrow{+} s_{t_j})$, i.e. $\forall t \in \mathcal{T}_{se}
:\sessionmath{CC}{se}{t}{e'}$. This completes the second part of the proof. 	

Consequently, $\forall se \in SE :\exists e: \forall t \in \mathcal{T}_{se} :\sessionmath{\mathcal{G}}{se}{t}{e} 
    \equiv \forall se \in SE : \exists e :  \forall t \in \mathcal{T}_{se}: \sessionmath{CC}{se}{t}{e}$
holds.

\section{Equivalence of PL-2+ and PSI}
\label{appendix:psi}

In this section, we prove that our state-based definition of PSI is equivalent to both the axiomatic formulation of PSI ($PSI_A$)
by Cerone et al. and to the 
cycle-based specification of PL-2+:
\par\textbf{Theorem~\ref{theorem:2pl}} Let $\mathcal{I}$ be PSI. Then $\exists e:\forall t \in \mathcal{T}:\commitmath{PSI}{t}{e} \equiv \lnot G1 \land \lnot G\text{-single}$ 
\par \textbf{Theorem~\ref{theorem:psi}} Let $\mathcal{I}$ be PSI. Then $\exists e:\forall t \in \mathcal{T}:\commitmath{PSI}{t}{e} \equiv PSI_A$

Before beginning, we first prove a useful lemma: if an execution
$e$, written $s_0 \rightarrow s_{t_1} \rightarrow s_{t_2} \rightarrow \dots \rightarrow s_{t_n}$ satisfies the predicate $\prereadmath{e}{\mathcal{T}}$, then any transaction $t$ that depends on a transaction $t'$ will always commit after $t'$
and all its dependents in the execution. We do so in two steps: we first
prove that $t$ will commit after the transactions that it directly reads from (Lemma~\ref{lemma:preddirect}), and then extend that result to all the transaction's
transitive dependencies (Lemma~\ref{lemma:predtrans}). Formally 

\begin{lemma}
	$\prereadmath{e}{\mathcal{T}} \Rightarrow \forall \hat{t} \in \mathcal{T}:\forall t \in \ddependsmath{e}{\hat{t}}, s_t \xrightarrow{+} s_{\hat{t}}$
	\label{lemma:preddirect}
\end{lemma}

\begin{proof}
	Consider any $\hat{t} \in \mathcal{T}$ and any $t \in \ddependsmath{e}{\hat{t}}$. $t$ is included in $\ddependsmath{e}{\hat{t}}$ if one
    of two cases hold: if $\exists o \in \Sigma_{\hat{t}}, t=t_{{sf_o}}$ ($\hat{t}$ reads the value created by $t$) or $s_{t} \xrightarrow{+}s_{\hat{t}} \wedge \mathcal{W}_{\hat{t}} \cap \mathcal{W}_{t} \neq \emptyset $ (t and $\hat{t}$ write the same objects and $t$ commits before $\hat{t}$). 
\begin{enumerate}
        \item $t \in \{t| \exists o \in \Sigma_{\hat{t}}: t=t_{{sf_o}}\}$
        Let $o_i$ be the operation such that  $t = t_{\sfomath{o_i}}$. By assumption, we have             $\prereadmath{e}{\mathcal{T}}$
        It follows that $\forall o, \sfomath{o} \xrightarrow{+} s_{\hat{t}}$.
        and consequently that $\sfomath{o_i} \xrightarrow{+} s_{\hat{t}}$
        and $s_t \xrightarrow{+} s_{\hat{t}}$.	
		\item $t \in \{t|s_{t} \xrightarrow{+} s_{\hat{t}} \wedge \mathcal{W}_{\hat{t}} \cap \mathcal{W}_{t} \neq \emptyset \}$, trivially we have $s_t \xrightarrow{+} s_{\hat{t}}$.
	\end{enumerate}
\end{proof}

We now generalise the result to hold transitively. 
\begin{lemma}
	$\prereadmath{e}{\mathcal{T}} \Rightarrow \forall t' \in \dependsmath{e}{t}: s_{t'} \xrightarrow{+} s_t$.
	\label{lemma:predtrans}
\end{lemma}

\begin{proof}
We prove this implication by induction.
\par \textbf{Base Case} Consider the first transaction $t_1$ in the execution. We want to prove that for all transactions $t$ that precede $t_1$ in the execution
$s_t \xrightarrow{*} s_{t_1} : \forall t' \in \dependsmath{e}{t}: s_{t'} \xrightarrow{*} s_{t}$. As $t_1$ is the first transaction in the execution, 
$\ddependsmath{e}{t_1}=\emptyset$ and consequently $\dependsmath{e}{t} = \emptyset$. We see this by contradiction: assume there exists a transaction \\ $t \in \ddependsmath{e}{t_1}$, by implication $s_t \xrightarrow{+} s_{t_1}$ (Lemma~\ref{lemma:preddirect}), violating
our assumption that $t_1$ is the first transaction in the execution. Hence the desired result 
trivially holds. 
\par \textbf{Induction Step} Consider the i-th transaction in the execution. We assume that $\forall t \text{ s.t. } s_t \xrightarrow{*} s_i$ the property $\forall t' \in \dependsmath{e}{t}: s_{t'} \xrightarrow{*} s_{t}$ holds. In otherwords, we assume that the property holds
for the first \textit{i} transactions. We now prove that the property holds for
the first \textit{i+1} transactions, specifically, we show that $ \forall t' \in \dependsmath{e}{t_{i+1}}: s_{t'} \xrightarrow{*}s_{t_{i+1}}$. A transaction $t'$ belongs to $\dependsmath{e}{t_{i+1}}$ if one of two conditions holds: either 
$t' \in \ddependsmath{e}{t_{i+1}}$, or $ \exists t_k \in \mathcal{T}:t' \in \dependsmath{e}{t_k} \wedge t_k \in \ddependsmath{e}{t_{i+1}}$. We consider each in turn:
\begin{itemize}
\item If $t' \in \ddependsmath{e}{t_{i+1}}$: by Lemma~\ref{lemma:preddirect}, we have $s_{t'}  \xrightarrow{+} s_{t_{i+1}}$.
\item If $ \exists t_k \in \ddependsmath{e}{t_{i+1}}: t' \in \dependsmath{e}{t_k}$: As $t_k \in \ddependsmath{e}{t_{i+1}}$, by Lemma~\ref{lemma:preddirect}, we have $s_{t_k}  \xrightarrow{+} s_{t_{i+1}}$, i.e. $s_{t_k} \xrightarrow{*} s_{t_i}$ ($s_{t_i}$ directly precedes $s_{t_{i+1}}$ in $e$ by construction).
The induction hypothesis holds for every transaction that strictly precedes $t_{i+1}$ in $e$,
hence $\forall t_{k'} \in \dependsmath{e}{t_k}: s_{t_{k'}} \xrightarrow{+} s_{t_k}$.
As $t' \in \dependsmath{e}{t_k}$ by construction, it follows that $s_{t'} 
\xrightarrow{+} s_{t_k}$. Putting everything together, we have $s_{t'} \xrightarrow{+} s_{t_k} \xrightarrow{+} s_{t_{i+1}}$, and consequently $s_{t'} \xrightarrow{+} s_{t_{i+1}}$. This completes the induction step of the proof.
\end{itemize}
Combining the base case, and induction step, we conclude: 
	$\prereadmath{e}{\mathcal{T}} \Rightarrow \forall t' \in \dependsmath{e}{t}: s_{t'} \xrightarrow{+} s_t$.
\end{proof}

\subsection{PL-2+}
\par\textbf{Theorem~\ref{theorem:2pl}} Let $\mathcal{I}$ be PSI. Then $\exists e:\forall t \in \mathcal{T}:\commitmath{PSI}{t}{e} \equiv \lnot \gone \land \lnot \gs$ 

Let us consider a history $H$ that contains the same set of transactions $\mathcal{T}$ as $e$.  The version order for $H$
, denoted as $<<$, is instantiated as follows: given an execution $e$ and an object $x$, $ x_i << x_j$ if and only if $x \in \mathcal{W}_{t_i} \cap \mathcal{W}_{t_j} \land  s_{t_{i}} \xrightarrow{*} s_{t_{j}}$. We show that, if a transaction
$t$ is in the depend set of a transaction $t'$, then there exists a path of
write-read/write-write dependencies from $t$ to $t'$ in the $DSG(H)$. Formally:
\begin{lemma}
	$\prereadmath{e}{\mathcal{T}} \Rightarrow \forall t' \in \dependsmath{e}{t}: t' \xrightarrow{ww/wr}^+ t$ in $DSG(H)$.
	\label{lemma:predchain}
\end{lemma}
\begin{proof}
We improve this implication by induction.
\par \textbf{Base Case} Consider the first transaction $t_1$ in the execution. 
We want to prove that for all transactions $t$ that precede $t_1$ in the execution
$\forall t \in \mathcal{T}$ such that  $s_t \xrightarrow{*} s_{t_1}$,
the following holds: $\forall t' \in \dependsmath{e}{t}: t' \xrightarrow{ww/wr}^+ t$ in $DSG(H)$. As $t_1$ is the first transaction in the execution, 
$\ddependsmath{e}{t_1}=\emptyset$ and consequently $\dependsmath{e}{t} = \emptyset$. We see this by contradiction: assume there exists a transaction $t \in \ddependsmath{e}{t_1}$, by implication $s_t \xrightarrow{+} s_{t_1}$ (Lemma~\ref{lemma:preddirect}), violating
our assumption that $t_1$ is the first transaction in the execution. ence the implication trivially holds. 
\par \textbf{Induction Step} Consider the i-th transaction in the execution. We assume that $\forall t$, s.t. $s_t \xrightarrow{*} s_{t_i}$, $ \forall t' \in \dependsmath{e}{t}: t' \xrightarrow{ww/wr}^+ t$. In otherwords, we assume that the property holds
for the first \textit{i} transactions. We now prove that the property holds for
the first \textit{i+1} transactions, specifically, we show that $ \forall t' \in \dependsmath{e}{t_{i+1}}: t' \xrightarrow{ww/wr}^+ t_{i+1}$.
A transaction $t'$ belongs to $\dependsmath{e}{t_{i+1}}$ if one of two conditions holds: either 
$t' \in \ddependsmath{e}{t_{i+1}}$, or $ \exists t_k \in \mathcal{T}:t' \in \dependsmath{e}{t_k} \wedge t_k \in \ddependsmath{e}{t_{i+1}}$. We consider each in turn:
\begin{itemize}
\item If $t' \in \ddependsmath{e}{t_{i+1}}$:
There are two cases:  $t' \in \{t| \exists o \in \Sigma_{t_{i+1}}:t=t_{{sf_o}}\}$
or, $t' \in  \{t|s_{t} \xrightarrow{+} s_{t_{i+1}} \wedge \mathcal{W}_{t_{i+1}} \cap \mathcal{W}_{t} \neq \emptyset \}$. If  $t' \in \{t| \exists o \in \Sigma_{t_{i+1}}:t=t_{{sf_o}}\}$, $t_{i+1}$ reads the version of an object that $t'$
wrote, hence $t_{i+1}$ read-depends on $t'$, i.e. $t' \xrightarrow{wr} t$. 

If $t' \in  \{t|s_{t} \xrightarrow{+} s_{t_{i+1}} \wedge \mathcal{W}_{t_{i+1}} \cap \mathcal{W}_{t} \neq \emptyset \}$: trivially, $s_{t'} \xrightarrow{+} s_{t_{i+1}}$. Let $x$ be the key that is written
by $t$ and $t_{i+1}$: $x \in \mathcal{W}_{t_{i+1}} \cap \mathcal{W}_{t}$.
By construction, the history $H$'s version order for $x$ is $x_{t'} << x_{t_{i+1}}$. By definition of version order, there must therefore a chain of $ww$ edges between $t'$ and $t_{i+1}$ in $DSG(H)$, where all of the transactions in the chain write the
next version of $x$. Thus: $t' \xrightarrow{ww}^+ t_{i+1}$ \changebars{}{also} holds. 

\item  If $ \exists t_k: t' \in \dependsmath{e}{t_k} \wedge t_k \in \ddependsmath{e}{t_{i+1}}$. As $t_k \in \ddependsmath{e}{t_{i+1}}$,
we conclude , as above that $t_k \xrightarrow{ww/wr}^+ t_{i+1}$.
Moreover, by Lemma~\ref{lemma:preddirect}, we have $s_{t_k}  \xrightarrow{+} s_{t_{i+1}}$, i.e. $s_{t_k} \xrightarrow{*} s_{t_i}$ ($s_{t_i}$ directly precedes $s_{t_{i+1}}$ in $e$ by construction).
The induction hypothesis holds for every transaction that precedes $t_{i+1}$ in $e$,
hence $\forall t_{k'} \in \dependsmath{e}{t_k}$: $t_{k'} \xrightarrow{ww/wr}^+ t_{k}$.
Noting $t' \in \dependsmath{e}{t_k}$, we see that  $t' \xrightarrow{ww/wr}^+ t_k$. Putting everything together, we obtain $t' \xrightarrow{ww/wr}^+ t_k \xrightarrow{ww/wr}^+ t_{i+1}$, i.e. $t' \xrightarrow{ww/wr}^+  t_{i+1}$
 by transitivity.  
\end{itemize}
Combining the base case, and induction step, we conclude: $\forall t: \forall t' \in \dependsmath{e}{t}: t' \xrightarrow{ww/wr}^+ t$.
\end{proof}

\par {\textbf{Equivalence}} We now prove Theorem~\ref{theorem:2pl}.

\par{\textbf{Theorem~\ref{theorem:2pl}}} Let $\mathcal{I}$ be PSI. Then $\exists e:\forall t \in \mathcal{T}:\commitmath{PSI}{t}{e} \equiv \lnot \gone \land \lnot \gs$ 
\begin{proof}
Let us recall the definition of PSI's commit test: 
\[
\prereadmath{e}{\mathcal{T}} \wedge \forall o \in \Sigma_t: \forall t' \in \dependsmath{e}{t}: o.k \in \mathcal{W}_{t'} \Rightarrow s_{t'}  \xrightarrow{*} sl_o
\]
	($\Rightarrow$) \textbf{First  we prove} $\exists e:\forall t \in \mathcal{T}:\commitmath{PSI}{t}{e} \Rightarrow \lnot \gone \land \lnot \gs$.
	
	Let \textit{e} be an execution that $\forall t \in \mathcal{T}:\commitmath{PSI}{t}{e}$, and \textit{H} be a history for committed transactions
	$\mathcal{T}$.
	
	We first instantiate the version order for \textit{H}, denoted as $<<$, as follows: given an execution $e$ and an object $x$, $ x_i << x_j$ if and only if $x \in \mathcal{W}_{t_i} \cap \mathcal{W}_{t_j} \land  s_{t_{i}} \xrightarrow{*} s_{t_{j}} $. It follows that, for any two states such that $(x,x_i) \in T_m
	\wedge (x,x_j) \in T_n \Rightarrow s_{T_m} \xrightarrow{+} s_{T_n}$. 

\par{\textbf{G1}} We next prove that $\forall t \in \mathcal{T}:\commitmath{PSI}{t}{e} \Rightarrow \lnot \gone$:	

\par{\textbf{G1-a}} Let us assume that \textit{H} exhibits
		phenomenon \gonea (aborted reads). There
		must exists events $w_i(x_i), r_j(x_i)$ in
		H such that $t_i$ subsequently aborted. $\mathcal{T}$ and any corresponding execution $e$, however, consists only of committed
		transactions. Hence $\forall e: \not\exists s \in \mathcal{S}_e, s.t.\ s \in \RSmath{ r_j(x_i)}$: i.e. $\lnot \prereadmath{e}{t_j}$, therefore $\lnot \prereadmath{e}{\mathcal{T}}$. There thus
		exists a transaction for which the commit test cannot be satisfied, for any e. We have a contradiction.

\par{\textbf {G1-b}} Let us assume that $H$ exhibits 
		phenomenon \goneb (intermediate reads). In an execution $e$, only the final writes of a transaction are applied. Hence,$\forall e: \not\exists s \in \mathcal{S}_e, s.t.\ s \in \RSmath{ r(x_{intermediate})}$, i.e. $\lnot \prereadmath{e}{t}$, therefore $\lnot \prereadmath{e}{\mathcal{T}}$. There
		thus exists a transaction $t$, which for all e, will not 
		satisfy the commit test. We once again have a contradiction.

\par{\textbf {G1-c}} Finally, let us assume that $H$ exhibits phenomenon \gonec: $DSG(H)$ must contain a cycle
		of read/write dependencies. We consider each possible
		edge in the cycle in turn:
		\begin{itemize}
			\item{$t_i\xrightarrow{ww}t_j$} There must exist
			an object $x$ such that $x_i << x_j$ (version order). By construction, version in \textit{H} is consistent with the execution order \textit{e}: we have  $s_{t_i} \xrightarrow{*} s_{t_j}$.
			\item{$t_i\xrightarrow{wr}t_j$} There must exist
			a read $o=r_j(x_i) \in \Sigma_{t_j}$ such that
			$t_j$ reads version $x_i$ written by $t_i$. 
			By assumption, $\commitmath{PSI}{e}{t_j}$ holds.
			By $\prereadmath{e}{\mathcal{T}}$ and Lemma~\ref{sfolemma}, we have $\sfomath{o} \xrightarrow{+} s_{t_j}$; 
			 and since $\sfomath{o}$ exists, $\sfomath{o} = s_{t_i} $.
		 It follows that  $s_{t_i} \xrightarrow{+} s_{t_j}$.
		
	\end{itemize}
If a history $H$ displays phenomena $\gonec$ there must exist a chain of transactions $t_i\rightarrow t_{i+1} \rightarrow ... \rightarrow t_{j}$ such that $i=j$.
A corresponding cycle
must thus exist in the execution $e$: $s_{t_i} \xrightarrow{*} s_{t_{i+1}} \xrightarrow{*} \dots \xrightarrow{*} s_{t_j}$. By definition however, a valid execution must be totally ordered. We once again have a contradiction.  

\par{\textbf \gs} We now prove that $\forall t \in \mathcal{T}:\commitmath{PSI}{t}{e} \Rightarrow \lnot \gs$

By way of contradiction, let us assume that \textit{H} exhibits
	phenomenon $\gs$: $DSG(H)$ must contain a directed cycle with exactly one anti-dependency edge. Let $t_1 \xrightarrow{ww/wr} t_{2} \xrightarrow{ww/wr} \dots \xrightarrow{ww/wr} t_{k} \xrightarrow{rw} t_1$ be the cycle in $DSG(H)$.
	
We first prove by induction that  $t_1 \in \dependsmath{e}{t_k}$,
where $t_k$ denotes the $k-th$ transaction that succedes $t_1$.
We then show that there exist a $t' \in \dependsmath{e}{t_k}$ such that  $ o.k \in \mathcal{W}_{t'} \Rightarrow s_{t'}  \xrightarrow{*} \slomath{o}$ does not hold

\par \textbf{Base case} We prove that $t_1 \in \dependsmath{e}{t_2}$.
We distinguish between two cases $t_1 \xrightarrow{ww} t_2$, and
$t_1 \xrightarrow{wr} t_2$. 
\begin{itemize}
\item If $t_1 \xrightarrow{ww} t_{2}$, there must exist an object
 $k$ that $t_1$ and $t_2$ both write: $k \in \mathcal{W}_{t_1}$ and $k \in \mathcal{W}_{t_2}$, therefore $\mathcal{W}_{t_1} \cap \mathcal{W}_{t_2} \neq \emptyset $. By construction, $t_i\xrightarrow{ww}t_j \Leftrightarrow s_{t_i} \xrightarrow{*} s_{t_j}$. Hence we have $s_{t_1} \xrightarrow{*} s_{t_2}$. By definition of $\ddependsmath{e}{t}$, it follows that $t_1 \in \ddependsmath{e}{t_2}$. 
\item If $t_1 \xrightarrow{wr} t_2$, there must exist an object $k$ such that
$t_2$ reads the version of the object created by transaction $t_1$: $o=r(k_1)$.
We previously proved that $t_i\xrightarrow{wr}t_j \Rightarrow s_{t_i} \xrightarrow{+} s_{t_j}$. It follows that $s_{t_1} \xrightarrow{+} s_{t_2}$ and 
$\sfomath{o} = s_{t_1}$, i.e. $t_1 = t_{\sfomath{o}}$. By definition, $t_1 \in \ddependsmath{e}{t_2}$. 
\end{itemize}
Since $\ddependsmath{e}{t_2} \subseteq \dependsmath{e}{t_2}$, it follows that $t_1 \in \dependsmath{e}{t_2}$.

\par \textbf{Induction step} Assume $t_1 \in \dependsmath{e}{t_i}$, we prove that $t_1 \in \dependsmath{e}{t_{i+1}}$.To do so,
we first prove that $t_i \in  \ddependsmath{e}{t_{i+1}}$. We distinguish between two cases:  $t_i \xrightarrow{ww} t_{i+1}$, and
$t_i \xrightarrow{wr} t_{i+1}$. 
\begin{itemize}
\item If $t_i \xrightarrow{ww} t_{i+1}$, there must exist an object
 $k$ that $t_i$ and $t_{i+1}$ both write: $k \in \mathcal{W}_{t_i}$ and $k \in \mathcal{W}_{t_{i+1}}$, therefore $\mathcal{W}_{t_i} \cap \mathcal{W}_{t_{i+1}} \neq \emptyset $. By construction, $t_i\xrightarrow{ww}t_j \Leftrightarrow s_{t_i} \xrightarrow{*} s_{t_j}$. Hence we have $s_{t_i} \xrightarrow{*} s_{t_{i+1}}$. By definition of $\ddependsmath{e}{t}$, it follows that $t_i \in \ddependsmath{e}{t_{i+1}}$. 
\item If $t_i \xrightarrow{wr} t_{i+1}$, there must exist an object $k$ such that
$t_{i+1}$ reads the version of the object created by transaction $t_i$: $o=r(k_i)$.
We previously proved that $t_i\xrightarrow{wr}t_j \Rightarrow s_{t_i} \xrightarrow{+} s_{t_j}$. It follows that $s_{t_i} \xrightarrow{+} s_{t_{i+1}}$ and 
$\sfomath{o} = s_{t_i}$, i.e. $t_i = t_{\sfomath{o}}$. By definition, $t_i \in \ddependsmath{e}{t_{i+1}}$. 
\end{itemize}
Hence, $t_{i} \in \ddependsmath{e}{t_{i+1}}$. The depends set includes
the depend set of every transaction that it directly depends on:
consequently $\dependsmath{e}{t_{i}} \subseteq \dependsmath{e}{t_{i+1}}$.
We conclude: $t_1 \in \dependsmath{e}{t_{i+1}}$.
	
Combining the base step and the induction step, we have proved that $t_1 \in \dependsmath{e}{t_k}$.

We now derive a contradiction. Consider the edge $t_k \xrightarrow{rw} t_1$ in the 
$\gs$ cycle: $t_k$ reads the version of an object $x$ that precedes the version written
by $t_1$. Specifically, there exists a version $x_m$ written by
transaction $t_m$ such that $r_k(x_m) \in \Sigma_{t_k}$, $w_1(x_1) \in \Sigma_{t_1}$ and $x_m << x_1$.
By definition of the PSI commit test for transaction $t_k$, if $t_1 \in \dependsmath{e}{t_k}$ and $t_1$'s write set intersect with $t_k$'s read
set, then $s_{t_1} \xrightarrow{*} \slomath{r_k(x_m)}$.  

	However, from $x_m << x_1$, we have $\forall s, s', s.t. (x,x_m) \in s \wedge (x,x_1) \in s' \Rightarrow s \xrightarrow{+} s'$. Since $(x,x_m) \in \slomath{r_k(x_m)} \wedge (x,x_1) \in s_{t_1}$, we have $\slomath{r_k(x_m)} \xrightarrow{+} s_{t_1}$. We previously proved that T $s_{t_1} \xrightarrow{*} \slomath{r_k(x_m)}$. We have a contradiction: $H$ does not exhibit
	phenomenon \gs, i.e. $\exists e:\forall t \in \mathcal{T}:\commitmath{PSI}{t}{e} \Rightarrow \lnot G1 \land \lnot \gs$.

($\Leftarrow$) We now prove the other direction $\lnot \gone \land \lnot \gs \Rightarrow \exists e:\forall t \in \mathcal{T}:\commitmath{PSI}{t}{e}$. 

We construct $e$ as follows: Consider only dependency edges in the DSG(H), by $\lnot G1$,  there exist no cycle consisting of only dependency edges, therefore the transactions can be topologically sorted respecting only dependency edges. Let $i_1,...i_n$ be a permutation of $1,2,...,n$ such that $t_{i_1},...,t_{i_n}$ is a topological sort of DSG(H) with only dependency edges. We construct an execution $e$ according to the topological order defined above: $e: s_0 \rightarrow s_{t_{i_1}}\rightarrow s_{t_{i_2}} \rightarrow ... \rightarrow s_{t_{i_n}}$. 

	First we show that $\prereadmath{e}{\mathcal{T}}$ is true: consider any transaction $t$, for any operation $o \in \Sigma_t$. If $o$ is a internal read operation or $o$ is a write operation, by definition $s_0 \in \RSmath{o}$ hence $\RSmath{o} \neq \emptyset$ follows trivially. 
		Consider the case now where $o$ is a read operation that reads a value written by another transaction $t'$. Since the topological order includes $wr$ edges and $e$ respects the topological order,  $t' \xrightarrow{wr} t$ in $DSG(H)$ implies $s_{t'} \xrightarrow{*} s_{t}$, then for any $o=r(x,x_{t'}) \in \Sigma_{t}$,  $s_{t'} \in \RSmath{o}$, therefore $\RSmath{o} \neq \emptyset$. Hence we have $\prereadmath{e}{t}$ is true. Therefore
        $\prereadmath{e}{\mathcal{T}}$ holds. 
	
	Next, we prove that $\forall o \in \Sigma_t: \forall t' \in \dependsmath{e}{t}: o.k \in \mathcal{W}_{t'} \Rightarrow s_{t'} \xrightarrow{*} \slomath{o}$ holds.
For any $t' \in \dependsmath{e}{t}$, by Lemma~\ref{lemma:predtrans}, $s_{t'} \xrightarrow{+}s_t$. Consider any $o \in \Sigma_t$, let $t'$ be a transaction such that $t' \in \dependsmath{e}{t} \wedge o.k \in \mathcal{W}_{t'}$, we now prove that $s_{t'} \xrightarrow{*} \slomath{o}$.
Consider the three possible types of operations in $t$: 
	\begin{enumerate}
		\item \textit{External Reads}: an operation reads an object version
		that was created by another transaction. 
		\item \textit{Internal Reads}: an operation reads an object version that
		itself created.
		\item \textit{Writes}: an operation creates a new object version. 
	\end{enumerate}
	We show that $s_{t'} \xrightarrow{*} \slomath{o}$ for each of  those operation types:
	\begin{enumerate}
		\item \textit{External Reads}. Let $o=r(x,x_{\hat{t}})\in \Sigma_t$ read the
		version for $x$ created by $\hat{t}$, where $\hat{t} \neq t$. Since $\prereadmath{e}{t}$ is true, we have $\RSmath{o} \neq \emptyset$, therefore $s_{\hat{t}} \xrightarrow{+} s_t$ and $\hat{t}= t_{\sfomath{o}}$. From $\hat{t}= t_{\sfomath{o}}$, we have $\hat{t} \in \ddependsmath{e}{t}$.
Now consider $t'$ and $\hat{t}$, we have currently proved that $s_{t'} \xrightarrow{+}s_t$ and $s_{\hat{t}} \xrightarrow{+} s_t$. There are two
cases: 
		\begin{itemize}
			\item $s_{t'} \xrightarrow{*} s_{\hat{t}}$: Consequently $s_{t'} \xrightarrow{*} s_{\hat{t}} =\sfomath{o} \xrightarrow{*} \slomath{o}$
It follows that $s_{t'} \xrightarrow{*} \slomath{o}$.	
			\item $s_{\hat{t}} \xrightarrow{+} s_{t'}$:
We prove that this cannot happen by contradiction. Since $o.k \in \mathcal{W}_{t'}$, $t'$ also writes key $x_{t'}$. By construction,
, $s_{\hat{t}} \xrightarrow{+} s_{t'}$ in $e$ implies $x_{\hat{t}} << x_{t'}$.
There must consequently exist a chain of $ww$ edges between $\hat{t}$ and $t'$ in $DSG(H)$, where all the transactions on the chain writes a new version of key $x$. Now consider the transaction in the chain directly after to $\hat{t}$, denoted as $\hat{t}_{+1}$, where $\hat{t} \xrightarrow{ww} \hat{t}_{+1} \xrightarrow{ww}^* t'$. $\hat{t}_{+1}$ overwrites the version of $x$ $t$ reads. Consequently, $t$ directly anti-depends on $\hat{t}_{+1}$, i.e. $t \xrightarrow{rw}\hat{t}_{+1}$. 
Moreover $t' \in \dependsmath{e}{t}$, by Lemma~\ref{lemma:predchain},  we have $t' \xrightarrow{ww/wr}^+ t$. There thus exists a cycle consists of only one anti dependency edges as $t \xrightarrow{rw} \hat{t}_{+1} \xrightarrow{ww}^* t' \xrightarrow{ww/wr}^+ t$, in contradiction with \gs. $s_{t'} \xrightarrow{*} s_{\hat{t}}$ holds.
\end{itemize}
$s_{t'} \xrightarrow{*} s_{\hat{t}}$ holds in all cases. Noting that $s_{\hat{t}}=
\slomath{o}$,
we conclude $s_{t'} \xrightarrow{*} \slomath{o}$.
		
		\item \textit{Internal Reads}. Let $o=r(x,x_t)$
		read $x_t$ such that $w(x,x_t)\xrightarrow{to}r(x,x_t)$. By definition of $\RSmath{o}$, we have $\slomath{o} = s_p(t)$. Since we have proved that $s_{t'} \xrightarrow{+} s_t$, therefore we have $s_{t'} \xrightarrow{*} s_p(t) = \slomath{o}$ (as $s_p(t) \rightarrow s_t$).
		
		\item \textit{Writes}.  Let $o=w(x,x_t)$ be a write operation. By definition of $\RSmath{o}$, we have $\slomath{o} = s_p(t)$. We
previously proved that $s_{t'} \xrightarrow{+} s_t$. Consequently we have $s_{t'} \xrightarrow{*} s_p(t) = \slomath{o}$ (as $s_p(t) \rightarrow s_t$).
\end{enumerate}

We conclude, in all cases, $\commitmath{PSI}{t}{e} \equiv \prereadmath{e}{t} \wedge \forall o \in \Sigma_t: \forall t' \in \dependsmath{e}{t}: o.k \in \mathcal{W}_{t'} \Rightarrow s_{t'} \xrightarrow{*} \slomath{o}$.
\end{proof}
  
\subsection{$PSI_A$}
We now prove the following theorem:  
\par \textbf{Theorem~\ref{theorem:psi}~}Let $\mathcal{I}$ be PSI. Then $\exists e:\forall t \in \mathcal{T}:\commitmath{PSI}{t}{e} \equiv PSI_A$

We note that this axiomatic specification, defined by Cerone et al.~\cite{cerone2015disc,cerone2015framework} is proven to be equivalent to the operational specification of Sovran et al.~\cite{sovran2011walter}, modulo an additional assumption: that each replica executes each transaction sequentially. 
The authors state that this is for syntactic elegance only, and does not change the essence of the proof. 

\subsubsection{Model Summary}

We provide a brief summary and explanation of the main terminology introduced
in Cerone et al.'s framework for reasoning about concurrency. We refer the reader to~\cite{cerone2015framework} for the full set of definitions.

The authors consider a database storing a set of objects $Obj = \{x, y, . . .\}$,
with operations $Op = \{read(x, n), write(x, n) | x \in Obj, n \in \mathbb{Z}\}$.
For simplicity, the authors assume the value space to be $\mathbb{Z}$. 

\begin{definition}
History events are tuples of the form $(\iota, op)$, where $\iota$ is an identifier from a countably infinite set EventId and $op \in Op$. Let $WEvent_x =\{(\iota, write(x, n)) | \iota \in EventId, n \in  \mathbb{Z}\}$, 
$REvent_x =\{(\iota, write(x, n)) | \iota \in EventId, n \in  \mathbb{Z}\}$, and $HEvent_{x} = REvent_{x} \cap WEvent_{x}$.
\end{definition}

\begin{definition}
	A transaction $T$ is a pair $(E, po)$, where $E \subseteq HEvent$ is an
	non-empty set of events with distinct identifiers, and the program order $po$ is a total order over $E$. A history $\mathcal{H}$ is a set of transactions with disjoint sets of event identifiers.
\end{definition} 

\begin{definition}
	An abstract execution is a triple $A = (\mathcal{H}, VIS, AR)$ where:
	visibility $VIS \subseteq \mathcal{H} \times \mathcal{H} $ is an acyclic relation; and
	arbitration $AR \subseteq \mathcal{H} \times \mathcal{H} $ is a total order such that $AR \supseteq VIS$.
\end{definition}

For simplicity, we summarise the model's main notation specificities:
\begin{itemize}
\item{\rule{7pt}{1pt}} Denotes a value that is irrelevant and implicitly existentially quantified.
\item{\boldmath$\max_R(A)$} Given a total order $R$ and a set $A$, $\max_R(A)$ is the element $u \in A$ such that $\forall v \in A. v = u \vee (v, u) \in R$. 
\item{\boldmath{$R_{-1} (u)$}} 
		For a relation $R \subseteq A \times A$ and
 		an element $u \in A$, we let $R_{-1} (u) = \{v | (v, u) \in R\}$. 
 		\item{\boldmath{$T\vdash Write\ x:n$}} $T$ writes to $x$ and
 		the last value written is $n$: $\max_{po}(E \cap WEvent_{x}) = (\_, write(x, n))$.
 		
 		\item{\boldmath$T\vdash Read\ x:n$} $T$ makes an external read from $x$, i.e., one before writing to $x$, and $n$ is the value returned by the first such read: $\min_{po}(E \cap HEvent_{x}) = (\_, read(x, n))$.
\end{itemize}
 
The authors introduce a number of consistency axioms. A consistency model
specification is a set of consistency axioms $\Phi$ constraining executions.
The model allow those histories for which there exists an execution that satisfies
the axioms:
\begin{definition}
$Hist_{\Phi} = \{ \mathcal{H} | \exists Vis,AR.(\mathcal{H}, AR) \vDash \Phi \}$
\end{definition}

The authors define several consistency axioms:
\begin{definition}
	\begin{description}
		 \item[INT] $\forall(E,po) \in \mathcal{H}. \forall event \in E. \forall x,n. (event = (\_,read(x,n))\wedge (po^{-1}(event)\cap \textit{HEvent}_x \neq \emptyset))$\\ $\Rightarrow \max_{po}(po^{-1}(event)\cap \textit{HEvent}_x)=(\_,\_(x,n))$
		 
		 \item[EXT] 	$\forall T \in \mathcal{H}.\forall x,n. T \vdash Read\ x:n \Rightarrow$
		 $((VIS^{-1}(T)\cap\{S| S \vdash Write\ x:\_\} = \emptyset \wedge n = 0) \vee$\\
		 $ \max_{AR}((VIS^{-1}(T) \cap\{S| S \vdash Write\ x:\_\}) \vdash Write\ x:n)$
		 
		 \item[TRANSVIS] VIS is transitive
		 \item[NOCONFLICT] $\forall T,S \in \mathcal{H}. (T \neq S \wedge T \vdash Write x:\_ \wedge S \vdash Write x: \_) \Rightarrow (T\xrightarrow{VIS} S \vee S \xrightarrow{VIS} T)$
	\end{description}
\end{definition}

$PSI_A$ is then defined with the following set of consistency axioms. 
\begin{definition}
	PSI allows histories for which there exists an execution that satisfies INT, EXT, TRANSVIS and NOCONFLICT:
	$Hist_{PSI} = \{H | \exists VIS, AR.(\mathcal{H}, VIS, AR) \models$ {INT, EXT, TRANSVIS, NOCONFLICT}\}.
\end{definition}

\subsubsection{Equivalence}

We first relate Cerone et al.'s notion of transactions to transactions in our model:
Cerone defines transactions as a tuple $(E,po)$ where $E$ is a set of events and $po$ is a program order over $E$. Our model similalry defines transactions as a tuple 
$(\Sigma_t, \xrightarrow{to})$, where $\Sigma_t$ is a set of operations, and $\xrightarrow{to}$ is the total order on $\Sigma_t$. These definitions are equivalent: events defined in Cerone are extensions of operations in our model (events include a unique identifier), while the partial order in Cerone maps to the program order in our model. For clarity, we denote transactions in Cerone's model as $T$ and transactions in our model as $t$. Finally, we relate our notion of
versions to Cerone's values.  

($\Rightarrow$) \textbf{We first prove} $\exists e:\forall t \in \mathcal{T}:\commitmath{PSI}{t}{e} \Rightarrow PSI_A$.

\par\textbf {Construction} Let $e$ be an execution such that $\forall t \in \mathcal{T}:\commitmath{PSI}{t}{e}$. We construct $AR$ and $VIS$ as follows:
$AR$ is defined as $T_i \xrightarrow{AR} T_j \Leftrightarrow s_{t_i} \xrightarrow{} s_{t_j}$ while $VIS$ order is defined as $T_i \xrightarrow{VIS} T_j \Leftrightarrow t_i \in \dependsmath{e}{t_j}$.
By definition, our execution is a total order, hence our constructed $AR$ is also a total order. $VIS$ defines an acyclic partial order that is a subset of $AR$ ( by $\prereadmath{e}{\mathcal{T}}$ and Lemma~\ref{lemma:predtrans}). 

We now prove that each consistency axiom holds:
\par\textbf {INT}	 $\forall(E,po) \in \mathcal{H}. \forall event \in E. \forall x,n. (event = (\_,read(x,n))\wedge (po^{-1}(event)\cap \textit{HEvent}_x \neq \emptyset))$\\ $\Rightarrow \max_{po}(po^{-1}(event)\cap \textit{HEvent}_x)=(\_,\_(x,n))$ Intuitively, the consistency axiom $INT$ ensures that the read of
an object returns the value of the transaction's last write to that object (if it exists). 
	 
	 For any $(E,po) \in \mathcal{H}$, we consider any $event$ and $x$ such that $(event = (\_,read(x,n))\wedge (po^{-1}(event)\cap \textit{HEvent}_x \neq \emptyset))$. We prove that $ \max_{po}(po^{-1}(event)\cap \textit{HEvent}_x)=(\_ ,\_ (x,n))$.
By assumption, $(po^{-1}(event)\cap \textit{HEvent}_x \neq \emptyset))$ holds, 
there must exist an event such that  $\max_{po}(po^{-1}(event)\cap \textit{HEvent}_x)$. This event is either a read operation, or a write operation:

\begin{enumerate}
	 	\item If $op = \max_{po}(po^{-1}(event)\cap \textit{HEvent}_x)$ is a write operation: given $event = (\_,read(x,n))$ and $op \xrightarrow{po} event$,
the equivalent statement in our model is $w(x,v_{op}) \xrightarrow{to} r(x,n)$.
By definition, our model enforces that $w(k,v') \xrightarrow{to} r(k,v) \Rightarrow v = v'$. Hence $v_{op} = n$, i.e. $op=(\_ , write(x,n))$, therefore $op = (\_ , \_ (x,n)$. Hence $INT$ holds. 

\item If $op=\max_{po}(po^{-1}(event)\cap \textit{HEvent}_x)$ is a read operation, We write $op = (\_,read(x,v_{op}))$. 
The equivalent formulation in our model is as follows. 
For $event = (\_,read(x,n))$, we write  $ o_1 = r(x,n)$,
and for $op$, we write $o_2 = r(x,v_{op})$ with $o_2 \xrightarrow{to} o_1$
where $o_1, o_2 \in \Sigma_t$. Now we consider the following two cases.

First, let us assume that there exists an operation $w(k,v)$ such that
$w(k,v) \xrightarrow{to} o_2 \xrightarrow{to} o_1$ (all three operations
belong to the same transaction). Given that $ \xrightarrow{to}$ is a total order, we have $w(k,v) \xrightarrow{to} o_1$ and $w(k,v) \xrightarrow{to} o_2$. It follows
by definition of candidate read state that $w(k,v') \xrightarrow{to} r(k,v) \Rightarrow v = v'$, where $v = n \land v = v_{op}$, i.e. $v_{op} = n$.
Hence $op = (\_ , \_ (x,n)$ and $INT$ holds.
Second, let us next assume that there does not exist an operation $w(k,v) \xrightarrow{to} o_2 \xrightarrow{to} o_1$. We prove by contradiction
that $v_{op} = n$ nonetheless. Assume that $v_{op} \neq n$,
and consider transactions $t_1$ that writes $(x,n)$, and $t_2$ that writes $(x,v_{op})$, by $\prereadmath{e}{\mathcal{T}}$ , we know that $\sfomath{o_1}, \sfomath{o_2}$ exist. We have $t_1 = t_{\sfomath{o_1}}$ and $t_2 = t_{\sfomath{o_2}}$. By definition of $\dependsmath{e}{t}$, we have $t_1,t_2 \in \ddependsmath{e}{t} \subseteq \dependsmath{e}{t}$, i.e. $t_1,t_2 \in \dependsmath{e}{t}$.
	 	
	 	We note that the sequence of states containing $(x,n)$ is disjoint from states containing $(x,v_{op})$: in otherwords, the sequence of states bounded
by $\sfomath{o_1}$ and $\slomath{o_1}$ and  $\sfomath{o_2}$ and $\slomath{o_2}$
are disjoint. Hence, we have either $s_{t_1} \xrightarrow{*} \slomath{o_1} \xrightarrow{+} s_{t_2} \xrightarrow{*} \slomath{o_2}$, or $s_{t_2} \xrightarrow{*} \slomath{o_2} \xrightarrow{+} s_{t_1} \xrightarrow{*} \slomath{o_1}$. Equivalently
either $t_2 \in \dependsmath{e}{t} \wedge o_1.k \in \mathcal{W}_{t_2} \wedge \slomath{o_1} \xrightarrow{+} s_{t_2}$, or $t_1 \in \dependsmath{e}{t} \wedge o_2.k \in \mathcal{W}_{t_1} \wedge \slomath{o_2} \xrightarrow{+} s_{t_1}$. In
both cases, this violates $\commitmath{PSI}{t}{e}$,  a contradiction.
We conclude  $v_{op} = n$, i.e. $op=(\_ , read(x,n))$, therefore $op = (\_ , \_ (x,n)$. 
 \end{enumerate}
We proved that $ \max_{po}(po^{-1}(event)\cap \textit{HEvent}_x)=(\_ ,\_ (x,n))$.
$INT$ holds.

\par{\textbf{EXT}} We now prove that $EXT$ holds for $\mathcal{H}$. Specifically, \\	$\forall T \in \mathcal{H}.\forall x,n. T \vdash Read\ x:n \Rightarrow$
	$((VIS^{-1}(T)\cap\{S| S \vdash Write\ x:\_\} = \emptyset \wedge n = 0) \vee$\\
	$ \max_{AR}((VIS^{-1}(T) \cap\{S| S \vdash Write\ x:\_\}) \vdash Write\ x:n)$

We proceed in two steps, we first show that there exist a transaction
$T$ that wrote (x,n), and next we show that $T$ is the most recent such
transaction. 
Consider any $T \in \mathcal{H}.\forall x,n. T \vdash Read\ x:n $ (a external read). Equivalently,
we consider a transaction $t$ in our model such that $r(x,n) \in \Sigma_t$.
Let $t_n$ be the transaction that writes $(x,n)$.
By assumption, $\prereadmath{e}{\mathcal{T}}$
holds hence $\sfomath{o}$ exists and $\sfomath{o} = s_{t_n}$, i.e. $t_n = t_{\sfomath{o}}$, as $t_n$ created the first state from which $o$ could
read from. By definition of $\dependsmath{e}{t}$, we have $t_n \in \ddependsmath{e}{t} \subseteq \dependsmath{e}{t}$, i.e. $t_n \in \dependsmath{e}{t}$. 
Moreover, we defined $VIS$ as $T_i \xrightarrow{VIS} T_j \Leftrightarrow t_i \in \dependsmath{e}{t_j}$. Hence, we have $T_n \xrightarrow{VIS} T$, 
and consequently $T_n \in VIS^{-1}(T)$. Since $write(x,n) \in \Sigma_{t_n}$, $T_n \vdash Write\ x:n$.

Next, we show that $T_n$ is larger than any other transaction $T'$ in $AR$: $T' \xrightarrow{VIS} T \wedge T' \vdash Write\ x:\_$. Consider the equivalent transaction $t'$ in our model, we know that $t' \in \dependsmath{e}{t}$ ($T' \xrightarrow{VIS}T$) and $x \in \mathcal{W}_{t'}$. As $o=r(x,n) \in \Sigma_t$ and $t' \in \dependsmath{e}{t} \wedge o.k \in \mathcal{W}_{t'} $, $\commitmath{PSI}{t}{e}$ implies that $s_{t'} \xrightarrow{*} \slomath{o}$.
We note that the sequence of states containing $(x,n)$ is disjoint from states containing $(x,x_{t'})$. It follows that  
$s_{t'} \xrightarrow{*} \sfomath{o}=s_{t_n}$. We can
strengthen this to say  $s_{t'} \xrightarrow{+} \sfomath{o}=s_{t_n}$ as $t'\neq t_n$. By construction, we have $T' \xrightarrow{AR} T_n$, i.e. $T_n = \max_{AR}((VIS^{-1}(T) \cap\{S| S \vdash Write\ x:\_\})$.
We conclude, $EXT$ holds.  

\par\textbf {TRANSVIS} We first prove that $t_i \in \dependsmath{e}{t_j} \wedge t_j \in \dependsmath{e}{t_k} \Rightarrow t_i \in \dependsmath{e}{t_k}$ and use this
result to prove that $VIS$ is transitive.

We proceed by induction, let $e$ be $s_0 \rightarrow s_{t_1} \rightarrow s_{t_2} \rightarrow \dots \rightarrow s_{t_n}$:	
	
\par\textbf {Base Case} Consider the first transaction $t_1$ in the execution. We want to prove that for all transactions $t$ that precede $t_1$ in the execution
		$\forall t' \in \dependsmath{e}{t_1}: \dependsmath{e}{t'} \subseteq \dependsmath{e}{t_1}$ As $t_1$ is the first transaction in the execution, 
		$\ddependsmath{e}{t_1}=\emptyset$ and consequently $\dependsmath{e}{t} = \emptyset$. We see this by contradiction: assume there exists a transaction $t \in \ddependsmath{e}{t_1}$, by implication $s_t \xrightarrow{+} s_{t_1}$ (Lemma~\ref{lemma:preddirect}), violating
		our assumption that $t_1$ is the first transaction in the execution. Hence the desired result 
		trivially holds. 
\par\textbf {Induction Step} Consider the i-th transaction in the execution. We assume that $\forall t \text{ s.t. } s_t \xrightarrow{*} s_i$ the property $\forall t' \in \dependsmath{e}{t}: \dependsmath{e}{t'} \subseteq \dependsmath{e}{t}$ holds. In other words, we assume that the property holds
		for the first \textit{i} transactions. We now prove that the property holds for
		the first \textit{i+1} transactions, specifically, we show that $ \forall t' \in \dependsmath{e}{t_{i+1}}: \dependsmath{e}{t'} \subseteq \dependsmath{e}{t_{i+1}}$. A transaction $t'$ belongs to $\dependsmath{e}{t_{i+1}}$ if one of two conditions holds: either 
		$t' \in \ddependsmath{e}{t_{i+1}}$, or $ \exists t_k \in \mathcal{T}:t' \in \dependsmath{e}{t_k} \wedge t_k \in \ddependsmath{e}{t_{i+1}}$. We consider each in turn:
		
		\begin{itemize}
			\item If $t' \in \ddependsmath{e}{t_{i+1}}$: by definition of $\dependsmath{e}{t_{i+1}}$, $\dependsmath{e}{t'} \subseteq \dependsmath{e}{t_{i+1}}$.
			\item If $ \exists t_k \in \ddependsmath{e}{t_{i+1}}: t' \in \dependsmath{e}{t_k}$: As $t_k \in \ddependsmath{e}{t_{i+1}}$, by definition of $\dependsmath{e}{t_{i+1}}$, $\dependsmath{e}{t_k} \subseteq \dependsmath{e}{t_{i+1}}$.
			Moreover, by Lemma~\ref{lemma:preddirect}, we have $s_{t_k}  \xrightarrow{+} s_{t_{i+1}}$, i.e. $s_{t_k} \xrightarrow{*} s_{t_i}$ ($s_{t_i}$ directly precedes $s_{t_{i+1}}$ in $e$ by construction).
			The induction hypothesis holds for every transaction that strictly precedes $t_{i+1}$ in $e$,
			hence $\forall t_{k'} \in \dependsmath{e}{t_k}: \dependsmath{e}{t_{k'}} \subseteq   \dependsmath{e}{t_k}$.
			As $t' \in \dependsmath{e}{t_k}$ by construction, it follows that $\dependsmath{e}{t'} \subseteq   \dependsmath{e}{t_k}$. Putting everything together, we have $\dependsmath{e}{t'} \subseteq   \dependsmath{e}{t_k} \subseteq \dependsmath{e}{t_{i+1}}$, and consequently $\dependsmath{e}{t'} \subseteq \dependsmath{e}{t_{i+1}}$. This completes the induction step of the proof.
		\end{itemize}
		
		Combining the base case, and induction step, we conclude:  $\forall t' \in \dependsmath{e}{t}: \dependsmath{e}{t'} \subseteq \dependsmath{e}{t}$.

If $T_i \xrightarrow{VIS} T_j \wedge T_j \xrightarrow{VIS} T_k$, by
construction we have $t_i \in \dependsmath{e}{t_j} \wedge t_j \in \dependsmath{e}{t_k}$. 
From $t_j \in \dependsmath{e}{t_k}$, we know, by induction, that $\dependsmath{e}{t_j} \subseteq \dependsmath{e}{t_k}$, and
consequently that \\ $t_i \in \dependsmath{e}{t_k}$. By construction, we have $T_i \xrightarrow{VIS} T_k$., hence we conclude: $VIS$ is transitive.

\par\textbf{NOCONFLICT}
	$\forall T,S \in \mathcal{H}. (T \neq S \wedge T \vdash Write\ x:\_ \wedge S \vdash Write x: \_) \Rightarrow (T\xrightarrow{VIS} S \vee S \xrightarrow{VIS} T)$
	
	 Consider any $T,S \in \mathcal{H}. (T \neq S \wedge T \vdash Write\ x:\_ \wedge S \vdash Write x: \_) $ and let $t_i, t_j$ be the equivalent transactions
in our model such that $w(x, x_i) \in \Sigma_{t_i}$ and $w(x, x_j) \in \Sigma_{t_j} $ and consequently $x \in \mathcal{W}_{t_i} \cap \mathcal{W}_{t_j}$.  
Since $e$ totally orders all the committed transactions, we have either $s_{t_i} \xrightarrow{+} s_{t_j}$ or $s_{t_j} \xrightarrow{+} s_{t_i}$. 	
	If $s_{t_i} \xrightarrow{+} s_{t_j}$,
	it follows from $s_{t_i} \xrightarrow{+} s_{t_j} \wedge \mathcal{W}_{t_i} \cap \mathcal{W}_{t_j} \neq \emptyset$ that $t_i \in \ddependsmath{e}{t_j} \subseteq \dependsmath{e}{t_j}$, i.e. $t_i \in  \dependsmath{e}{t_j}$, and
    consequently $T \xrightarrow{VIS} S$.
	
	Similarly, if $s_{t_j} \xrightarrow{+} s_{t_i}$, 
	it follows from $s_{t_j} \xrightarrow{+} s_{t_i} \wedge \mathcal{W}_{t_i} \cap \mathcal{W}_{t_j} \neq \emptyset$ that $t_j \in \ddependsmath{e}{t_i} \subseteq \dependsmath{e}{t_i}$, i.e. $t_j \in  \dependsmath{e}{t_i}$, and consequently $S \xrightarrow{VIS} T$.

	We conclude: $T\xrightarrow{VIS} S \vee S \xrightarrow{VIS} T$, NOCONFLICT is true.

($\Leftarrow$) Now we prove that $ PSI_A  \Rightarrow \exists e:\forall t \in \mathcal{T}:\commitmath{PSI}{t}{e}$.

By assumption, AR is a total order over $\mathcal{T}$.  We construct an execution $e$ by applying transactions in the same order as AR, i.e. $s_{t_i} \xrightarrow{+} s_{t_j} \Leftrightarrow T_i \xrightarrow{AR} T_j$ and subsequently prove that $e$ satisfies $\forall t \in \mathcal{T}:\commitmath{PSI}{t}{e}$.

\par{\textbf{Preread}} First we show that $\prereadmath{e}{\mathcal{T}}$ is true: consider any transaction $t$, for any operation $o \in \Sigma_t$. If $o$ is a internal read operation or $o$ is a write operation, by definition $\sfomath{o} = s_0$ hence $\sfomath{o}  \xrightarrow{*} s_t$ follows trivially. On
the other hand, consider the case where $o$ is a read operation that reads a value written by another transaction $t'$: let $T$ and $T'$ be the corresponding
transaction in Cerone's model. We have $T \vdash Read\ x:n$ and $T' \vdash Write\ x:n$. Assuming that values are uniquely identifiable, we have $T'= \max_{AR}(VIS^{-1}(T) \cap\{S| S \vdash Write\ x:\_\})$ by EXT, and consequently $T' \in VIS^{-1}(T)$.  As $VIS \subseteq AR$, $T' \xrightarrow{VIS}T$ and
consequently $T' \xrightarrow{AR}T$.
Recall that we apply transactions in the same order as $AR$, hence we have $s_{t'} \xrightarrow{+} s_{t}$. 
Since we have $(x,n) \in s_{t'}$ and $s_{t'} \xrightarrow{+} s_{t}$, it
follows that $s_{t'} \in \RSmath{o}$, hence $\RSmath{o} \neq \emptyset$. 
We conclude: for any transaction $t$, for any operation $o \in \Sigma_t$, $\RSmath{o} \neq \emptyset$, therefore $\prereadmath{e}{\mathcal{T}}$ is true.

Now consider any $t\in \mathcal{T}$, we want to prove that $\forall o \in \Sigma_t: \forall t' \in \dependsmath{e}{t}: o.k \in \mathcal{W}_{t'} \Rightarrow s_{t'} \xrightarrow{*} sl_o$.

First we prove that $\forall t' \in  \dependsmath{e}{t} \Rightarrow T'\xrightarrow{VIS} T$. 

We previously proved that $\prereadmath{e}{\mathcal{T}}$ is true. Hence, by Lemma~\ref{lemma:predchain} we know that there is a chain $t' \xrightarrow{wr/ww}^+ t$. Consider any edge on the chain: $t_i \xrightarrow{ww/wr} t_j$:

\begin{enumerate}
	\item $t_i \xrightarrow{ww} t_j$: We have $T_i,T_j \in \mathcal{H}$ and $(T_i \neq T_j \wedge T_i \vdash Write\ x:\_ \wedge T_j \vdash Write\ x:\_)$, therefore by NOCONFLICT, we have $t_i \xrightarrow{VIS} t_j \vee t_j \xrightarrow{VIS} t_i$. Note that $s_{t_i} \xrightarrow{*} s_{t_j}$, we know that $t_i \xrightarrow{AR} t_j$, and since $VIS \subseteq AR$, we have $t_i \xrightarrow{VIS} t_j$.
	\item $t_i \xrightarrow{wr} t_j$. We map the initial values in Cerone et al
    from $0$ to $\bot$.	Let $n$ be the value that $t_i$ writes and $t_j$ reads.
    A transaction cannot write empty value, i.e. $\bot$, to a key. It follows
    that $T_j \vdash Read\ x:n $ and $n \neq 0$. By EXT, $ \max_{AR}(VIS^{-1}(T_j) \cap\{S| S \vdash Write\ x:\_\}) \vdash Write\ x:n$. 
	Since $T_i \vdash Write \ x:n$, $T_i = \max_{AR}(VIS^{-1}(T_j) \cap\{S| S \vdash Write\ x:\_\})$ hold, and consequently $T_i \in VIS^{-1}(T_j) $, i.e. $T_i \xrightarrow{VIS} T_j$.
\end{enumerate}

Now we consider the chain $t' \xrightarrow{wr/ww}^+ t$, and we have that $T' \xrightarrow{VIS}^+ T$, by TRANSVIS, we have $T' \xrightarrow{VIS} T$.

Now, consider any $o \in \Sigma_t$ such that $o.k \in \mathcal{W}_{t'}$, let $o.k = x$, therefore $T' \vdash Write\ x:\_$. We previously proved that $T' \xrightarrow{VIS} T$. Hence we have $T' \in VIS^{-1}(T) \cap\{S| S \vdash Write\ x:\_\}$. Now we consider the following two cases.

If $o$ is an external read, and reads the value $(x,\hat{x})$ written by $\hat{t}$. As transactions cannot write  an empty value, i.e. $\bot$, to a key, we have $T \vdash Read\ x:\hat{x} $ and $\hat{x} \neq 0$. By EXT, $ \max_{AR}(VIS^{-1}(T_j) \cap\{S| S \vdash Write\ x:\_\}) \vdash Write\ x:\hat{x}$. 
Since $\hat{T}\vdash Write\ x:n$, we have $\hat{T} = \max_{AR}(VIS^{-1}(T_j) \cap\{S| S \vdash Write\ x:\_\})$
, therefore $T' \xrightarrow{AR} \hat{T}$ or $T'  =\hat{T}$. Note that we apply transactions in the same order as $AR$, therefore we have $s_{t'} \xrightarrow{+} s_{\hat{t}}$ or $s_{t'} = s_{\hat{t}}$, i.e. $s_{t'} \xrightarrow{*}s_{\hat{t}}$. 
Since we proved that $\prereadmath{e}{t}$ is true, we have $\sfomath{o}$ exists and $s_{\hat{t}} = \sfomath{o}$, note that by definition $\sfomath{o} \xrightarrow{*} \slomath{o}$. Now we have $s_{t'}\xrightarrow{*} s_{\hat{t}}=\sfomath{o}\xrightarrow{*} \slomath{o}$, therefore $s_{t'} \xrightarrow{*} \slomath{o}$  . 

If $o$ is an internal read operation or write operation, then $\slomath{o} = s_p(t)$. Since $t' \in \dependsmath{e}{t}$, by Lemma~\ref{lemma:predtrans}, we have $s_{t'} \xrightarrow{+} s_{t}$, therefore $s_{t'} \xrightarrow{*} s_p({t}) = \slomath{o}$, i.e. $s_{t'} \xrightarrow{*} \slomath{o}$.

\end{document}